\documentclass[11pt,letterpaper]{article}
\usepackage{setspace}
\usepackage{amsfonts, amsmath, amssymb, amsthm}
\usepackage{versions}
\usepackage[small]{complexity}
\usepackage{stmaryrd}
\usepackage{algpseudocode}
\usepackage{layout}
\usepackage[margin=1in,nohead]{geometry}
\usepackage[pdftex,pagebackref=true,colorlinks]{hyperref} 
\usepackage{algorithm}

\excludeversion{miforbes}
\excludeversion{amir-notes}

\DeclareMathOperator{\rank}{rank}
\DeclareMathOperator{\trace}{Trace}

\DeclareMathOperator{\supp}{Supp}
\DeclareMathOperator{\sr}{SR}
\DeclareMathOperator{\lne}{LNE}
\DeclareMathOperator{\sspan}{Span}

\renewcommand{\O}{\mathcal{O}}
\renewcommand{\C}{\mathcal{C}}
\newcommand{\llb}{\llbracket}
\newcommand{\rrb}{\rrbracket}
\newcommand{\la}{\langle}
\newcommand{\ra}{\rangle}

\begin{document}




\newtheorem{THM}{Theorem}
\newtheorem*{theorem*}{Theorem}
\newtheorem{theorem}{Theorem}[section]
\newtheorem{corollary}[theorem]{Corollary}
\newtheorem{conjecture}[theorem]{Conjecture}
\newtheorem{lemma}[theorem]{Lemma}
\newtheorem*{lemma*}{Lemma}
\newtheorem{observation}[theorem]{Observation}
\newtheorem{construction}[theorem]{Construction}
\newtheorem{proposition}[theorem]{Proposition}
\newtheorem{definition}[theorem]{Definition}
\newtheorem{claim}[theorem]{Claim}
\newtheorem{fact}[theorem]{Fact}
\newtheorem{assumption}[theorem]{Assumption}
\newtheorem{notation}[theorem]{Notation}
\theoremstyle{remark}
\newtheorem{remark}[theorem]{Remark}





\newcommand{\enc}{{\sf Enc}}
\newcommand{\dec}{{\sf Dec}}
\renewcommand{\E}{{\rm Exp}}
\newcommand{\Var}{{\rm Var}}
\newcommand{\Z}{{\mathbb Z}}
\newcommand{\F}{{\mathbb F}}
\renewcommand{\K}{{\mathbb K}}
\renewcommand{\H}{{\mathbb H}}
\newcommand{\N}{{\mathbb N}}
\newcommand{\integers}{{\mathbb Z}^{\geq 0}}
\renewcommand{\R}{{\mathbb R}}
\newcommand{\Q}{{\cal Q}}
\newcommand{\calA}{{\cal A}}
\newcommand{\eqdef}{{\stackrel{\rm def}{=}}}
\newcommand{\from}{{\leftarrow}}
\newcommand{\vol}{{\rm Vol}}
\newcommand{\ip}[1]{{\langle #1 \rangle}}
\newcommand{\wt}{{\rm {wt}}}
\renewcommand{\vec}[1]{{\mathbf #1}}
\newcommand{\mspan}{{\rm span}}
\newcommand{\rs}{{\rm RS}}
\newcommand{\RM}{{\rm RM}}
\newcommand{\Had}{{\rm Had}}
\newcommand{\calc}{{\cal C}}
\newcommand{\coeff}{\mathfrak{C}}

\title{On Identity Testing of Tensors, Low-rank Recovery and Compressed Sensing}
\author{%
Michael A.\ Forbes\thanks{Email: \texttt{miforbes@mit.edu}, Department of Electrical Engineering and
Computer Science, MIT CSAIL, 32 Vassar St., Cambridge, MA 02139, Supported by NSF grant 6919791, MIT
CSAIL and a Siebel Scholarship.}
		\and
Amir Shpilka\thanks{Faculty of Computer Science, Technion --- Israel Institute of Technology, Haifa,
Israel, \texttt{shpilka@cs.technion.ac.il}.  The research leading to these results has received
funding from the European Community's Seventh Framework Programme (FP7/2007-2013) under grant
agreement number 257575.}}

\maketitle

\begin{abstract}

We study the problem of obtaining efficient, deterministic, \textit{black-box polynomial identity
testing algorithms} for depth-3 set-multilinear circuits (over arbitrary fields).  This class of
circuits has an efficient, deterministic, white-box polynomial identity testing algorithm (due to
Raz and Shpilka~\cite{RazShpilka05}), but has no known such black-box algorithm.  We recast this
problem as a question of finding a low-dimensional subspace $\mathcal{H}$, spanned by rank 1
tensors, such that any non-zero tensor in the dual space $\ker(\mathcal{H})$ has high rank.  We
obtain explicit constructions of essentially optimal-size hitting sets for tensors of degree 2
(matrices), and obtain quasi-polynomial sized hitting sets for arbitrary tensors (but this second
hitting set is less explicit).

We also show connections to the task of performing \textit{low-rank recovery} of matrices, which is
studied in the field of compressed sensing.  Low-rank recovery asks (say, over $\mathbb{R}$) to
recover a matrix $M$ from few measurements, under the promise that $M$ is rank $\le r$.  In this
work, we restrict our attention to recovering matrices that are exactly rank $\le r$ using
deterministic, non-adaptive, linear measurements, that are free from noise. Over $\mathbb{R}$, we
provide a set (of size $4nr$) of such measurements, from which $M$ can be recovered in
$\O(rn^2+r^3n)$ field operations, and the number of measurements is essentially optimal.  Further, the
measurements can be taken to be all rank-1 matrices, or all sparse matrices.  To the best of our
knowledge no explicit constructions with those properties were known prior to this work.

We also give a more formal connection between low-rank recovery and the task of \textit{sparse (vector)
recovery}: any sparse-recovery algorithm that exactly recovers vectors of length $n$ and sparsity
$2r$, using $m$ non-adaptive measurements, yields a low-rank recovery scheme for exactly recovering
$n\times n$ matrices of rank $\le r$, making $2nm$ non-adaptive measurements.  Furthermore, if the
sparse-recovery algorithm runs in time $\tau$, then the low-rank recovery algorithm runs in time
$\O(rn^2+n\tau)$.  We obtain this reduction using linear-algebraic techniques, and not using convex
optimization, which is more commonly seen in compressed sensing algorithms.

Finally, we also make a connection to \textit{rank-metric codes}, as studied in coding theory.
These are codes with codewords consisting of matrices (or tensors) where the distance of matrices
$A$ and $B$ is $\rank(A-B)$, as opposed to the usual hamming metric.  We obtain essentially
optimal-rate codes over matrices, and provide an efficient decoding algorithm.  We obtain codes over tensors as
well, with poorer rate, but still with efficient decoding.

\end{abstract}

\newpage

\thispagestyle{empty}

\tableofcontents

\thispagestyle{empty}

\newpage

\setcounter{page}{1}

\section{Introduction}

We start with a motivating example.  Let $\vec{x}$ and $\vec{y}$ be vectors of $n$ variables each.
Let $M$ be an $n\times n$ matrix (over some field, say $\mathbb{R}$), and define the quadratic form
\[f_M(\vec{x},\vec{y})\eqdef\vec{x}^\dagger M\vec{y}\;.\] Suppose now that we are given an oracle to
$f_M$, that can evaluate $f_M$ on inputs $(\vec{x},\vec{y})$ that we supply.  The type of question
we consider is: how many (deterministically chosen) evaluations of $f_M$ must we make in order to
determine whether $A$ is non-zero?

It is not hard to show that $n^2$ evaluations to $f_M$ are necessary and sufficient to determine
whether $A$ is non-zero.  The question becomes more interesting when we are promised that
$\rank(M)\le r$.  That is, given that $\rank(M)\le r$, can we (deterministically) determine whether
$M=0$ using $\ll n^2$ evaluations of $f_M$?  It is not hard to show that there (non-explicitly)
\textit{exist} $\approx 2nr$ evaluations to determine whether $M=0$, and one of the new results in this paper is
to give an \textit{explicit} construction of $2nr$ such evaluations (over $\mathbb{R}$).

We also consider various generalizations of this problem.  The first generalization is to move from
matrices (which are in a sense 2 dimensional) to the more general notion of \textit{tensors} (which
are in a sense $d$-dimensional).  That is, a tensor is a map $T:[n]^d\to\F$ and like a matrix we can
define a polynomial \[f_T(x_{1,1},\ldots,x_{1,n},\ldots,x_{d,1},\ldots,x_{d,n})\eqdef
\sum_{i_1,\ldots,i_d\in[n]} T(i_1,\ldots,i_d) \prod_{j=1}^dx_{j,i_j}\;.\]  As with matrices, tensors
have a notion of rank (defined later), and we can ask: given that $\rank(T)\le r$ how many
(deterministically chosen) evaluations of $f_T$ are needed to determine whether $T=0$. As $T=0$ iff
$f_T=0$, we see that this problem is an instance of \textit{polynomial identity testing}, which
asks: given oracle access to a polynomial $f$ that is somehow ``simple'', how many
(deterministically chosen) queries to $f$ are needed to determine whether $f=0$?

The above questions ask whether a certain matrix or tensor is zero.  However, we can also ask for
more, and seek to reconstruct this matrix/tensor fully.  That is, how many (deterministically
chosen) evaluations to $f_M$ are needed to determine $M$?  This question can be seen to be related
to compressed sensing and sparse recovery, where the goal is to reconstruct a ``simple'' object from
``few'' measurements.  In this case, ``simple'' refers to the matrix being low-rank, as opposed to a
vector being sparse. As above, it is not hard to show that there \textit{exist} $\approx 4nr$
evaluations that determine $M$, and this paper gives an \textit{explicit} construction of $4nr$ such
evaluations, as well as an efficient algorithm to reconstruct $M$ from these evaluations.

We will now place this work in a broader context by providing background on polynomial identity
testing, compressed sensing and low-rank recovery, and the theory of rank-metric codes.

\subsection{Polynomial Identity Testing}

Polynomial identity testing (PIT) is the problem of deciding whether a polynomial (specified by an
arithmetic circuit) computes the identically zero polynomial.  The obvious deterministic algorithm
that completely expands the polynomial unfortunately takes exponential time.  This is in contrast to the fact that
there are several (quite simple) randomized algorithms that solve this problem quite efficiently.
Further, some of these randomized algorithms treat the polynomial as a \textit{black-box}, so that
they only use the arithmetic circuit to evaluate the polynomial on chosen points, as opposed to a
\textit{white-box} algorithm which can examine the internal structure of the circuit.  Even in the white-box
model, no efficient deterministic algorithms are known for general circuits.

Understanding the deterministic complexity of PIT has come to be an important problem in theoretical
computer science.  Starting with the work of Kabanets and Impagliazzo~\cite{KabanetsImpagliazzo04},
it has been shown that the existence of efficient deterministic (white-box) algorithms for PIT has a
tight connection with the existence of explicit functions with large circuit complexity.  As proving
lower bounds on circuit complexity is one of the major goals of theoretical computer science, this
has led to much research into PIT.

Stronger connections are known when the deterministic algorithms are black-box.  For, any such
algorithm corresponds to a \textit{hitting set}, which is a set of evaluation points such that any
small arithmetic circuit computing a non-zero polynomial must evaluate to non-zero on at least one
point in the set. Heintz and Schnorr~\cite{HeintzSchnorr80}, as well as Agrawal~\cite{Agrawal05},
showed that any deterministic black-box PIT algorithm very easily yields explicit polynomials that
have large arithmetic circuit complexity.  Moreover, Agrawal and Vinay \cite{AgrawalVinay08} showed
that a deterministic construction of a polynomial size hitting set for arithmetic circuits of
depth-$4$ gives rise to a quasi-polynomial sized hitting set for general arithmetic circuits.  Thus,
the black-box deterministic complexity of PIT becomes interesting even for constant-depth circuits.
However, currently no polynomial size hitting sets are known for general depth-3 circuits.  Much of recent
work on black-box deterministic PIT has identified certain subclasses of circuits for which small
hitting sets can be constructed, and this work fits into that paradigm.  See \cite{SY10} for a
survey of recent results on PIT.

One subclass of depth-3 circuits is the model of \textit{set-multilinear} depth-3 circuits, first
introduced by Nisan and Wigderson \cite{NisanWigderson96}.  Raz and Shpilka~\cite{RazShpilka05} gave
a polynomial-time white-box PIT algorithm for non-commutative arithmetic formulas, which contains
set-multilinear depth-3 circuits as a subclass.  However, no polynomial-time black-box deterministic
PIT algorithm is known for set-multilinear depth-3 circuits.  The best known black-box PIT results
for the class of set-multilinear circuits, with top fan-in $\le r$ and degree $d$, are hitting sets
of size $\min (n^d,\poly((nd)^r))$, where the first part of bound comes from a simple argument (presented
in Lemma~\ref{lem:naivelrr}), and the second part of the bound ignores that we have
set-multilinear polynomials, and simply uses the best known hitting sets for so-called
$\Sigma\Pi\Sigma(k)$ circuits as established by Saxena and Seshadhri~\cite{SaxenaSeshadhri11}.  For
non-constant $d$ and $r$, these bounds are super-polynomial.  Improving the size of these hitting
sets is the primary motivation for this work.

To connect PIT for set-multilinear depth-3 circuits with the above questions on matrices and
tensors, we now note that any such circuit of top fan-in $\le r$, degree $d$, on $dn$ variables (and thus
size $\le dnr$), computes a polynomial $f_T$, where $T$ is an $[n]^d$ tensor of rank $\le r$.
Conversely, any such $f_T$ can be computed by such a circuit.  Thus, constructing better hitting
sets for this class of circuits is exactly the question of finding smaller sets of
(deterministically chosen) evaluations to $f_T$ to determine whether $T=0$.

\subsection{Low-Rank Recovery and Compressed Sensing}

Low-rank Recovery (LRR) asks (for matrices) to recover an $n\times n$ matrix $M$ from few
\textit{measurements} of $M$.  Here, a measurement is some inner product $\la M,H\ra$, where $H$ is
an $n\times n$ matrix and the inner product $\la \cdot,\cdot\ra$ is the natural inner product on
$n^2$ long vectors.  This can be seen as the natural generalization of the \textit{sparse recovery}
problem, which asks to recover sparse vectors from few linear measurements.  For, over matrices, our
notion of sparsity is simply that of being low-rank.

Sparse recovery and compressed sensing are active areas of research, see for example~\cite{CSweb}.
Much of this area focuses on constructing distributions of measurements such that the unknown sparse
vector can be recovered efficiently, with high probability.  Also, it is often assumed that the
sequence of measurements will not depend on any of the measurement results, and this is known as
\textit{non-adaptive sparse recovery}.  We note that Indyk, Price and Woodruff~\cite{sradaptive}
showed that \textit{adaptive sparse recovery} can outperform non-adaptive measurements in certain
regimes.  Much of the existing work also focuses on efficiency concerns, and various algorithms
coming from convex programming have been used.  As such, these algorithms tend to be stable
under noise, and can recover approximations to the sparse vector (and can even do so only if the
original vector was approximately sparse).  One of the initial achievements in this field is an
efficient algorithm for recovery of a $k$-sparse\footnote{A vector is $k$-sparse if it has at most
$k$ non-zero entries.} approximation of $n$-entry vector in $\O(k\log(n/k))$
measurements~\cite{crt06}.

Analogous questions for low-rank recovery have also been explored (for example, see~\cite{lrr} and
references there in).  Initial work (such as~\cite{ct09,cp09}) asked the question of low-rank
\textit{matrix completion}, where entries of a low-rank matrix $M$ are revealed individually (as
opposed measuring linear combinations of matrix entries).  It was shown in these works that for an $n\times n$ rank $\le r$
matrix that $\O(nr\polylog n)$ noisy samples suffice for \textit{nuclear-norm minimization} to
complete the matrix efficiently.  Further works (such as \cite{EldarNP11}) prove that a randomly
chosen set of measurements (with appropriate parameters) gives enough information for low-rank
recovery, other works (such as \cite{CandesPlan11,RechtFP10}) giving explicit conditions on the
measurements that guarantee that the nuclear norm minimization algorithm works, and finally other
works seek alternative algorithms for certain ensembles of measurements (such
as\footnote{Interestingly, \cite{KhajehnejadOH11} use what they call {\em subspace expanders} a
notion that was studied before in a different context in theoretical computer science and
mathematics under the name of {\em dimension expanders} \cite{LubotzkyZelmanov,DvirShpilka08a}.}
\cite{KhajehnejadOH11}).  As in the sparse recovery case, most of these work seek stable algorithms
that can deal with noisy measurements as well as matrices that are only approximately low-rank.
Finally, we note that some applications (such as quantum state tomography) have additional
requirements for their measurements (for example, they should be easy to prepare as quantum states)
and some work has gone into this as well~\cite{gross1,gross2}.

We now make a crucial observation which shows that black-box PIT for the quadratic form $f_M$ is
actually very closely related to low-rank recovery of $M$.  That is, note that
$f_M(\vec{x},\vec{y})=\vec{x}^\dagger M\vec{y}=\la M,\vec{x}^\dagger\vec{y}\ra$. That is, an
evaluation of $f_M$ corresponds to a measurement of $M$, and in particular this measurement is
realized as a rank-1 matrix.  Thus, we see that any low-rank-recovery algorithm that only uses
rank-1 measurement can also determine if $M$ is non-zero, and thus also performs PIT for quadratic
forms.  Conversely, suppose we have a black-box PIT algorithm for rank $\le 2r$ quadratic forms.
Note then that for any $M,N$ with rank $\le r$, $M-N$ has rank $\le 2r$.  Thus, if $M\ne N$ then
$f_{M-N}$ will evaluate to non-zero on some point in the hitting set.  As $f_{M-N}=f_M-f_N$, it
follows that a hitting set for rank $\le 2r$ matrices will distinguish $M$ and $N$.  In particular,
this shows that information-theoretically any hitting set for rank $\le 2r$ matrices is also an LRR
set.  Thus, in addition to constructing hitting sets for the quadratic forms $f_M$, this paper will
also use those hitting sets as LRR sets, and also give efficient LRR algorithms for these
constructions.

\subsection{Rank-Metric Codes}\label{sec:introrankcodes}

Most existing work on LRR has focused on random measurements, whereas the interesting aspect of PIT
is to develop deterministic evaluations of polynomials. As the main motivation for this paper is to
develop new PIT algorithms, we will seek deterministic LRR schemes.  Further, we will
want results that are field independent, and so this work will focus on noiseless measurements (and
matrices that are exactly of rank $\le r$).  In such a setting, LRR constructions are very related
to \textit{rank-metric codes}.  These codes (related to \textit{array codes}), are error-correcting
codes where the messages are matrices (or tensors) and the normal notion of distance (the Hamming
metric) is replaced by the rank metric (that is, the distance of matrices $M$ and $N$ is
$\rank(M-N)$).  Over matrices, these codes were originally introduced independently by Gabidulin,
Delsarte and Roth~\cite{GabidulinKorzhik72,Gabidulin85a,Gabidulin85b,delsarte,Roth91}.  They showed,
using ideas from BCH codes, how to get optimal (that is, meeting an analogue of the Singleton bound)
rank-metric codes over matrices, as well as how to decode these codes efficiently. A later result by
Meshulam~\cite{Meshulam95} constructed rank-metric codes where every codeword is a Hankel matrix.
Roth~\cite{Roth91} also showed how to construct rank-metric codes from \textit{any} hamming-metric
code, but did not provide a decoding algorithm.  Later, Roth~\cite{Roth96} considered rank-metric
codes over tensors and gave decoding algorithms for a constant number of errors. Roth also discussed
analogues to the Gilbert-Varshamov and Singleton bounds in this regime.  This alternate metric is
motivated by \textit{crisscross errors} in data storage scenarios, where corruption can occur in
bursts along a row or column of a matrix (and are thus rank-1 errors).

We now explain how rank-metric codes are related to LRR.  Suppose we have a set of matrices
$\mathcal{H}$ which form a set of (non-adaptive, deterministically chosen) LRR measurements that can
recover rank $\le r$ matrices.  Define the code $\C$ as the set of matrices orthogonal to each
matrix in $\mathcal{H}$.  Thus, $\C$ is a linear code.  Further, given some $M\in\C$ and $E$ such
that $\rank(E)\le r$, it follows that $\mathcal{H}(M+E)=\mathcal{H}E$ (where we abuse notation and
treat $M$ and $E$ as $n^2$-long vectors, and $\mathcal{H}$ as an $|\mathcal{H}|\times n^2$ matrix).
That $\mathcal{H}$ is an LRR set means that $E$ can be recovered from the measurements
$\mathcal{H}E$.  Thus the code $\C$ can correct $r$ errors (and has minimum distance $\ge 2r+1$, by
a standard coding theory argument, as encapsulated in Lemma~\ref{lem:mindist}).  Similarly, given a
rank-metric code $\C$ that can correct up to rank $\le r$ errors, the parity checks of this code
define an LRR scheme. Thus, a small LRR set is equivalent to a rank-metric code with good rate.

The previous subsection showed the tight connection between LRR and PIT.  Via the above paragraph,
we see that hitting sets for quadratic forms are equivalent to rank-metric codes,  when the parity
check constraints are restricted to be rank 1 matrices.

\subsection{Reconstruction of Arithmetic Circuits}

Even more general than the PIT and LRR problems, we can consider the problem of reconstruction of
general arithmetic circuits only given oracle access to the evaluation of that circuit. This is the
arithmetic analog of the problem of learning a function using membership queries.  For more
background on reconstruction of arithmetic circuits we refer the reader to \cite{SY10}.  Just as
with the PIT and LRR connection, PIT for a specific circuit class gives information-theoretic
reconstruction for that circuit class.  As we consider the PIT question for tensors, we can also
consider the reconstruction problem.

The general reconstruction problem for tensors of degree $d$ and rank $r$ was considered before in
the literature \cite{BBV96,BBBKV00,KlivansShpilka06} where learning algorithms were given for any
value of $r$. However, those algorithms are inherently randomized. Also of note is that the
algorithms of \cite{BBBKV00,KlivansShpilka06} output a {\em multiplicity automata}, which in the
context of arithmetic circuits can be thought of as an {\em arithmetic branching program}. In
contrast, the most natural form of the reconstruction question would be to output a degree $d$
tensor.

\subsection{Our Results}

In this subsection we informally summarize our results.  We again stress that our results handle
matrices of exactly rank $\le r$, and we consider non-adaptive, deterministic measurements.  The
culminating result of this work is the connection showing that low-rank recovery reduces to
performing sparse-recovery, and that we can use dual Reed-Solomon codes to instantiate the
sparse-recovery oracle to achieve a low-rank recovery set that only requires rank-1 (or even sparse)
measurements.  We find the fact that we can transform an algorithm for a combinatorial property
(recovering sparse signals) to an algorithm for an algebraic property (recovering low-rank matrices)
quite interesting.

\paragraph{Hitting Sets for Matrices and Tensors} We begin with constructions of hitting sets for
matrices, so as to get black box PIT for quadratic forms.  By improving a construction of
rank-preserving matrices from Gabizon-Raz~\cite{GabizonRaz08}, we are able to show the
following result, which we can then leverage to construct hitting sets.

\begin{theorem*}[Theorem~\ref{thm:bivariate poly evaluation}]
	Let $n \ge r\ge 1$.  Let $\F$ be a ``large'' field, and let $g\in\F$ have ``large''
	multiplicative order.  Let $M$ be an $n\times n$ matrix of rank $\le r$ over $\F$. Let
	$\hat{f}_M(x,y)=\vec{x}^\dagger M\vec{y}$ be the bivariate polynomial defined by the vectors
	$\vec{x}\in\F^n$ and $\vec{y}\in\F^n$ such that\footnote{In this paper, vectors and matrices
	are indexed from zero, so $\vec{x}=(1,x,x^2, \ldots, x^{n-1})^\dagger$.} $(\vec{x})_i=x^i$
	and $(\vec{y})_i=y^i$.

	Then $M$ is non-zero iff one of the univariate polynomials
	$\hat{f}_M(x,x),\hat{f}_M(x,gx),\ldots,\hat{f}_M(x,g^{r-1}x)$ is non-zero.
\end{theorem*}

Intuitively this says that we can test if the quadratic form $f_M$ is zero by testing whether each
of $r$ univariate polynomials are zero.  As these univariate polynomials are of degree $<2n$, it
follows that we can interpolate them fully using $2n$ evaluations.  As such a univariate polynomial
is zero iff all of these evaluations are zero, this yields a $2nr$ sized hitting set.  While this
only works for ``large'' fields, we can combine this with results on simulation of large fields (see
Section~\ref{sec:pit for tensors over small fields}) to derive results over any field with some
loss. This is encapsulated in the next results for black-box PIT, where the log factors are
unnecessary over large fields.

\begin{theorem*}[Corollaries~\ref{cor:hitsmallmatriximproper} and \ref{cor:hitsmallmatrix}]
	Let $n\ge r\ge 1$.  Let $\F$ be any field, then there is a $\poly(n)$-explicit\footnote{A
	$n\times n$ matrix is $t$-explicit if each entry can be (deterministically) computed in $t$
	steps, where field operations are considered unit cost.} hitting set for $n\times n$
	matrices of rank $\le r$, of size $\O(nr\lg^2n)$.
\end{theorem*}

\begin{theorem*}[Corollary~\ref{cor:hitsmalltensor}]
	\sloppy Let $n,r\ge 1$ and $d\ge2$.  Let $\F$ be any  field, then there is a
	$\poly((nd)^d,r^{\lg d})$-explicit hitting set for $[n]^d$ tensors of rank $\le r$, of size
	$\O(dnr^{\lg d}\cdot (d\lg(nd))^d)$.
\end{theorem*}

If $\F$ is large enough then the $\O((d\lg(nd))^d)$ term is unnecessary. In such a situation, this
is a quasi-polynomial sized hitting set, improving on the $\min(n^d,\poly((nd)^r))$ sized hitting
set achievable by invoking the best known results for $\Sigma\Pi\Sigma(k)$
circuits~\cite{SaxenaSeshadhri11}.  However, this hitting set is not as explicit as the construction
of~\cite{SaxenaSeshadhri11} since it takes at least $n^d$ time to compute, as opposed to
$\poly(n,d,r)$.
Nevertheless, although it takes $\poly((nd)^d,r^{\lg d})$ time to construct the set, the fact that
it is of quasi-polynomial size is quite interesting and novel.  Indeed, in general it is not clear
at all how to construct a quasi-polynomial sized hitting set for general circuits (or just for
depth-$3$ circuits), when one is allowed even an $\exp(nd)$ construction time (where $n$ is the
number of variables, and $d$ is the degree of the output polynomial).
We note that this result improves on the two obvious hitting sets seen in Lemmas~\ref{lem:naivelrr}
and \ref{lem:szhit}. The first gives $n^d$ tensors in the hitting set and is
$\polylog(n,d,r)$-explicit while the second gives a set of size $\approx dnr$ while not being
explicit at all.  The above result non-trivially interpolates between these two results.
Finally, we mention that in Remark~\ref{remark:sqrt trick} we explain how one can achieve (roughly)
a $\poly(r(dn)^{\sqrt{d}})$-constructible hitting set of the same size. As this is a somewhat mild
improvement (this is still not the explicitness that we were looking for) we only briefly sketch the
argument.

\paragraph{Low-Rank Recovery} As mentioned in the previous section, black-box PIT results imply LRR
constructions in an information theoretic sense.  Thus, the above hitting sets imply LRR
constructions but the algorithm for recovery is not implied by the above result.  To yield
algorithmic results, we actually establish a stronger claim.  That is, we first show that the above
hitting sets embed a natural sparse-recovery set arising from the dual Reed-Solomon code.  Then we
develop an algorithm that shows that \textit{any} sparse-recovery set gives rise to a
low-rank-recovery set, and that recovery can be performed efficiently given an oracle for sparse
recovery.  This connection (in the context that any error-correcting code in the hamming metric
yields an error-correcting code in the rank-metric) was independently made by Roth~\cite{Roth91}
(see Theorem 3), who did not give a recovery procedure for the resulting LRR scheme.  The next
theorem, which is the main result of the paper, shows this connection is also efficient with respect
to recovery.

\begin{theorem*}[Theorem~\ref{thm:lrr to sparse}]
	Let $n\ge r\ge 1$.  Let $\mathcal{V}$ be a set of (non-adaptive) measurements for
	$2r$-sparse-recovery for $n$-long vectors.  Then there is a $\poly(n)$-explicit set
	$\mathcal{H}$, which is a (non-adaptive) rank $\le r$ low-rank-recovery set for $n\times n$
	matrices, with a recovery algorithm running in time $\O(rn^2+n\tau)$, where $\tau$ is the
	amount of time needed to do sparse-recovery from $\mathcal{V}$.  Further,
	$|\mathcal{H}|=2n|\mathcal{V}|$, and each matrix in $\mathcal{H}$ is $n$-sparse.
\end{theorem*}

This result shows that sparse-recovery and low-rank recovery (at least in the exact case) are very
closely connected.  Interestingly, this shows that sparse-recovery (which can be regarded as a
combinatorial property) and low-rank recovery (which can be regarded as an algebraic property) are
tightly connected.  Many fruitful connections have taken this form, such as in spectral graph
theory, and perhaps the connection presented here will yield yet further results.

Also, the algorithm used in the above result is purely linear-algebraic, in contrast to the convex
optimization approaches that many compressed sensing works use.  However, we do not know if the
above result is stable to noise, and regard this issue as an important question left open by this
work.

When the above result is combined with our hitting set results, we achieve the following LRR scheme
for matrices (and an LRR scheme for tensors, with parameters similar to
Corollary~\ref{cor:hitsmalltensor} mentioned above, and Corollary~\ref{cor:rankmetrictensors}
mentioned below, is derived in Corollary~\ref{cor:smalllrrtensor}).

\begin{theorem*}[Corollary~\ref{cor:smalllrr}]
	Let $n\ge r\ge 1$.  Over any field $\F$, there is an $\poly(n)$-explicit set $\mathcal{H}$,
	of $\O(rn\lg^2n)$ size, such that measurements against $\mathcal{H}$ allow recovery of
	$n\times n$ matrices of rank $\le r$ in time $\poly(n)$. Further, the matrices in
	$\mathcal{H}$ can be chosen to be all rank 1, or all $n$-sparse.
\end{theorem*}

\noindent We note again that over large fields these logarithmic factors are seen to be unneeded.

Some prior work~\cite{GabidulinKorzhik72,Gabidulin85a,Gabidulin85b,delsarte,Roth91} on LRR focused
on finite fields, and as such based their results on BCH codes.  The above result is based on (dual)
Reed-Solomon codes, and as such works over any field (when combined with results allowing simulation
of large fields by small fields).  Other prior work~\cite{RechtFP10} on exact LRR permitted
randomized measurements, while we achieve deterministic measurements.

Further, we are able to do LRR with measurements that are either all $n$-sparse, or all rank-1.  As
Roth~\cite{Roth91} independently observed, the $n$-sparse LRR measurements can arise from any
(hamming-metric) error-correcting code (but he did not provide decoding). Tan, Balzano and
Draper~\cite{TanBalzanoDraper} showed that random $(n\lg n)$-sparse measurements provide essentially
the same low-rank recovery properties as random measurements.  Thus, our results essentially achieve
this deterministically.

We further observe that a specific code (the dual Reed-Solomon code) allows a change of basis for
the measurements, and in this new basis the measurements are all rank 1. Recht et
al.~\cite{RechtFP10} asked whether low-rank recovery was possible when the measurements were rank 1
(or ``factored''), as such measurements could be more practical as they are simpler to generate and
store in memory.  Thus, our construction answers this question in the positive direction, at least
for exact LRR.

\paragraph{Rank-Metric Codes}

Appealing to the connection between LRR and rank-metric codes, we achieve the following
constructions of rank-metric codes.

\begin{theorem*}[Corollary~\ref{cor:rankmetricmatrices}]
	Let $\F$ be any field, $n\ge 1$ and $1\le r\le n/2$.  Then there are $\poly(n)$-explicit
	rank-metric codes with $\poly(n)$-time decoding for up to $r$ errors, with parameters
	$[[n]^2,(n-2r)^2 \cdot \O(\lg^2 n),2r+1]_\F$, and the parity checks on this code can be
	chosen to be all rank-1 matrices, or all $n$-sparse matrices.
\end{theorem*}

Earlier work on rank-metric codes over finite
fields~\cite{GabidulinKorzhik72,Gabidulin85a,Gabidulin85b,delsarte,Roth91} achieved
$[[n]^2,n(n-2r),2r+1]_{\mathbb{F}_q}$ rank-metric codes, with efficient decoding algorithms.  These
are optimal (meeting the analogue of the Singleton bound for rank-metric codes).  However, these
constructions only work over finite fields.  While our code achieves a worse rate, its construction
works over any field, and over infinite fields the $\O(\lg^2n)$ term is unneeded.  Further,
Roth~\cite{Roth91} observed that the resulting $[[n]^2,(n-2r)^2,2r+1]$ code is optimal (see
discussion of his Theorem 3) over algebraically closed fields (which are infinite).

We are also able to give rank-metric codes over tensors, which can correct errors up to rank
$\approx n^{d/\lg d}$ (out of a maximum $n^{d-1}$), while still achieving constant rate. The
rank-metric code arising from the naive low-rank recovery of Lemma~\ref{lem:naivelrr} never achieves
constant rate, and prior work by Roth~\cite{Roth96} only gave decoding against a constant number of
errors.

\begin{theorem*}[Corollary~\ref{cor:rankmetrictensors}]
	\sloppy Let $\F$ be any field, $n,r\ge1$ and $d\ge 2$.  Then there are $\poly((nd)^d,r^{\lg
	d})$-explicit rank-metric codes with $\poly((nd)^d,r^{\lg d})$-time decoding for up to $r$
	errors, with parameters $[[n]^d,n^d-\O(d^2nr^{\lg d}\lg(dn)),2r+1]_\F$.
\end{theorem*}

We note here that our decoding algorithm will return the \textit{entire} tensor, which is of size
$n^d$. Trivially, any algorithm returning the entire tensor must take at least $n^d$ time.  In this
case, the level of explicitness of the code we achieve is reasonable.   However, a more desirable
result would be for the algorithm to return a rank $\le r$ representation of the tensor, and thus
the $n^d$ lower bound would not apply so that one could hope for faster decoding algorithms.
Unfortunately, even for $d=3$ an efficient algorithm to do so would imply $\P=\NP$.  That is, if an
algorithm (even one which is not a rank-metric decoding or low-rank recovery algorithm) could
produce a rank $\le r$ decomposition for any rank $\le r$ tensor, then one could compute tensor-rank
by as it is the minimum $r$ such that the resulting rank $\le r$ decomposition actually computes the
desired tensor (this can be checked in $\poly(n^d)$ time).  However, H{\aa}stad~\cite{Hastad90}
showed that tensor-rank (over finite fields) is \NP-hard for any fixed $d\ge 3$.  It follows that
for any (fixed) $d\ge 3$, if one could recover (even in $\poly(n^d)$-time) a rank $\le r$ tensor
into its rank $\le r$ decomposition, then $\P=\NP$.  Thus, we only discuss recovery of a tensor by
reproducing its entire list of entries, as opposed to its more concise representation.

Finally, we remark that in \cite{Roth96} Roth discussed the question of decoding rank-metric codes
of degree $d=3$, gave decoding algorithms for errors of rank $1$ and $2$, and wrote that ``Since
computing tensor rank is an intractable problem, it is unlikely that we will have an efficient
decoding algorithm $\ldots$  otherwise, we could use the decoder to compute the rank of any tensor.
Hence, if there is any efficient decoding algorithm, then we expect such an algorithm to recover the
error tensor without necessarily obtaining its rank. Such an algorithm, that can handle any
prescribed number of errors, is not yet known.'' Thus, our work gives the first such algorithm for
tensors of degree $d>2$.

\subsection{Proof Overview}\label{sec: proof overview}

In this section we give proof outlines of the results mentioned so far.

\paragraph{Hitting Sets for Matrices} The main idea for our hitting set construction is to reduce
the question of hitting (non-zero) $n\times n$ matrices to a question of hitting (non-zero) $r\times r$
matrices.  Once this reduction is performed, we can then run the naive hitting set of
Lemma~\ref{lem:naivelrr}, which queries all $r^2$ entries.  This can loosely be seen in analogy with
the kernelization process in fixed-parameter tractability, where a problem depending on the input
size, $n$, and some parameter, $k$, can be solved by first reducing to an instance of size $f(k)$,
and then brute-forcing this instance.

To perform this kernelization, we first note that any $n\times n$ matrix $M$ of rank exactly $r$ can
be written as $M=PQ^\dagger$, where $P$ and $Q$ are $n\times r$ matrices of rank exactly $r$.  To
reduce $M$ to an $r\times r$ matrix, it thus suffices to reduce $P$ and $Q$ each to $r\times r$
matrices, denoted $P'$ and $Q'$.  As this reduction must preserve the fact that $M$ is non-zero, we
need that $P'Q'\ne 0$.  We enforce this requirement by insisting that $P'$ and $Q'$ are also rank
exactly $r$, so that $M'=P'Q'$ is also non-zero.

To achieve this rank-preservation, we turn to a lemma of Gabizon-Raz~\cite{GabizonRaz08} (we note
that this lemma has been used before for black-box PIT~\cite{KarninShpilka08,SaxenaSeshadhri11}).
They gave an explicit family of $\O(nr^2)$-many $r\times n$-matrices $\{A_\ell\}_\ell$, such that for any
$P$ and $Q$ of rank exactly $r$, at least one matrix $A_\ell$ from the family is such that
$\rank(A_\ell P)=\rank(A_\ell Q)=r$. Translating this result into our problem, it follows that one of the
$r\times r$ matrices $A_\ell MA_\ell^\dagger$ is full-rank. The $(i,j)$-th entry of $A_\ell
MA_\ell^\dagger$ is $\la M,(A_\ell)_i(A_\ell)_j^\dagger\ra$, where $(A_\ell)_i$ is the $i$-th row of
$A_\ell$. It follows that querying each entry in these $r\times r$ matrices corresponds to a rank 1
measurement of $M$, and thus make up a hitting set.  As there were $\O(nr^2)$ choices of $\ell$ and
$r^2$ choices of $(i,j)$, this gives a $\O(nr^4)$-sized hitting set.

To achieve a smaller hitting set, we use the following sequence of ideas.  First, we observe that in
the above, we can always assume $i=0$.  Loosely, this is because $A_\ell MA_\ell^\dagger$ is always
full-rank, or zero.  Thus, only the first row of $A_\ell MA_\ell^\dagger$ needs to be queried to
determine this.  Second, we improve upon the Gabizon-Raz lemma, and provide an explicit family of
rank-preserving matrices with size $\O(nr)$.  This follows from modifying their construction so the
degree of a certain determinant is smaller.  To ensure that the determinant is a non-zero
polynomial, we show that it has a unique monomial that achieves maximal degree, and that the term
achieving maximal degree has a non-zero coefficient as a Vandermonde determinant (formed from powers
of an element $g$, which has large multiplicative order) is non-zero.  Finally, we observe that the
hitting set constraints can be viewed as a constraints regarding polynomial interpolation.  This
view shows that some of the constraints are linearly-dependent, and thus can be removed.  Each of
the above observations saves a factor of $r$ in the size of the hitting set, and thus produces an
$\O(nr)$-sized hitting set.

\paragraph{Low-Rank Recovery} Having constructed hitting sets, Lemma~\ref{lem:hittolrr} implies that
the same construction yields low-rank-recovery sets.  As this lemma does not provide a recovery
algorithm, we provide one.  To do so, we must first change the basis of our hitting set.  That is,
the hitting set $\mathcal{B}$ yields a set of constraints on a matrix $M$, and we are free to choose
another basis for these constraints, which we call $\mathcal{D}$.  The virtue of this new basis is
that each constraint is non-zero only on some $k$-diagonal (the entries $(i,j)$ such that $i+j=k$).
It turns out that these constraints are the parity checks of a dual Reed-Solomon code with distance
$\Theta(r)$.  This code can be decoded efficiently using what is known as Prony's
method~\cite{prony1795}, which was developed in 1795.  We give an exposition in
Section~\ref{sec:prony}, where we show how to syndrome-decode this code up to half its minimum
distance, counting erasures as half-errors.  Thus, given a $\Theta(r)$-sparse vector (which can be
thought of as errors from the vector $\vec{0}$) these parity checks impose constraints from which
the sparse vector can be recovered.  Put another way, our low-rank-recovery set naturally embeds a
sparse-recovery set along each $k$-diagonal.

Thus, in designing a recovery algorithm for our low-rank recovery set, we do more and show how to
recover from any set of measurements which embed a sparse-recovery set along each $k$-diagonal.  In
terms of error-correcting codes, this shows that any hamming-metric code yields a rank-metric code
over matrices, and that decoding the rank-metric code efficiently reduces to decoding the hamming-metric code.

To perform recovery, we introduce the notion of a matrix being in $(<k)$-upper-echelon form.
Loosely, this says that $M^{(< k)}$, the entries $(i,j)$ of the matrix with $i+j< k$, are in
row-reduced echelon form.  We then show that for any matrix $M$ in $(<k)$-upper-echelon form, the
$k$-diagonal is $2\rank(M)$-sparse.  As an example, suppose $M^{(<k)}$ was entirely zero.  It
follows then that $M$ is in $(<k)$-upper-echelon form.  Further, the rows that have non-zero entries on
the $k$-diagonal of $M$ are then linearly-independent, as they form a triangular system.  It follows
that the $k$-diagonal can only have $\rank(M)$ non-zero entries.  The more general case is slightly
more complicated technically, but not conceptually. Thus, this echelon-form translates the notion of
low-rank into the notion of sparsity.

The algorithm then follows naturally.  We induct on $k$, first putting $M^{(<k)}$ into
$(<k)$-upper-echelon form (using row-reduction), and then invoking a sparse-recovery oracle on the
$k$-diagonal of $M$ to recover it.  This then yields $M^{(\le k)}$, and we increment $k$.  However,
as described so far, the use of the sparse-recovery oracle is adaptive.  We show that the
row-reduction procedure can be understood such that the adaptive use of the sparse-recovery oracle
can be simulated using non-adaptive calls to the oracle.  More specifically, we will apply the
measurements of the sparse-recovery oracle on each $k$-diagonal of $M$ (which may not be sparse),
and show how to compute the measurements of the adaptive algorithm (where the $k$-diagonals are
sparse) from the measurements made. Putting these steps together, this shows that exact non-adaptive
low-rank-recovery reduces to exact non-adaptive sparse-recovery.  Instantiating this claim with our
hitting sets from above gives a concrete low-rank-recovery set, with accompanied recovery algorithm.

\paragraph{Hitting Sets and Low-Rank Recovery for Tensors} The results for matrices naturally
generalize to tensors in the sense that an $\llb n\rrb^{2d}$ tensor can be viewed as an $\llb
n^d\rrb^2$ matrix.  However, we can do better.  Specifically, the hitting set results were done via
\textit{variable reduction}, as encapsulated by Theorem~\ref{thm:bivariate poly evaluation}, which
shows that a rank $\le r$ bivariate polynomial
$f_M(x,y)=(1,x,x^2,\ldots,x^{n-1})M(1,y,y^2,\ldots,y^{n-1})^\dagger$ is zero iff a set of $r$
univariate polynomials are all zero.  Further, the degrees of these polynomials is only twice the
original degree.  As each univariate polynomial can be interpolated using $\O(n)$ measurements, this
yields $\O(nr)$ measurements total.  This motivates the more general idea of treating a degree $d$
tensor as a $d$-variate polynomial, and showing that we can test whether this polynomial is zero by
testing if a collection of $d'$-variate polynomials are zero, for $d'<d$.  Recursing on this
procedure then reduces the $d$-variate case to the univariate case, and the univariate case is
brute-force interpolated.

The recursion scheme we develop for this is to show that a $d$-variate polynomial is zero iff $r$
$d/2$-variate polynomials are zero, and this naturally leads to an $\O(dnr^{\lg d})$-sized hitting
set.  To prove its correctness, we show that the bivariate case (corresponding to matrices) applied
to two groups of variables allows us to reduce to a single group of variables (with an increase
in the number of polynomials to test). Finally, since we saw how to do low-rank recovery for
matrices, and the tensor-case essentially only uses the matrix case, we can also turn this hitting
set procedure into a low-rank recovery algorithm.

\paragraph{Simulation of Large Fields by Small Fields} Most all of the results mentioned require a
field of size $\approx \poly(n^d)$.  When getting results over small fields, we show that, with some
loss, we can simulate such large fields inside the hitting sets.  We break-up each tensor $H$ in the
original hitting set into new tensors $H_i$ such that for any $\F$-tensor $T$, $\la T,H\ra$ can be
reconstructed from the set of values $\{\la T,\tilde{H}_i\ra\}_i$.  To do so, we use the well-known
representation of a extension field $\K$ of $\F$ as a field of matrices over $\F$.  As the entries
of a rank-1 tensor are multiplications of $d$ elements of $\K$, we can expand these multiplications
out as iterated matrix multiplications, which yields $(\dim_\F\K)^{d+1}$ terms to consider, each of
which corresponds to some $\tilde{H}_i$.

\paragraph{Rank-Metric Codes} The above techniques give the existence of low-rank-recovery sets (and
corresponding algorithms) for tensors, over any field.  Via the connections presented in
Section~\ref{sec:introrankcodes}, this readily yields rank-metric codes with corresponding
parameters.

\section{Notation}\label{sec:notation}

We now fix some notation. For a positive integer $n$ we denote $[n]\eqdef\{1,\ldots,n\}$ and $\llb
n\rrb\eqdef\{0,\ldots,n-1\}$.  We use $\binom{S}{k}$ to denote the set of all subsets of $S$ of size
$k$.  Given a set $S$ of integers, we denote $n-S\eqdef \{n-s:s\in S\}$.  All logarithms will be
base 2.  Given a polynomial $f\in\F[x_1,\ldots,x_m]$, $\deg(f)$ will denote the \textit{total
degree} of $f$, and $\deg_{x_i}(f)$ will denote the \textit{individual degree} of $f$ in the
variable $x_i$.  That is, the polynomial $xy$ has total degree 2 and individual degree 1\ in the
variable $x$ and individual degree 0\ in the variable $z$.  Given a monomial $\vec{x}^\vec{\alpha}$,
$\coeff_{\vec{x}^\vec{\alpha}}(f)$ will denote the coefficient of $\vec{x}^\vec{\alpha}$ in the
polynomial $f$.

Vectors, matrices, and tensors will all begin indexing from 0, instead of from 1.  The number $n$
will typically refer to the number of rows of a matrix, and $m$ the number of columns. $I_n$ will
denote the $n\times n$ identity matrix.  Denote $E_{i,j}$ to be the $n\times n$ square matrix with
its $(i,j)$-th entry being 1, and all other entries being zero. A vector is $k$-sparse if it has at
most $k$ non-zero entries.  Given a matrix $A$, $A^\dagger$ will denotes its transpose.  Given a
vector $\vec{x}\in\F^n$, $|\vec{x}|\eqdef n$.

A list of $n$ values in $\F$ is \textit{$t$-explicit} if each entry can be computed in $t$ steps,
where we allow operations in $\F$ to be done at unit cost.

Frequently throughout this paper we will divide a matrix into its diagonals, which we define as the
entries $(i,j)$ where $i+j$ is constant.  The following notation will make this discussion more
convenient.

\begin{notation}
	Let $M$ be an $n\times m$ matrix.  The \textbf{$k$-diagonal of $M$} is the set of entries
	$\{M_{i,j}\}_{i+j=k}$.  The \textbf{$(\le k)$-diagonals of $M$} is the set of entries
	$\{M_{i,j}\}_{i+j\le k}$.   The \textbf{$(<k)$-diagonals of $M$} is the set of entries
	$\{M_{i,j}\}_{i+j< k}$

	$M^{(k)}$, $M^{(\le k)}$ and $M^{(<k)}$ will denote the $k$-diagonal, $(\le k)$-diagonals
	and $(<k)$-diagonals of $M$, respectively.
\end{notation}

This notation will be frequently abused, in that a diagonal will refer to a set of positions in a
matrix in addition to referring to the values in those positions. However, \textit{the main
diagonal} of a matrix will refer to the entries $\{(i,i)\}_i$ of that matrix.

\section{Preliminaries}\label{sec:prelim}

In this section we formally define tensors as well as the PIT and LRR problems. We first discuss
tensors, and their notion of rank.  Rank-metric codes will be defined and discussed in
Section~\ref{sec:tensor codes}. Recall that we index starting at $0$, so we will use the product
space $\llb n\rrb ^d$ instead of $[n]^d$ for the domains of tensors.

\begin{definition}
	A \textbf{tensor} over a field $\F$ is a function $T:\prod_{j=1}^d\llb n_j\rrb \rightarrow
	\F$.  It is said to have degree $d$ and size $(n_1,\ldots,n_d)$.  If all of the $n_j$ are
	equal to $n$, then $T$ is said to have size $\llb n\rrb ^d$.

	Given two tensor $T_1,T_2$ of size $\prod_{j=1}^d\llb n_j\rrb$, $\la
	T_1,T_2\ra\eqdef\sum_{i_j\in\llb n_j\rrb} T_1(i_1,\ldots,i_d)T_2(i_1,\ldots,i_d)$.
\end{definition}

Note that the above inner product is the natural inner product when regarding a $\prod_{j=1}^d\llb
n_j\rrb$ tensor as a vector of dimension $\prod_{j=1}^dn_j$.  We now define the notion of rank.
Loosely, a tensor is rank 1 if it can be ``factored'' along each dimension, and a tensor is rank
$\le r$ if it can be expressed as the sum of $\le r$ rank 1 tensors.

\begin{definition}
	A tensor $T:\prod_{j=1}^d \llb n_j\rrb\to\mathbb{F}$ is \textbf{rank-one} if for $j\in[d]$
	there are vectors $\vec{v}_j\in\mathbb{F}^{n_j}\setminus\{\vec{0}\}$ such that
	$T=\otimes_{j=1}^d\vec{v}_j$.  That is, for all $i_j\in[n_j]$,
	$T(i_1,\ldots,i_d)=\prod_{j=1}^d\vec{v}_j(i_j)$ where $\vec{v}_j(i_j)$ denotes the $i_j$-th
	coordinate of $\vec{v}_j$.

	The \textbf{rank} of a tensor $T:\prod_{j=1}^d\llb n_j\rrb\to\mathbb{F}$, is defined as the
	minimum number of terms in a summation of rank-1 tensors expressing $T$, that is,
	\begin{equation*}
		\rank_\mathbb{F}(T)=\min\left\{r:T=\sum_{\ell=1}^r \otimes_{j=1}^d
		\vec{v}_{j,\ell}\text{, } \vec{v}_{j,\ell}\in \mathbb{F}^{n_j} \right\} \;.
	\end{equation*}
	\label{definition:simpletensor}
\end{definition}

As one might hope, when $d=2$ the above definitions reduce to the definition of a matrix, and
matrix-rank, respectively.  Further, the inner-product is then their Frobenius inner product.  That
is, $\la M_1,M_2\ra=\trace(M_1M_2^\dagger)$.

We now define the polynomial of a tensor.

\begin{definition}
	Let $T:\prod_{j=1}^d \llb n_j\rrb\to\mathbb{F}$ be a tensor, and let
	$\vec{x}_1,\ldots,\vec{x}_d$ be vectors of variables, so
	$\vec{x}_j=(x_{j,0},\ldots,x_{j,n_j-1})$ for all $j\in[d]$.  Then define
	\[f_T(\vec{x}_1,\ldots,\vec{x}_d)\eqdef\sum_{i_j\in\llb n_j\rrb} T(i_1,\ldots,i_d)\prod_{j=1}^D
	x_{j,i_j}=\la T,\vec{x}_1\otimes\cdots\otimes \vec{x}_d\ra \;,\]
	and define the $d$-variate polynomial
	\[\hat{f}_T(x_1,\ldots,x_d)\eqdef\sum_{i_j\in\llb n_j\rrb} T(i_1,\ldots,i_d)\prod_{j=1}^D
	x_{j}^{i_j}=f_T(\hat{\vec{x}}_1,\ldots,\hat{\vec{x}}_d) \;,\]
	where $(\hat{\vec{x}}_j)_i\eqdef x_j^i$.
	\label{defn:f_T}
\end{definition}

Note that the second equality in the first equation of the above definition follows from the
definition of the inner product over tensors.  As a matrix $M$ is also a tensor, we will also use
this notation when considering the polynomial $f_M(\vec{x},\vec{y})\eqdef \vec{x}^\dagger M\vec{y}$,
as the above definition readily generalizes the notion of a quadratic form.  Note that $\hat{f}_T$
allows us to consider any $d$-variate polynomial to be a tensor, and the \textit{rank} of such a
polynomial will simply be the rank of the corresponding tensor.

We now show the connection of these polynomials $f_T$ to set-multilinear depth-3 circuits.  We do
not seek to define all of the relevant terms in this notion, and instead refer the reader to the
recent survey~\cite{SY10}, and will simply define the subclass we are interested in.

\begin{definition}\label{def:set multilinear}
	For $j\in[d]$, let $\vec{x}_j = (x_{j,0},\ldots,x_{j,n-1})$ be vectors of variables.  A
	degree $d$, set-multilinear, $\Sigma\Pi\Sigma$ circuit with top fan-in $r$, is a polynomial
	of the following form \[C(\vec{x}_1,\ldots,\vec{x}_d) = \sum_{\ell=1}^{r} \prod_{j=1}^{d}\la
	\vec{v}_{j,\ell},\vec{x}_j\ra \] where each $\vec{v}_{j,\ell}\in\F^n$.
\end{definition}

We now see the following connection between these circuits and tensors.

\begin{lemma}
	The polynomials computed by degree $d$ set-multilinear $\Sigma\Pi\Sigma$ circuits, with top
	fan-in $\le r$, on $dn$ variables, are exactly the polynomials $f_T$, for tensors $T:\llb
	n\rrb^d\to \F$ with rank $\le r$.
\end{lemma}
\begin{proof}
	\underline{$\impliedby$:}  Suppose $T$ is of rank $\le r$, so $T=\sum_{\ell=1}^r
	\otimes_{j=1}^d \vec{v}_{j,\ell}$ for $\vec{v}_{j,\ell}\in \mathbb{F}^{n}$.  Then $f_T=\la
	T,\vec{x}_1\otimes\cdots\otimes \vec{x}_d\ra=\sum_{\ell=1}^r\la \otimes_{j=1}^d
	\vec{v}_{j,\ell},\vec{x}_1\otimes\cdots\otimes \vec{x}_d\ra=\sum_{\ell=1}^r\prod_{j=1}^d\la
	\vec{v}_{j,\ell},\vec{x}_j\ra$, and this final polynomial is computed as a set-multilinear
	$\Sigma\Pi\Sigma$ circuit.

	\underline{$\implies$:} This argument is simply the reverse of the above.
\end{proof}

We also get the following result for the polynomial $\hat{f}_T$.

\begin{lemma}
	For $T:\llb n\rrb^d\to \F$ with rank $\le r$,
	$\hat{f}_T(x_1,\ldots,x_d)=\sum_{\ell=1}^r\prod_{j=1}^d p_{j,\ell}(x_j)$, where $\deg
	p_{j,\ell}<n$.
	\label{lem:hatf}
\end{lemma}
\begin{proof}
	As $T$ is rank $\le r$, $T=\sum_{\ell=1}^r
	\otimes_{j=1}^d \vec{v}_{j,\ell}$ for $\vec{v}_{j,\ell}\in \mathbb{F}^{n}$.  
	Then $\hat{f}_T=f_T(\hat{\vec{x}}_1,\ldots,\hat{\vec{x}}_d)=\sum_{\ell=1}^r\prod_{j=1}^d\la
	\vec{v}_{j,\ell},\hat{\vec{x}}_j\ra$.  Taking $p_{j,\ell}(x_j)\eqdef \la
	\vec{v}_{j,\ell},\hat{\vec{x}}_j\ra$ yields the result.
\end{proof}

Recall that, as discussed in the introduction, set-multilinear $\Sigma\Pi\Sigma$ circuits have a
white-box polynomial-time PIT algorithm due to Raz and Shpilka~\cite{RazShpilka05} but no known
polynomial-sized black-box PIT algorithm.  By the above connection, this is the same as creating
hitting sets for tensors, which we will now define.

\begin{definition}\label{def:hitting set for sps}
	Let $\K$ be an extension of $\F$. A \textbf{hitting set $\mathcal{H}$ for $\prod_{j=1}^d\llb
	n_j\rrb$ tensors of rank $\le r$ over $\F$} is a set of points $\mathcal{H} \subseteq
	\prod_{j=1}^d(\K^{n_j})$ such that for any $T:\prod_{j=1}^d\llb n_j\rrb\to\F$ of rank $\le
	r$, $T$ is a non-zero iff there exists $(\vec{a}_1,\ldots,\vec{a}_d)\in\mathcal{H}$ such
	that $f_T(\vec{a}_1,\ldots,\vec{a}_d)\ne 0$.
\end{definition}

However, we saw in Definition~\ref{defn:f_T} that evaluating $f_T$ is equivalent to taking an inner
product of $T$ with a rank-1 tensor.  This leads to the following equivalent definition.

\begin{definition}[Reformulation of Definition~\ref{def:hitting set for sps}]\label{def:hitting set for tensors}
	Let $\K$ be an extension of $\F$. A \textbf{hitting set $\mathcal{H}$ for $\prod_{j=1}^d\llb
	n_j\rrb$ tensors of rank $\le r$ over $\F$} is a set of rank-1 tensors $\mathcal{H}
	\subseteq \K^{\prod_{j=1}^d\llb n_j\rrb}$ such that for any $T:\prod_{j=1}^d \llb
	n_j\rrb\to\F$ of rank $\le r$, $T$ is a non-zero iff there exists $H\in\mathcal{H}$ such
	that $\la T,H\ra\ne 0$.

	If $\mathcal{H}$ instead is not constrained to consist of rank-1 tensors, then we say
	$\mathcal{H}$ is an \textbf{improper hitting set}.
\end{definition}

As is common in PIT literature, we allow the use of the extension field $\K$, and in our case
$|\K|\le\poly(|\F|)$ will be sufficient.  However, the results of Section~\ref{sec:pit for tensors
over small fields} will show how to remove the need for $\K$ from our results (with some loss).

We now define our notion of a low-rank recovery set, extending Definition~\ref{def:hitting set for
tensors}.  Note that we drop here the restriction that the tensors must be rank 1.

\begin{definition}\label{defn:lrr}
	A set of tensors $\mathcal{R}\subseteq\K^{\prod_{j=1}^d\llb n_j\rrb}$ is an
	\textbf{$r$-low-rank-recovery set} if for every tensor $T:\prod_{j=1}^d\llb n_j\rrb\to\F$
	with rank $\le r$, $T$ is uniquely determined by $\vec{y}$, where
	$\vec{y}\in\K^{\mathcal{R}}$ is defined by $y_{R}\eqdef \la T,R\ra$, for $R\in\mathcal{R}$.

	An algorithm performs \textbf{recovery} from $\mathcal{R}$ if, for each such $T$, it
	recovers $T$ given $\vec{y}$.
\end{definition}

We now show that, despite low-rank recovery being a stronger notion than a hitting set, hitting sets
imply low-rank recovery with some loss in parameters, as seen by the following lemma.

\begin{lemma}\label{lem:hittolrr}
	If $\mathcal{H}$ is a (proper or improper) hitting-set for $\prod_{j=1}^d \llb n_j\rrb$
	tensors of rank $\le 2r$, then $\mathcal{H}$ is an $r$-low-rank-recovery set for
	$\prod_{j=1}^d \llb n_j\rrb$ tensors also.
\end{lemma}
\begin{proof}
	Let $A,B\in\F^{\prod_{j=1}^d \llb n_j\rrb}$ be two tensors of rank $\le r$ such that their
	inner products with the tensors in $\mathcal{H}$ are the same.  By linearity of the inner
	product, it follows then that the tensor $A-B$ has rank $\le 2r$ and has zero inner product
	with each tensor in $\mathcal{H}$.  As $\mathcal{H}$ is a hitting set, it follows that
	$A-B=0$, and thus $A=B$.  Therefore, tensors of rank $\le r$ are determined by their inner
	products with $\mathcal{H}$ and thus $\mathcal{H}$ is an $r$-low-rank-recovery set.
\end{proof}

We now discuss some trivial LRR results.  The first result is the obvious low-rank recovery
construction, which is extremely explicit but requires many measurements.

\begin{lemma}\label{lem:naivelrr}
	For $n\ge 1$, $d\ge 2$, there is a $\polylog(n,d,r)$-explicit $r$-low-rank-recovery set for
	$\llb n\rrb^d$ tensors, of size $n^d$.  Further, recovery of $T$ is possible in $\poly(n^d)$
	time.
\end{lemma}
\begin{proof}
	For $\vec{i}=(i_1,\ldots,i_d)\in\llb n\rrb$, let the rank 1 tensor $R_\vec{i}:\llb
	n\rrb^d\to\F$ be the rank 1 tensor, which is the indicator function for the set
	$\{(i_1,\ldots,i_d)\}$.  Thus, $\la T,R_\vec{i}\ra=T(i_1,\ldots,i_d)$.  It follows that
	$T=0$ iff each such inner product is zero, and further that recovery of $T$ is possible (in
	$\poly(n^d)$ time).  The explicitness of the recovery set is also clear.
\end{proof}

We now will show that, via the probabilistic method, one can show that much smaller low-rank
recovery sets exist.  To do so, we first cite the following form of the Schwartz-Zippel Lemma.

\begin{lemma}[Schwartz-Zippel Lemma~\cite{Schwartz80,Zippel79}]
	Let $f\in\F[x_1,\ldots,x_m]$ be a non-zero polynomial of total degree $\le d$, and
	$S\subseteq\F$. Then $\Pr_{\vec{x}\in S^m}[f(\vec{x})=0]\le d/|S|$.
\end{lemma}

We now give a (standard) probabilistic method proof that small hitting sets exist (over finite
fields).  We present this not as a tight result, but as an example of what parameters one can hope
to achieve.

\begin{lemma}\label{lem:szhit}
	Let $\F_q$ be the field on $q$ elements.  Let $n\ge 1$ and $q>d\ge 2$.  Then there is a
	hitting set for $\llb n\rrb^d$ tensors of rank $\le r$, of size $\le
	dnr/\log_q(q/d)+1\approx dnr$.
	Further, there is an $r$-low-rank recovery set of size $\le 2dnr/\log_q(q/d)+2$.
\end{lemma}
\begin{proof}
	For any non-zero tensor $T:\llb n\rrb^d\to F$, $f_T$ has degree $d$, and thus by the
	Schwartz-Zippel Lemma, for a random $\vec{a}\in\F_q^n$, $f_T(\vec{a})=0$ with probability at
	most $d/q$.  There are at most $q^{dnr}$ such non-zero tenors.  By a union bound, it follows
	that $k$ random points are not a hitting set for rank $\le r$ tensors with probability at
	most $q^{dnr}(d/q)^k$, which is $<1$ if $k>dnr/\log_q(q/d)$.  The low-rank-recovery set
	follows from Lemma~\ref{lem:hittolrr}.
\end{proof}

We now briefly remark on the tightness of the above result.  The general case of tensors is not well
understood, as it is not well-understood how many tensors there are of a given rank.  For matrices,
the situation is much more clear.  In particular, Roth~\cite{Roth91} showed (using the language of
rank-metric codes) that over finite fields the best (improper) hitting set for $n\times n$ matrices
of rank $\le r$ is of size $nr$, and over algebraically closed fields the best (improper) hitting
set is of size $(2n-r)r$.  As we will aim to be field independent, the second bound is more
relevant, and we indeed match this bound (as seem in Theorem~\ref{thm:better hitting for matrices})
with a proper hitting set.

Clearly, the above lemma is non-explicit.  However, it yields a much smaller hitting set than the
$n^d$ result given in Lemma~\ref{lem:naivelrr}.  Note that previous work (even for $d=2$) on LRR and
rank-metric codes did not focus on requiring that the measurements are rank-1 tensors, and thus
cannot be used for PIT.  Given this lack of knowledge, this paper seeks to construct proper hitting sets,
and low-rank-recovery sets, that are both explicit \textit{and} small.

We remark that any explicit hitting set naturally leads to tensor rank lower bounds\footnote{This
connection, along with the connection to rank-metric codes mentioned earlier, can be put in a more
broad setting: hitting sets (and thus lower-bounds) for circuits from some class $\mathcal{C}$ are
in a sense equivalent to $\mathcal{C}$-metric linear codes.  That is, codes where $\text{dist}(x,y)$
is defined as the size of the smallest circuit whose truth table is the string $x-y$. We do not
pursue this idea further in this work.}.  The
following lemma, which can be seen as a special case of the more general results of
Heintz-Schnorr~\cite{HeintzSchnorr80} and Agrawal~\cite{Agrawal05}, shows this connection more
concretely.

\begin{lemma}\label{lem:hittolbs}
	Let $\mathcal{H}$ be a hitting set for $\llb n\rrb^d$ tensors of rank $\le r$, such that
	$|\mathcal{H}|<n^d$.  Then there is a $\poly(n^d,|\mathcal{H}|)$-explicit tensor of rank
	$>r$.
\end{lemma}
\begin{proof}
	Consider the constraints imposed on a tensor $T$ by the system of equations $\la
	T,\mathcal{H}\ra=\vec{0}$.  There are $|\mathcal{H}|$ constraints and $n^d$ variables.  It
	follows that there is a non-zero $T$ solving this system.  By the definition of a hitting
	set, it follows that $\rank(T)\not\le r$.  That $T$ is explicit follows from Gaussian
	Elimination.
\end{proof}

For $d=2$, the above is less interesting, as matrix rank is well understood and we know many
matrices of high rank.  For $d\ge 3$, tensor rank is far less understood.  For $d=3$, the best known
lower bounds for the rank of explicit tensors, over arbitrary fields, due to Alexeev, Forbes, and
Tsimerman~\cite{alexeev+forbes+tsimerman:2011:tensor-rank}, are $3n-\O(\lg n)$ (over $\F_2$, a lower
bound of $3.52n$ is known, essentially due to Brown and Dobkin~\cite{BrownDobkin80}).  More
generally, for any fixed $d$, no explicit tensors are known with tensor rank $\omega(n^{\lfloor
d/2\rfloor})$.  The above lemma shows that constructing hitting sets is at least as hard as getting a
lower bound on any specific tensor.  In particular, constructing a hitting set for $\llb
n\rrb^d$ tensors of rank $\le r$ of size $\O(dnr^k)$ with $k<2$ would yield new tensor rank lower
bounds for odd $d$, in particular $d=3$.  Such lower bounds would imply new circuit lower bounds,
using the results of Strassen~\cite{Strassen73a} and Raz~\cite{Raz10}.  Our results give a hitting
set with $k\approx \lg d$, and we leave open whether further improvements are possible.

We will mention the definitions and preliminaries of rank-metric codes in Section~\ref{sec:tensor
codes}.

\subsection{Paper Outline} We briefly outline the rest of the paper.  In Section~\ref{sec: improved
GR} we give our improved construction of rank-preserving matrices, which were first constructed by
Gabizon-Raz~\cite{GabizonRaz08}.  In Section~\ref{sec:pit for matrices} we then use this
construction to give our reduction from bivariate identity testing to univariate identity testing
(Section~\ref{sec:bivariate}), which then readily yields our hitting set for matrices
(Section~\ref{subsect:hitmatrix}).  In Section~\ref{sec:diag} we show an equivalent hitting set,
which is more useful for low-rank-recovery.

Section~\ref{sec:pit for tensors} extends the above results to tensors, where
Section~\ref{sec:dvariate} reduces $d$-variate identity testing to univariate identity testing, and
Section~\ref{sec:hittensor} uses this reduction to construction hitting sets for tensors.  Finally,
Section~\ref{sec:pit for tensors over small fields} shows how to extend these results to any field.

Low-rank recovery of matrices is discussed in Section~\ref{sec:explicit LRR}.  It is split into two
parts.  Section~\ref{sec:prony} shows how to decode dual Reed-Solomon codes, which we use as a
sparse-recovery oracle.  Section~\ref{sec:lrrtosparse} shows how to, given any such sparse-recovery
oracle, perform low-rank-recovery of matrices.  Instantiating the oracle with dual Reed-Solomon
codes gives our low-rank-recovery construction.

Section~\ref{sec:tensor codes} shows how to extend our LRR algorithms to tensors, and how to use
these results to construct rank-metric codes. Finally, Section~\ref{sec:discuss} discusses some problems left open by this work.

\section{Improved Construction of Rank-preserving Matrices}\label{sec: improved GR}

In this section we will give an improved version of the Gabizon-Raz lemma~\cite{GabizonRaz08} on the
construction of rank-preserving matrices.  The goal is to transform an $r$-dimensional subspace
living in an $n$-dimensional ambient space, to an $r$-dimensional subspace living in an
$r$-dimensional ambient space.  We will later show (see Theorem~\ref{thm:bivariate poly evaluation})
how to use such a transformation to reduce the problem of PIT for $n\times m$ matrices of rank $\le
r$ to the problem of PIT for $r\times r$ matrices of rank $\le r$.

We first present the Gabizon-Raz lemma (\cite{GabizonRaz08}, Lemma 6.1), stated in the language of
this paper.

\begin{lemma*}[Gabizon-Raz (\cite{GabizonRaz08}, Lemma 6.1)]
	Let $1\le r\le n$. Let $M\in\F^{n\times r}$ be of rank $r$. Define $A_\alpha\in\F^{r\times
	n}$ by $(A_\alpha)_{i,j}=\alpha^{ij}$.  Then there are $\le nr^2$ values $\alpha\in \F$ such
	that $\rank(A_\alpha M)<r$.
\end{lemma*}

Our version of this lemma gives a set of matrices parameterized by $\alpha$ where there are only
$nr$ values of $\alpha$ that lead to $\rank(A_\alpha M)<r$.  This extra factor of $r$ allows us to
achieve an $\O((n+m)r)$-sized hitting set for matrices instead of a $\O((n+m)r^2)$-sized hitting
set.  We comment more on the necessity of this improvement in Remark~\ref{rmk:gr}.  We now state our
version of this lemma. Our proof is very similar to that of Gabizon-Raz.

\begin{theorem}\label{thm:better rank preserving}
	Let $1\le r\le n$. Let $M\in\F^{n\times r}$ be of rank $r$.  Let $\K$ be a field extending
	$\F$, and let $g\in\K$ be an element of order $\ge n$.  Define $A_\alpha\in\K^{r\times n}$
	by $(A_\alpha)_{i,j}=(g^i\alpha)^j$.  Then there are $\le nr-\binom{r+1}{2}<nr$ values
	$\alpha\in \K$ such that $\rank(A_\alpha M)<r$.
\end{theorem}
\begin{proof}
	We will now treat $\alpha$ as a variable, and thus refer to $A_\alpha$ simply as $A$.  The
	matrix $AM$ is an $r\times r$ matrix, and thus the claim will follow from showing that
	$\det(AM)$ is a non-zero polynomial in $\alpha$ of degree $\le nr-\binom{r+1}{2}$. As $r \ge
	1$, $nr-\binom{r+1}{2}<nr$.

	To analyze this determinant, we invoke the Cauchy-Binet formula.
	\begin{lemma*}[Cauchy-Binet Formula, Lemma~\ref{lem: cauchy-binet}]
		Let $m\ge n\ge 1$. Let $A\in\F^{n\times m}$, $B\in\F^{m\times n}$.  For
		$S\subseteq\llb m\rrb$, let $A_S$ be the $n\times |S|$ matrix formed from $A$ by
		taking the columns with indices in $S$.  Let $B_S$ be defined analogously, but with
		rows.  Then \[\det(AB)=\sum_{S\in\binom{\llb m\rrb}{n}}\det(A_S)\det(B_S)\]
	\end{lemma*}
	so that \[\det(AM)=\sum_{S\in\binom{\llb n\rrb}{r}}\det(A_S)\det(M_S)\] For
	$S=\{k_1,\ldots,k_r\}$,
	\begin{align*}
		\det(A_S)
		&=
		\begin{vmatrix}
			(\alpha)^{k_1}		&\cdots	&(\alpha)^{k_r}\\
			(g\alpha)^{k_1}		&\cdots	&(g\alpha)^{k_r}\\
			\vdots			&\ddots	&\vdots\\
			(g^{r-1}\alpha)^{k_1}	&\cdots	&(g^{r-1}\alpha)^{k_r}
		\end{vmatrix}
		=
		\begin{vmatrix}
			1		&\cdots	&1\\
			g^{k_1}		&\cdots	&g^{k_r}\\
			\vdots		&\ddots	&\vdots\\
			(g^{k_1})^{r-1}	&\cdots	&(g^{k_r})^{r-1}
		\end{vmatrix}
		\cdot
		\alpha^{\sum_{\ell=1}^r k_\ell}\\
		&=\alpha^{\sum_{\ell=1}^r k_\ell}\prod_{1\le i<j\le r}(g^{k_j}-g^{k_i})
	\end{align*}
	By assumption the order of $g$ is $\ge n$, so the elements $(g^k)_{0\le k<n}$ are
	distinct, implying that the above Vandermonde determinant is non-zero.

	Further, we observe that $\deg_\alpha \det(A_S)=\sum_{k\in S} k$.  As $S\in\binom{\llb
	n\rrb}{r}$, it follows that $\sum_{k\in S} k\le \sum_{k=n-r}^{n-1} k=nr-\binom{r+1}{2}$, and
	thus $\deg_\alpha\det(AM)\le nr-\binom{r+1}{2}$ also.

	We now show $\det(AM)$ is not identically zero, as a polynomial in $\alpha$.  We show this
	by showing that there is no cancellation of terms at the highest degree of $\det(AM)$.  That
	is, there is a unique set $S\in\binom{\llb n\rrb}{r}$ maximizing $\sum_{k\in S}k$ subject to
	$\det(M_S)\neq0$.  This is proven by the following lemma.

	\begin{lemma}\label{lem: maximizing sum of rows}
		Let $m\ge n\ge 1$. Let $M$ be a $n\times m$ matrix of rank $n$.  For
		$S\subseteq\llb m\rrb$, denote $M_S$ as the $n\times |S|$ matrix formed by taking the
		columns in $M$ (in order) whose indices are in $S$.  Denote $w(S)\eqdef \sum_{s\in
		S} s$.  Then there is a unique set $S\in\binom{\llb m\rrb}{n}$ that maximizes $w(S)$
		subject to $\det(M_S)\ne 0$.
	\end{lemma}
	\begin{proof}
		The proof uses the ideas of the Steinitz Exchange Lemma.  That is, recall the
		following facts in linear algebra. If sets $S_1,S_2$ are both sets of linearly
		independent vectors, and $|S_1|>|S_2|$, then there is some $\vec{v}\in S_1\setminus
		S_2$ such that $S_2\cup\{\vec{v}\}$ is linearly independent.  Thus, if $S_1,S_2$ are
		both sets of linearly independent vectors and $|S_1|=|S_2|$ then for any $\vec{w}\in
		S_2\setminus S_1$ there is a vector $\vec{v}\in S_1\setminus S_2$ such that
		$(S_2\setminus\{\vec{w}\})\cup\{\vec{v}\}$ is linearly independent.

		Now suppose (for contradiction) that there are two different sets $S_1,S_2\subseteq
		\llb m\rrb$ that maximize $w(S)$ over the sets such that $\det(M_S)\neq 0$, so that
		$|S_1|=|S_2|=n$.  Pick the
		smallest index $k$ in the (non-empty) symmetric difference $(S_2\setminus S_1) \cup
		(S_1\setminus S_2)$.  Without loss of generality suppose $k\in S_2\setminus S_1$.
		It follows that there is an index $l\in S_1\setminus S_2$ such that the columns in
		$S_3\eqdef(S_2\setminus\{k\})\cup\{l\}$ are linearly independent (by the Steinitz
		Exchange Lemma), and thus $\det(M_{S_3})\ne 0$ as $|S_3|=n$ by construction.

		By choice of $k$ and construction of $l$, $k\ne l$ and thus $k<l$.  Thus,
		$w(S_3)=w(S_2)+l-k>w(S_2)$.  However, this contradicts that $S_2$ was a maximizer to
		$w(S)$ subject to $\det(M_S)\ne 0$. Thus, the assumption of non-unique maximizers is
		false; there must be a unique maximizer.
	\end{proof}

	Thus $\det(AM)$ is a non-zero polynomial of degree $\le nr-\binom{r+1}{2}$ in $\alpha$, so
	there are at most that many values such that $\det(AM)=0$.
\end{proof}

We remark that Lemma~\ref{lem: maximizing sum of rows} can be seen as a special case of a more
general result about matroids, which states that if each element in the ground set has a unique
(positive) weight, then there is a unique independent set with maximal weight.  However, as we index
matrix columns starting at 0 this general fact does not immediately apply.  Rather, we implicitly
use that all bases in vector matroids have the same number of vectors.  In such a case, the weight
function can be shifted by an additive constant without affecting the property of having a unique
maximizer.

We now extend the above result to the case when the rank of the $n\times r$ matrix may be less than
$r$.  This will be useful when studying hitting sets for rank $\le r$ matrices, for then we do not
know the true rank of the unknown matrix, and only have the bound of ``$\le r$''.

\begin{corollary}\label{cor:better rank preserving}
	Let $1\le s\le r\le n$. Let $M\in\F^{n\times r'}$ be of rank $s$, for $r'\ge s$.  Let $\K$
	be a field extending $\F$, and let $g\in\K$ be an element of order $\ge n$.  Define
	$A_\alpha\in\K^{r\times n}$ by $(A_\alpha)_{i,j}=(g^i\alpha)^j$.  Then there are $\le
	nr-\binom{r+1}{2}<nr$ values $\alpha\in \K$ such that the first $s$ rows of $A_\alpha M$
	have rank $<s$.
\end{corollary}
\begin{proof}
	Consider $M'\in\F^{n\times s}$ to be a matrix formed from $s$ basis columns of $M$.  It
	follows, from Theorem~\ref{thm:better rank preserving}, that there are at most
	$ns-\binom{s+1}{2}$ values of $\alpha$ such that the $s\times n$ matrix $A'_\alpha$ has
	$\rank(A'_\alpha M')<s$.  As $\rank(AM')=\rank(AM)$ holds for any $A$, there are at most
	$ns-\binom{s+1}{2}$ many values of $\alpha$ such that $\rank(A'_\alpha M)<s$. Also, as
	$ns-\binom{s+1}{2}\le nr-\binom{r+1}{2}$ for $s\le r\le n$, it also holds that there are
	$\le nr-\binom{r+1}{2}$ values of $\alpha$ such that $\rank(A'_\alpha M)<s$.  Finally the
	claim follows by observing that, by construction, $A'_\alpha M$ is exactly the first $s$
	rows of $A_\alpha M$.
\end{proof}

\section{Identity Testing for Matrices}\label{sec:pit for matrices}

The previous section showed how we can map an $r$-dimensional subspace of an $n$-dimensional ambient
space to an $r$-dimensional subspace of an $r$-dimensional ambient space.  In this section, we will
use this map to reduce the PIT problem for rank $\le r$ matrices of size $n\times m$ to the PIT
problem from rank $\le r$ matrices of size $r\times r$.  This will be done by applying the dimension
reduction twice, once to the rows and once to the columns.  Further, the $r\times r$ version can be
solved in $r^2$ evaluations, using the naive approach of Lemma~\ref{lem:naivelrr} in querying each
entry in the matrix.  When phrased this way, one can show that this gives a
$\Theta((n+m)\poly(r))$-sized hitting set. This reduction idea is analogous to the kernelization
technique used in fixed-parameter tractability, but we do not develop this connection further. While
this idea demonstrates the feasibility of the rough bound cited above, we actually achieve a
$\Theta((n+m)r)$-sized hitting set via tighter analysis.

\subsection{Variable Reduction}\label{sec:bivariate}

Before giving the hitting set construction and its analysis, we first present the main theorem used
in the analysis.  While its statement seems unrelated to the intuition presented above, the proof
will exploit this intuition. When interpreting the result, recall that we index entries in matrices
(and vectors) starting at 0, as well as recalling the definition of $\hat{f}_T$ from a tensor $T$.

\begin{theorem}\label{thm:bivariate poly evaluation}
	Let $m\ge n \ge r\ge 1$.  Let $\K$ be an extension of $\F$ such that $g\in\K$ has order $\ge
	m$. Let $M$ be an $n\times m$ matrix of rank $\le r$ over $\F$. Then $M$ is non-zero (over
	$\F$) iff one of the univariate polynomials
	$\hat{f}_M(x,x),\hat{f}_M(x,gx),\ldots,\hat{f}_M(x,g^{r-1}x)$ is non-zero (over $\K$).
\end{theorem}
\begin{proof}
	\underline{$(\impliedby):$} If $M$ is zero then so must all $\hat{f}_M(x,g^ix)$ be as well.  Taking
	the contrapositive yields this direction.

	\underline{$(\implies):$}  Say $\rank(M)=s$.  By assumption $0<s\le r$.  Recall that putting
	$M$ into reduced row-echelon form yields a decomposition $M=PQ^\dagger$, such that
	$P\in\F^{n\times s}$ and $Q\in\F^{m\times s}$ such that $\rank(P)=\rank(Q)=s$.  We remark
	that it is crucial for our proof that we have ``$\rank(P)=\rank(Q)=s$'' here.  Invoking the
	bound ``$\rank(P),\rank(Q)\le s$'', which one gets directly via the definition of rank of
	$M$, is insufficient.

	We now exploit the kernelization idea mentioned above.  Consider the matrices $A_\alpha \in
	\K^{r\times n}$ and $B_\alpha\in\K^{r\times m}$ defined by $(A_\alpha)_{i,j}=(g^i\alpha)^j$
	and $(B_\alpha)_{i,j}=(g^i\alpha)^j$.  Now consider $A_\alpha P$ and $B_\alpha Q$, which
	have sizes $r\times s$ each.  Write them in block notation as
	$\begin{pmatrix}P'_\alpha\\P''_\alpha\end{pmatrix}$ and
	$\begin{pmatrix}Q'_\alpha\\Q''_\alpha\end{pmatrix}$  such that $P'_\alpha$ and $Q'_\alpha $
	are both $s\times s$ matrices.

	By our refinement of the Gabizon-Raz lemma~\cite{GabizonRaz08}, our
	Corollary~\ref{cor:better rank preserving}, it follows that there are $<nr$ values of
	$\alpha$ such that $\rank(P'_\alpha)<s$ and $<mr$ values of $\alpha$ such that
	$\rank(Q'_\alpha)<s$.  By the union bound, there are $<(n+m)r$ values such that
	$\rank(P'_\alpha)<s$ or $\rank(Q'_\alpha)<s$.  Let $\H$ be an extension field of $\K$, such
	that $|\H|\ge (n+m)r$.  It follows that there is some $\alpha\in\H$ such that
	$\rank(P'_\alpha)=s$ and $\rank(Q'_\alpha)=s$.  Fix this as the value of $\alpha$, and we
	now drop $\alpha$ from our notation.

	Via block multiplication we see that
	\begin{equation*}
		AMB^\dagger=AP(BQ)^\dagger=
			\begin{pmatrix}P'\\P''\end{pmatrix}\begin{pmatrix}Q'&Q''\end{pmatrix}
			=\begin{pmatrix}P'Q'&P'Q''\\P''Q'&P''Q''\end{pmatrix}
	\end{equation*}
	As $\rank(P')=s$ and $\rank(Q')=s$, it follows that $\rank(P'Q')=s$. We remark that it is
	here where the naive bound ``$\rank(P),\rank(Q)\le s$'' is insufficient, and we crucially
	use that ``$\rank(P)=\rank(Q)=s$''.

	As $\rank(P'Q')=s$, and $P'Q'$ is an $s\times s$ matrix, it follows that some entry in its
	first row (which has index 0, by our notation) is non-zero. As $P'Q'$ is a principal minor
	of $AMB^\dagger$, it follows that some entry in the first row of $AMB^\dagger$ is non-zero.
	Denote row $i$ of $A$ as $A_i$, and row $j$ of $B$ as $B_j$.  As the first row of
	$AMB^\dagger$ is $A_0MB^\dagger$, it follows then there is some $0\le \ell\le r-1$ such that
	$A_0MB_\ell^\dagger\ne 0$.  Expanding this evaluation out, we see that
	\begin{align*}
		A_0MB_\ell^\dagger
			&=\la M, A_0B_\ell^\dagger\ra
			=\sum_{i=0,j=0}^{n-1,m-1} M_{i,j}\cdot (A_0)_i(B_\ell)_j\\
			&=\sum_{i=0,j=0}^{n-1,m-1} M_{i,j}A_{0,i}B_{\ell,j}\\
			&=\sum_{i=0,j=0}^{n-1,m-1} M_{i,j}(g^0\alpha)^i(g^\ell\alpha)^j\\
			&=\hat{f}_M(\alpha,g^\ell\alpha)
	\end{align*}
	Thus, we see that $\hat{f}_M(x,g^\ell x)$ has a non-zero point over the field $\H$.  It follows that
	it is a non-zero polynomial over $\H$.  As it has coefficients over $\K$, $\hat{f}_M(x,g^\ell x)$ is
	non-zero over $\K$ as well.
\end{proof}

\begin{remark}
	We now remark on how to implement the kernelization idea, mentioned in the introduction to
	this section, in a more straight-forward sense. One can see that $\rank(P'Q')=s$ shows that
	$AMB^\dagger\ne 0$.  As $AMB^\dagger$ is of size $r\times r$, we can then run the naive
	$r^2$-size hitting set of Lemma~\ref{lem:naivelrr} for $r\times r$-sized matrices, which
	checks each individual entry.  Noting that the $(i,j)$-th entry of $AMB^\dagger$ is equal to
	$\la M,A_iB_j^\dagger\ra$ we see that we can implement this naive hitting set as a hitting
	set for $n\times m$ matrices.

	Thus, for each $\alpha$ there are $r^2$ rank-1 matrices to test, and we need at most
	$(n+m)r$ choices of $\alpha$ (where here we assume $\K$ is at least this big).  It follows
	that there exists an explicit hitting set of size $(n+m)r^3$.
\end{remark}

\begin{remark}
	\label{rmk:gr}
	We briefly discuss the necessity of our version of the Gabizon-Raz lemma for the above
	proof.  The above proof does not invoke our version of the lemma in the fullest, in the
	sense that the $nr$ bound on the number of ``bad'' $\alpha$ was only used in the sense that
	it was a finite bound.  Thus, given that our version of the lemma ``only'' improves the
	$nr^2$ bound of Gabizon-Raz to $nr$, it may be unclear why our version is needed here.

	The crucial use of our version of the lemma is keeping the degree low.  That is, if one
	invoked the original Gabizon-Raz lemma, one would result in ``$M$ is non-zero iff one of the
	univariate polynomials $\hat{f}_M(x,x),\hat{f}_M(x,x^2),\ldots,\hat{f}_M(x,x^r)$ is
	non-zero''. While this is correct, it will lead to a larger hitting set as one needs to
	interpolate $r$ polynomials, each of degree $\approx rn$, which will give a
	$\Theta(nr^2)$-sized set instead of the $\Theta(nr)$-sized set we are able to achieve.
\end{remark}

We also state an equivalent version of this result, which will be useful for higher-degree tensors.

\begin{corollary}\label{cor:bivariate poly evaluation}
	Let $m\ge n \ge r\ge 1$. Over the field $\F$, consider the bivariate polynomial
	$f(x,y)=\sum_{i=1}^r p_i(x)q_i(y)$ such that $\deg(p_i)<n$ and $\deg(q_i)<m$ for all $i$.
	Let $\K$ be an extension of $\F$ such that $g\in\K$ has order $\ge m$.

	\sloppy Then $f$ is non-zero (over $\F$) iff one of the univariate polynomials
	$f(x,x),f(x,gx),\ldots,f(x,g^{r-1}x)$ is non-zero (over $\K$).
\end{corollary}

\subsection{The Hitting Set for Matrices}\label{subsect:hitmatrix}

In this subsection we use the theorem of the last subsection to construct hitting sets for matrices.
First, recall our notion of a hitting set for matrices, as given in Definition~\ref{def:hitting set
for tensors}.  Now recall that Theorem~\ref{thm:bivariate poly evaluation} shows that for any $M$ of
rank $\le r$, $M$ is non-zero iff one of the polynomials in $\{\hat{f}_M(x,g^\ell x)\}_{0\le \ell
<r}$ is non-zero.  In the preliminaries it was seen that evaluating one of these polynomials at a
point $\alpha$ is equivalent to taking an inner product $\la M,A\ra$ with a rank-1 matrix $A$.  This
leads naturally to the following idea: interpolate each of the $r$ polynomials in
$\{\hat{f}_M(x,g^\ell x)\}_{0\le \ell <r}$.  As each polynomial is of degree $\le n+m-2$, this will
lead to $(n+m-1)r$ inner products.  Then $M$ is non-zero iff one of these inner products is
non-zero.  This is the exact idea, which we now make formal.

\begin{construction}\label{construction:construction of hitting set for matrices}
	Let $m\ge n\ge r\ge 1$. Let $\K$ be an extension of $\F$ such that $g\in\K$ is of order $\ge
	m$ and $\alpha_0,\ldots,\alpha_{n+m-2}\in\K$ are distinct.  Let $B_{k,\ell}\in\K^{n\times
	m}$ to be the rank-1 matrix defined by $(B_{k,\ell})_{i,j}=\alpha_k^i(g^\ell\alpha_k)^j$,
	and let $\mathcal{B}_{r,n,m}\eqdef\{B_{k,\ell}\}_{0\le \ell<r,0\le k\le n+m-2}$.
\end{construction}

We now give the analysis for this hitting set.

\begin{theorem}\label{thm:hitting set for matrices}
	Let $m\ge n\ge r\ge 1$.  Then $\mathcal{B}_{r,n,m}$, as defined in
	Construction~\ref{construction:construction of hitting set for matrices}, has the following
	properties:
	\begin{enumerate}
		\item $\mathcal{B}_{r,n,m}$ is a hitting set for $n\times m$ matrices of rank $\le
		r$ over $\F$.
		\item $|\mathcal{B}_{r,n,m}|=(n+m-1)r$
		\item $\mathcal{B}_{r,n,m}$ can be computed in $\poly(m)$ operations, where
		operations (including a successor function in some enumeration of $\K$) over $\K$
		are counted at unit cost.
	\end{enumerate}
\end{theorem}
\begin{proof}
	\underline{$|\mathcal{B}_{r,n,m}|=(n+m-1)r$:} This is by definition.

	\underline{$\mathcal{B}_{r,n,m}$ can be computed in $\poly(m)$ operations:}  We assume here
	an enumeration of elements in $\K$ such that the successor in this enumeration can be
	computed at unit cost.  We also will assume testing whether an element is zero, as well as
	arithmetic operations in the field, are done at unit cost.

	First observe that there are at most $m$ solutions to $x^m-1$ over $\K$, so if we enumerate
	$m+1$ elements of $\K$, then we can find a $g\in\K$ with order $\ge m$.  This is in
	$\poly(m)$ operations. Similarly, the enumeration will give us $n+m-1$ distinct elements
	which yield the desired $\alpha_k$.  Then, computing each $B_{k,l}$ can be done in
	$\poly(m)$ steps, and there are $\poly(m)$ of them.  Thus, all of $\mathcal{B}_{r,n,m}$ can
	be computed in this many operations.

	\underline{$\mathcal{B}_{r,n,m}$ is a hitting set:} $\mathcal{B}_{r,n,m}$ is a set of rank-1
	matrices by construction, so it remains to prove that it hits each low-rank matrix. Let $M$
	be $n\times m$ matrix of rank $\le r$ in $\F$.  By Theorem~\ref{thm:bivariate poly
	evaluation}, we see that $M$ is non-zero iff one of the polynomials $\{\hat{f}_M(x,g^\ell
	x)\}_{0\le \ell<r}$ is non-zero.  Thus, $\hat{f}_M(\alpha_k,g^\ell\alpha_k)=\sum_{0\le
	i<n,0\le j<m} M_{i,j}\alpha_k^i(g^\ell\alpha_k)^j=\la M,B_{k,\ell}\ra$.  As each
	$\hat{f}_M(x,g^\ell x)$ is of degree $\le n+m-2$ and we evaluate each polynomial at $n+m-1$
	points, each $\hat{f}_M(x,g^\ell x)$ is fully determined by these evaluations via the
	polynomial interpolation map.  Specifically, if $\hat{f}_M(x,g^\ell x)$ is non-zero then it
	must have a non-zero evaluation for some $\alpha_k$.  As some $\hat{f}_M(x,g^\ell x)$ is
	non-zero by Theorem~\ref{thm:bivariate poly evaluation}, it follows that $\la
	M,B_{k,\ell}\ra\ne 0$ for some $0\le \ell<r$ and $0\le k\le n+m-2$.
\end{proof}

One deficiency with this construction is that for large $r$ it is suboptimal by a factor of $2$.
That is, in the regime where $n=m$ and $r=n-1$ this construction gives a hitting set of size
$(2n-1)(n-1)\approx 2n^2$.  However, the naive hitting set yields an $n^2$-sized setting.  In the
next subsection we show that this is an artifact of the analysis.  That is, by pruning unneeded
matrices from the hitting set, one can show that our construction always (for $r<n$) does better
than the naive construction.  This result proven in Theorem~\ref{thm:better hitting for matrices}.

\subsection{An Alternate Construction}\label{sec:diag}

In the previous subsections we saw that a low-rank matrix $M$ is non-zero iff one of the polynomials
$\{\hat{f}_M(x,g^\ell x)\}$ was non-zero.  To construct a hitting set, we then interpolated each
$\hat{f}_M$ at enough points to determine which, if any, were non-zero.  However, we are
interpolating many ``related'' polynomials all on the same points, so it is natural to wonder if
there are some redundancies in this process.

To phrase things differently, observe that testing a matrix $M$ against a hitting set $\mathcal{H}$
is really asking of $M\in\ker\mathcal{H}$.  The promise that $M$ is low-rank ensures that
$M\in\ker\mathcal{H}$ iff $M$ is zero.  The number of tests done is $|\mathcal{H}|$, but the number
of \textit{actual} tests is $\rank(\mathcal{H})$, where we consider $\mathcal{H}$ as vectors in the
vector space $\K^{nm}$.  That is, some of the matrices in $\mathcal{H}$ may be linearly dependent,
and these are redundancies that can be pruned.

The aim of this section is to present hitting sets (and improper hitting sets) that have linearly
independent test matrices.  The initial motivation is to observe that the point of the evaluations
of the $\{\hat{f}_M(x,g^\ell x)\}$ was to interpolate the coefficients.  Thus, instead of doing these
evaluations, we can express the coefficients of the $\{\hat{f}_M(x,g^\ell x)\}$ directly as linear
combinations of the entries in $M$.  This will lead to the following improper hitting set.

\begin{construction}\label{constr:diag}
	Let $m\ge n\ge r \ge 1$. Let $\K$ be an extension of $\F$ such that $g\in\K$ is of order
	$\ge m$. Let $D_{k,\ell}\in\K^{n\times m}$ be the matrix defined be
	\begin{equation*}
		(D_{k,\ell})_{i,j}=
			\begin{cases}
				g^{\ell j}	&	\text{if } i+j=k\\
				0		&	\text{else}
			\end{cases}
	\end{equation*}
	Define $\mathcal{D}_{r,n,m}\eqdef \{D_{k,\ell}\}_{\genfrac{}{}{0pt}{}{0\le k\le n+m-2}{0\le
	\ell <r}}$, and $\mathcal{D}'_r\eqdef \{D_{k,\ell}\}_{\genfrac{}{}{0pt}{}{0\le k\le
	n+m-2}{0\le \ell< \min(r,k+1,(n+m)-(k+1))}}$.
	\label{construction:diag}
\end{construction}

We now analyze this construction.

\begin{theorem}\label{thm:diagonal set for matrices}
	Let $m\ge n\ge r\ge 1$.  Then $\mathcal{D}_{r,n,m}$, as defined in
	Construction~\ref{constr:diag}, has the following properties:
	\begin{enumerate}
		\item $\mathcal{D}_{r,n,m}$ is an improper hitting set for $n\times m$ matrices of
		rank $\le r$ over $\F$.
		\item $\sspan(\mathcal{D}_{r,n,m})=\sspan(\mathcal{B}_{r,n,m})$ (as vectors in
		$\K^{nm}$)
		\item $|\mathcal{D}_{r,n,m}|=(n+m-1)r$
		\item Each matrix in $\mathcal{D}_{r,n,m}$ is $n$-sparse.
		\item $\mathcal{D}_{r,n,m}$ can be computed in $\poly(m)$ operations, where
		operations (including a successor function in some enumeration of $\K$) over $\K$
		are counted at unit cost.
	\end{enumerate}
	and $\mathcal{D}_{r,n,m}'$, as defined in Construction~\ref{constr:diag}, has the following
	properties:
	\begin{enumerate}
		\item $\mathcal{D}_{r,n,m}'$ is an improper hitting set for $n\times m$ matrices of
		rank $\le r$ over $\F$.
		\item $\mathcal{D}_{r,n,m}'$ is linearly independent (as vectors in $\K^{nm}$) and
		$\sspan(\mathcal{D}_{r,n,m})=\sspan(\mathcal{D}_{r,n,m}')$
		\item $|\mathcal{D}_{r,n,m}'|=(n+m-r)r$
		\item Each matrix in $\mathcal{D}_{r,n,m}'$ is $n$-sparse.
		\item $\mathcal{D}_{r,n,m}'$ can be computed in $\poly(m)$ operations, where
		operations (including a successor function in some enumeration of $\K$) over $\K$
		are counted at unit cost.
	\end{enumerate}
\end{theorem}
\begin{proof}
	\underline{$|\mathcal{D}_{r,n,m}|=(n+m-1)r$:} This is by definition.

	\underline{Sparsity of $\mathcal{D}_{r,n,m}$:} Each matrix in the hitting set has support in
	some $k$-diagonal, and each diagonal has at most $n$ non-zero entries.

	\underline{$\mathcal{D}_{r,n,m}$ can be computed in $\poly(m)$ operations:}  The details are
	very similar to the proof that $\mathcal{B}_{r,n,m}$ can be computed in $\poly(m)$
	operations, as seen in Theorem~\ref{thm:hitting set for matrices}, so we omit the specifics.

	\sloppy \underline{$\mathcal{D}_{r,n,m}$ is an improper hitting set:} Let $M$ be $n\times m$
	matrix of rank $\le r$ in $\F$. By Theorem~\ref{thm:bivariate poly evaluation}, we see that $M$ is non-zero
	iff one of the polynomials $\{\hat{f}_M(x,g^\ell x)\}_{0\le \ell<r}$ is non-zero.  Recall the
	notation that $\coeff_{x^k}(f)$ denotes the coefficient of $f$ on $x^k$.  Thus,
	$\coeff_{x^{k}}(\hat{f}_M(x,g^\ell x))=\sum_{i+j=k} M_{i,j}g^{\ell j}=\la M,D_{k,\ell}\ra$.  Thus,
	it follows that some $\hat{f}_M(x,g^\ell x)$ is non-zero iff one the inner products $\la
	M,D_{k,\ell}\ra$ is non-zero.  Invoking Theorem~\ref{thm:bivariate poly evaluation}
	completes this claim.

	This can also be seen from the fact
	$\sspan(\mathcal{D}_{r,n,m})=\sspan(\mathcal{B}_{r,n,m})$.  Thus, a for a matrix $M$,
	$M\in\ker\mathcal{D}_{r,n,m}\iff M\in\ker\mathcal{B}_{r,n,m}$.

	\underline{$\sspan(\mathcal{D}_{r,n,m})\supseteq\sspan(\mathcal{B}_{r,n,m})$:} For any $M$
	(not just those of rank $\le r$) we have that $\la M,D_{k,\ell}\ra=\coeff_{x^{k}}(\hat{f}_M(x,g^\ell
	x))$ and $\la M,B_{k,\ell}\ra=\hat{f}_M(\alpha_k,g^\ell\alpha_k)$ and thus $\la
	M,B_{k,\ell}\ra=\sum_{k'=0}^{n+m-2} \alpha_k^{k'}\la M,D_{k',\ell}\ra$.  By taking $M$ for
	each element in some basis, it follows that $B_{k,\ell}=\sum_{k'=0}^{n+m-2}
	\alpha_k^{k'}D_{k',\ell}$

	\underline{$\sspan(\mathcal{D}_{r,n,m})\subseteq\sspan(\mathcal{B}_{r,n,m})$:} Similar to
	the above case, we get that for any $M$, \[\la M,
	D_{k,\ell}\ra=\sum_{k'=0}^{n+m-2}\coeff_{x^{k'}}\left(\prod_{k''\ne
	k}\frac{x-\alpha_{k'}}{\alpha_{k''}-\alpha_{k'}}\right)\la M,B_{k',\ell}\ra\] via Lagrange
	interpolation.  As the coefficients of this linear dependence are independent of $M$ (they
	only depend on the $\alpha_k$), by taking $M$ for each element of some basis it follows that
	the same linear dependence for $D_{k,\ell}$ and $\{B_{k',\ell}\}_{k'}$ exists, giving the
	claim.

	\underline{$\mathcal{D}_{r,n,m}'$ can be computed in $\poly(m)$ operations:} As with
	$\mathcal{D}_{r,n,m}$, these details are omitted.

	\underline{Sparsity of $\mathcal{D}_{r,n,m}'$:} Each matrix in the hitting set has support
	in some $k$-diagonal, and each diagonal has at most $n$ non-zero entries.

	\underline{$\mathcal{D}_{r,n,m}'$ is an improper hitting set:} This follows from showing
	that $\mathcal{D}_{r,n,m}\subseteq\sspan(\mathcal{D}_{r,n,m}')$, as this implies that for a
	matrix $M$, $M\in\ker\mathcal{D}_{r,n,m}\iff M\in\ker\mathcal{D}_{r,n,m}'$. Thus, as
	$\mathcal{D}_{r,n,m}$ is an improper hitting set so is $\mathcal{D}_{r,n,m}'$.

	\underline{$\sspan(\mathcal{D}_{r,n,m})\supseteq\sspan(\mathcal{D}_{r,n,m}')$:} This is
	clear, as $\mathcal{D}_{r,n,m}\supseteq\mathcal{D}_{r,n,m}'$.

	\underline{$\sspan(\mathcal{D}_{r,n,m})\subseteq\sspan(\mathcal{D}_{r,n,m}')$:} Begin by
	observing that $D_{k,\ell}$ is non-zero only on the $k$-diagonal and the $k$-diagonal has
	$\min(k+1,n,(n+m)-(k+1))$ entries.  Further, the $k$-diagonals of the matrices
	$\{D_{k,\ell}\}_{0\le \ell<r}$ form the rows of a $r\times \min(k+1,n,(n+m)-(k+1))$
	Vandermonde matrix.  This Vandermonde matrix is formed by taking powers of $\le n$
	consecutive powers of $g$, which by the order of $g$ are distinct.  It follows that the
	first $\min(r,\min(k+1,n,(n+m)-(k+1)))$ of the $\{D_{k,\ell}\}_{0\le \ell<r}$ form a basis
	for the rest.  As $r\le n$, $\min(r,\min(k+1,n,(n+m)-(k+1)))=\min(r,k+1,(n+m)-(k+1))$, so
	$\{D_{k,\ell}\}_{0\le \ell<\min(r,k+1,(n+m)-(k+1))}$ are a basis for $\{D_{k,\ell}\}_{0\le
	\ell<r}$ (recall that we start indexing from zero).  Ranging over all $k$ shows that the
	claim holds.

	\underline{$\mathcal{D}_{r,n,m}'$ is linearly independent:} Notice that $D_{k,\ell}$ and
	$D_{k',\ell'}$ have disjoint support if $k\ne k'$.  The previous paragraph shows
	$\{D_{k,\ell}\}_{0\le \ell<\min(r,k+1,(n+m)-(k+1))}$ are linearly independent for each
	$k$, and the fact about disjoint support for differing $k$ shows that taking the union over
	$k$ does not introduce any linearly dependencies.

	\underline{$|\mathcal{D}_{r,n,m}'|=(n+m-r)r$:} For $r\le k+1\le (n+m)-r$ we see that
	$r=\min(r,k+1,(n+m)-(k+1))$, so $\mathcal{D}_{r,n,m}'$ offers no savings in this regime.
	For $0\le k<r$, $\mathcal{D}_{r,n,m}'$ takes $r-(k+1)$ fewer matrices then
	$\mathcal{D}_{r,n,m}$.  For $n+m> k+1\ge(n+m)-r$,  $\mathcal{D}_{r,n,m}'$ takes
	$r-((n+m)-(k+1))$ fewer matrices then $\mathcal{D}_{r,n,m}$.  It follows that
	$|\mathcal{D}_{r,n,m}'|=|\mathcal{D}_{r,n,m}|-r(r-1)=(n+m-1)r-r(r-1)=(n+m-r)r$.
\end{proof}

\noindent We sketch another proof of this result in Remark~\ref{rmk:reprovediag}.

The above results also imply that $\rank(\mathcal{B}_{r,n,m})=(n+m-r)r$, which is better than the analysis
given in Theorem~\ref{thm:hitting set for matrices}.  This immediately gives that there are  explicit
$(n+m-r)r$-sized (proper) hitting sets for $n\times m$ matrices of rank $\le r$, as we can (in
$\poly(m)$ steps) find a basis for $\mathcal{B}_{r,n,m}$.  This basis will consist of rank-1 matrices, and
also be the desired hitting set.  However, in the interest of being more explicit, we present the
following construction.

\begin{construction}\label{constr:better hitting for matrices}
	Let $m\ge n\ge r\ge 1$. Let $\K$ be an extension of $\F$ such that $g\in\K$ is of order $\ge
	m$ and $\alpha_0,\ldots,\alpha_{n+m-2}\in\K$ are distinct.  Let $B'_{k,\ell}\in\K^{n\times
	m}$ to be the rank-1 matrix defined by $(B'_{k,\ell})_{i,j}=\alpha_k^i(g^\ell\alpha_k)^j$,
	and let $\mathcal{B}_{r,n,m}'\eqdef\{B_{k,\ell}\}_{0\le \ell<r,0\le k\le (n+m-2)-2\ell}$.
\end{construction}

We now give the analysis for this hitting set.

\begin{theorem}\label{thm:better hitting for matrices}
	Let $m\ge n\ge r\ge 1$.  Then $\mathcal{B}_{r,n,m}'$, as defined in
	Construction~\ref{constr:better hitting for matrices}, has the following properties:
	\begin{enumerate}
		\item $\sspan\mathcal{B}_{r,n,m}'=\sspan\mathcal{B}_{r,n,m}$, where
		$\mathcal{B}_{r,n,m}$ is defined in Construction~\ref{construction:construction of
		hitting set for matrices}.
		\item $\mathcal{B}_{r,n,m}'$ is a hitting set for $n\times m$ matrices of rank $\le
		r$ over $\F$.
		\item $|\mathcal{B}_{r,n,m}'|=(n+m-r)r$
		\item $\mathcal{B}_{r,n,m}'$ is linearly independent (as vectors in $\K^{nm}$) 
		\item $\mathcal{B}_{r,n,m}'$ can be computed in $\poly(m)$ operations, where
		operations (including a successor function in some enumeration of $\K$) over $\K$
		are counted at unit cost.
	\end{enumerate}
\end{theorem}
\begin{proof}
	\underline{$|\mathcal{B}_{r,n,m}|=(n+m-r)r$:} The size is equal to $\sum_{\ell=0}^{r-1}
	((n+m-1)-2\ell)=(n+m-1)r-2\binom{r}{2}=(n+m-r)r$.

	\underline{$\mathcal{B}_{r,n,m}'$ can be computed in $\poly(m)$ operations:}  The details
	are very similar to the proof that $\mathcal{B}_{r,n,m}$ can be computed in $\poly(m)$
	operations, as seen in Theorem~\ref{thm:hitting set for matrices}, so we omit the specifics.

	\underline{$\mathcal{B}_{r,n,m}'$ is an hitting set:} This follows from showing that
	$\mathcal{B}_{r,n,m}\subseteq\sspan(\mathcal{B}_{r,n,m}')$, as this implies that for a
	matrix $M$, $M\in\ker\mathcal{B}_{r,n,m}\iff M\in\ker\mathcal{B}_{r,n,m}'$.  Thus, as
	$\mathcal{B}_{r,n,m}$ is an hitting set so is $\mathcal{B}_{r,n,m}'$.

	\underline{$\sspan\mathcal{B}_{r,n,m}'\subseteq\sspan\mathcal{B}_{r,n,m}$:} This is clear as
	$\mathcal{B}_{r,n,m}'\subseteq\mathcal{B}_{r,n,m}$.

	\sloppy \underline{$\sspan\mathcal{B}_{r,n,m}'\supseteq\sspan\mathcal{B}_{r,n,m}$:} We will
	actually show $\mathcal{D}_{r,n,m}\subseteq\sspan\mathcal{B}_{r,n,m}'$, which by
	Theorem~\ref{thm:diagonal set for matrices} is sufficient. Let $M$ be any matrix (even of
	rank $>r$).  We will show that the inner-products $\la M,\mathcal{B}_{r,n,m}'\ra$ determine
	the inner-products $\la M,\mathcal{D}_{r,n,m}\ra$.  Then we show that this implies the
	claim.

	Recall that the
	inner-product of a matrix $D\in\mathcal{D}_{r,n,m}$ is simply a coefficient
	$\coeff_{x^k}(\hat{f}_M(x,g^ix))$ for some $0\le k\le n+m-2$ and $0\le i<r$.  So to prove the claim
	we will speak of these coefficients determining other such coefficients.
	
	Now observe that for any $k\in\{0,\ldots,r-1\}$, the coefficients $\coeff_{x^k}(\hat{f}_M(x,x))$,
	$\coeff_{x^k}(\hat{f}_M(x,gx))$, $\ldots,\coeff_{x^k}(\hat{f}_M(x,g^{r-1}x))$ are linear combinations of the
	$k+1\le r$ elements in $\{M_{i,j}\}_{i+j=k}$.  Just as in the analysis of
	$\mathcal{D}_{r,n,m}'$ in Theorem~\ref{thm:diagonal set for matrices}, the first $k+1$ of
	these linear combinations are rows of a Vandermonde matrix over distinct numbers, and thus
	these linear combinations span all vectors.  Thus, it follows that the coefficients
	$\{\coeff_{x^k}\hat{f}_M(x,g^ix)\}_{0\le i <k+1}$ determine the coefficients
	$\{\coeff_{x^k}\hat{f}_M(x,g^ix)\}_{0\le i <r}$.

	Similarly, for any $k\in\{(n+m)-(r+1),\ldots,(n+m)-2\}$ the coefficients
	$\{\coeff_{x^k}\hat{f}_M(x,g^ix)\}_{0\le i <(n+m)-(k+1)}$ determine the coefficients
	$\{\coeff_{x^k}\hat{f}_M(x,g^ix)\}_{0\le i <r}$.  We now use these facts in the following claim.

	\begin{claim}
		The coefficients of $\hat{f}_M(x,g^{k+1}x)$ are determined by the coefficients of
		$\hat{f}_M(x,x),\hat{f}_M(x,gx),\ldots,\hat{f}_M(x,g^{k}x)$ and the evaluations of $\hat{f}_M(x,g^{k+1} x)$ to any
		$(n+m-1)-2(k+1)$ distinct points.
	\end{claim}
	\begin{proof}
		By the above reasoning, the coefficients $\coeff_{x^{k'}}(\hat{f}_M(x,g^{k+1}x))$ with
		$k'\in\{0,\ldots,k\}\cup\{(n+m-2)-k,\ldots,(n+m)-2\}$ are already determined by the
		coefficients given.
		
		Now, consider the polynomial \[h(x)\eqdef\frac{\hat{f}_M(x,g^{k+1}x)-\sum_{k'=0}^k
		\coeff_{x^{k'}}(\hat{f}_M(x,g^{k+1}x))x^{k'}-\sum_{k'=(n+m-2)-k}^{n+m-2}
		\coeff_{x^{k'}}(\hat{f}_M(x,g^{k+1}x))x^{k'}} {x^{k+1}}\] By construction, $h$ of degree
		$\le(n+m-2)-2(k+1)$, and evaluation of $h$ is possible given oracle access to
		$\hat{f}_M(x,g^{k+1}x)$ as the relevant coefficients referenced are already determined.

		Thus, it follows that $h$ is determined by interpolation at any $(n+m-1)-2(k+1)$
		distinct points.  Once $h$ is determined, the above equation determines the as yet
		undetermined coefficients of $\hat{f}_M(x,g^{k+1}x)$.
	\end{proof}

	Thus, to determine all of the coefficients of the polynomials $\{\hat{f}_M(x,g^\ell x)\}_{0\le \ell
	<r}$ we first interpolate $\hat{f}_M(x,x)$ at $n+m-1$ distinct points.  The above claim then shows
	how to interpolate $\hat{f}_M(x,gx)$ using $(n+m-1)-2$ evaluations to $\hat{f}_M(x,gx)$, given access to the
	coefficients of $\hat{f}_M(x,x)$.  Inducting on the above claim shows we can interpolate all of the
	coefficients in $\{\hat{f}_M(x,g^\ell x)\}_{0\le \ell <r}$ from the evaluations
	$\{\hat{f}_M(\alpha_k,g^\ell\alpha_k)\}_{0\le \ell <r, 0\le k\le (n+m-2)-2\ell}$.  Rephrasing this,
	we see that the inner-products $\la M,\mathcal{D}_{r,n,m}\ra$ are determined by the
	inner-products $\la M,\mathcal{B}_{r,n,m}'\ra$.

	Now consider a matrix $B\notin \sspan\mathcal{B}_{r,n,m}'$.  It follows that the dual space
	of $\mathcal{B}_{r,n,m}'$ is strictly larger than the dual space of
	$\mathcal{B}_{r,n,m}'\cup\{B\}$, so that there is a non-zero matrix $M_0$ such that $\la
	M_0,\mathcal{B}_{r,n,m}'\ra=\vec{0}$ but $\la M_0,B\ra\ne 0$. But as $\la 0_{n\times
	m},\mathcal{B}_{r,n,m}'\ra=\vec{0}$ and $\la 0_{n\times m},B\ra= 0$, it follows that the
	inner-product $\la M_0, \mathcal{B}_{r,n,m}'\ra$ does not determine the inner-product $\la
	M_0,B\ra$.  As $\la M,\mathcal{B}_{r,n,m}'\ra$ determines $\la M,\mathcal{D}_{r,n,m}\ra$, it
	must be that $\mathcal{D}_{r,n,m}\subseteq\sspan\mathcal{B}_{r,n,m}'$.

	\underline{$\mathcal{B}_{r,n,m}'$ is linearly independent:} As
	$\sspan(\mathcal{B}_{r,n,m}')=\sspan(\mathcal{D}_{r,n,m}')$,
	$|\mathcal{B}_{r,n,m}'|=|\mathcal{D}_{r,n,m}'|$, and $\mathcal{D}_{r,n,m}'$ is linearly
	independent, it follows that $\mathcal{B}_{r,n,m}'$ is also.
\end{proof}

Thus, we achieve an explicit hitting set of size $(n+m-r)r$.  For $r=n$ we see that this equals
$nm$, matching the naive bound.  For $r\le n-1$, $(n+m-r)r$ is increasing with $r$, so $(n+m-r)r\le
(n+m-(n-1))(n-1)=(m+1)(n-1)=nm+n-m-1< nm$.  Thus, we see that our hitting set is always smaller
than the naive hitting set, for $r<n$.

\section{Identity Testing for Tensors}\label{sec:pit for tensors}

In this section we show how to construct hitting sets for $\llb n\rrb^d$ tensors of arbitrary degree $d$.
We will only discuss tensors of shape $\llb n\rrb^d$ for simplicity.  The proof technique will be to use
the results for $d=2$ as a black-box as a way to induct on $d$.  That is,
Corollary~\ref{cor:bivariate poly evaluation} shows that one can test identity of degree $<n$, rank
$\le r$ bivariate polynomials by testing the identity of $r$ univariate polynomials, each of degree
$< 2n$.  This effectively reduces the $d=2$ case to the $d=1$ case, while increasing the number of
polynomials to test by a factor of $r$.  As degree $<2n$ univariate polynomials can be fully
interpolated cheaply, this shows that this is a viable base case for recursion.

Intuitively, it seems like this variable reduction process should be able to be continued so that a rank $\le r$
$d$-variate polynomial can be identity tested by testing identity of $\approx r^d$ univariate
polynomials each of degree $\approx dn$.  This is indeed possible.  However, we are able to do
better here by using a reduction process that reduces a $d$-variate polynomial to a $d/2$-variate
polynomial while only increasing the number of polynomials to test by a factor of $r$.  Thus, a
$d$-variate polynomial can identity tested by testing $\approx r^{\lg d}$ univariate polynomials,
each of degree $<dn$.  Unfortunately, this set of polynomials will require $\approx (dn)^d$ time to
construct.

The section will be split into two parts.  The first will state the variable reduction theorem that
was mentioned above.  The second part will detail the hitting set arising from this theorem.

\subsection{Variable Reduction}\label{sec:dvariate}

As with the $d=2$ case, will need a variable reduction result in order to construct our hitting set.
We detail this result in this subsection.  We first illustrate some lemmas about variable
reduction.

\begin{lemma}\label{lem:permute}
	Let $f(x_1,\ldots,x_d)$ be a $d$-variate polynomial.  Let $\pi:[d]\to[d]$ be a permutation.
	Then, $f(x_1,\ldots,x_d)=0$ iff $f(x_{\sigma(1)},\ldots,x_{\sigma(d)})=0$.
\end{lemma}
\begin{proof}
	Consider the map $\N^{d}\to\N^{d}$ defined by $(i_1,\ldots,i_{d})\mapsto
	(i_\sigma(1),\ldots,i_{\sigma(d)})$. This is exactly the action on the degrees of monomials
	over the variables $x_1,\ldots,x_d$ when performing the substitution $x_i\mapsto
	x_{\sigma(i)}$.  Note that this map is bijective.

	Thus, when mapping $f(x_1,\ldots,x_d)$ to $f(x_{\sigma(1)},\ldots,x_{\sigma(d)})$ we see
	that there can be no cancellations, as distinct monomials are mapped to distinct monomials.
	Thus, the two polynomials have the same number of non-zero coefficients.  In particular,
	they are either both zero or non-zero.
\end{proof}

The above lemma is most useful in conjunction with the next lemma, which shows a simple $d$-variate
to $(d-1)$-variate reduction.

\begin{lemma}\label{lem:reshape}
	Let $f(x,y,z_1,\ldots,z_d)$ be a $(d+2)$-variate polynomial such that $\deg_{x}(f)<n$. Then
	for any $m\ge n$, $f(x,y,z_1,\ldots,z_d)=0$ iff $f(x,x^m,z_1,\ldots,z_d)=0$.
\end{lemma}
\begin{proof}
	Consider the map $\N^{d+2}\to\N^{d+1}$ defined by $(i_1,i_2,i_3,\ldots,i_{d+2})\mapsto
	(i_1+mi_2,i_3,\ldots,i_{d+2})$. This is exactly the action on the degrees of monomials over
	the variables $x,y,z_1,\ldots,z_d$ when performing the substitution $y\mapsto x^m$.

	Notice that this map is injective when restricted to $\llb n\rrb\times\N^{d+1}$, as $n\le
	m$.  That is, if $i+mj=i'+mj'$ with $(i,j),(i',j')\in\llb n\rrb\times\Z$ then $i\equiv
	i'\mod{m}$ which means $i=i'$, and thus $j=j'$ as well.

	Thus, when mapping $f(x,y,z_1,\ldots,z_d)$ to $f(x,x^m,z_1,\ldots,z_d)$ we see that there
	can be no cancellations, as distinct monomials are mapped to distinct monomials.  Thus, the
	two polynomials have the same number of non-zero coefficients.  In particular, they are
	either both zero or non-zero.
\end{proof}

The above lemmas show that we can ``reshape'' our polynomials, in the sense that we have fewer
variables but larger individual degrees.  To perform our $d$-variate variable reduction, we will
reshape our polynomial into a bivariate polynomial, such that the individual degrees are now
$\approx n^{d/2}$.  We can then apply our bivariate variable reduction to get a univariate
polynomial of degree $\approx n^{d/2}$.  One can then reverse the reshaping, to yield a
$d/2$-variate polynomial, with individual degrees $\approx n$.  One then recurses appropriately.

In order to understand the recursion pattern sketched above, we will introduce the following
function.

\begin{definition}
	Let $n\ge 1$, $b\ge 0$.  Let $0\le k<2^d$.  Define
	\[L_{n,b}(k,i_1,\ldots,i_d)=\sum_{\genfrac{}{}{0pt}{}{1\le j\le d}{\lfloor k/2^{j-1}\rfloor
	\equiv 1 \bmod 2}} i_j(n2^b)^{\lfloor k/2^j\rfloor}\]
\end{definition}

We now observe that it obeys the following properties.

\begin{proposition}
	Let $n\ge 1$, $b\ge 0$, with $0\le k<2^d$.  Then
	\begin{enumerate}
		\item \[L_{n,b}(k,i_1,\ldots,i_d)=
			\begin{cases}
				0										&	\text{if }k=0\\
				i_1(n2^b)^{\lfloor k/2\rfloor}+L_{n,b}(\lfloor k/2\rfloor,i_2,\ldots,i_{d})	&	k\equiv 1\bmod 2\\
				L_{n,b}(\lfloor k/2\rfloor,i_2,\ldots,i_{d})					&	\text{else}
			\end{cases}\]\label{prop:lnbprop:recur}
		\item For $b\ge 1$, $L_{2n,b-1}(k,i_1,\ldots,i_d)=L_{n,b}(k,i_1,\ldots,i_d)$\label{prop:lnbprop:slosh}
		\item $L_{n,b}(k,i_1,\ldots,i_d)\le (n2^b)^{\lfloor k/2\rfloor}\sum_{j\in[d]}i_j$\label{prop:lnbprop:bound}
		\item $L_{n,b}(k,i_1,\ldots,i_d)$ can be computed in time
		$\poly(|n|,b,d,k,|i_1|,\ldots,|i_d|)$, where $|\cdot|$ is the length, in bits, of a
		number.\label{prop:lnbprop:compute}
	\end{enumerate}
	\label{prop:lnbprop}
\end{proposition}
\begin{proof}
	\underline{\eqref{prop:lnbprop:recur}:} We first note that $\lfloor\lfloor
	k/2^j\rfloor/2^{j'}\rfloor=\lfloor k/2^{j+j'}\rfloor$, which is most easily seen by
	observing that these operations bit truncate (on the right) the binary representation of
	$k$.  If $k=0$ then in both formulas $L_{n,b}(k,i_1,\ldots,i_d)=0$. If $k\equiv 1\bmod 2$,
	then
	\begin{align*}
		L_{n,b}(k,i_1,\ldots,i_d)
			&=i_1(n2^b)^{\lfloor k/2\rfloor}+\sum_{\genfrac{}{}{0pt}{}{2\le j\le d}{\lfloor k/2^{j-1}\rfloor \equiv 1 \bmod 2}} i_j(n2^b)^{\lfloor k/2^j\rfloor}\\
			&=i_1(n2^b)^{\lfloor k/2\rfloor}+\sum_{\genfrac{}{}{0pt}{}{1\le j\le d-1}{\lfloor k/2^{j-2}\rfloor \equiv 1 \bmod 2}} i_{j+1}(n2^b)^{\lfloor k/2^{j-1}\rfloor}\\
			&=i_1(n2^b)^{\lfloor k/2\rfloor}+\sum_{\genfrac{}{}{0pt}{}{1\le j\le d-1}{\lfloor \lfloor k/2\rfloor /2^{j-1}\rfloor \equiv 1 \bmod 2}} i_{j+1}(n2^b)^{\lfloor \lfloor k/2\rfloor/2^{j}\rfloor}\\
			&=i_1(n2^b)^{\lfloor k/2\rfloor}+L_{n,b}(\lfloor k/2\rfloor,i_2,\ldots,i_{d})
	\end{align*}
	which is exactly the above recursion.  The case $k\equiv 0\bmod 2$ is analogous.

	\underline{\eqref{prop:lnbprop:slosh}:} The definition of $L_{n,b}$ only depends on $n2^b$.
	Thus, as $2n\cdot 2^{b-1}=n\cdot 2^b$, this is immediate.

	\underline{\eqref{prop:lnbprop:bound}:} This is immediate.

	\underline{\eqref{prop:lnbprop:compute}:} The natural way of computing the formula
	$L_{n,b}(k,i_1,\ldots,i_d)$ is done in the given time bound.
\end{proof}

We will now prove our multi-variate variable reduction theorem.  We prove here the case when the
number of variables is a power of 2, for simplicity. The general case, with some loss, will follow
as a corollary.  The following notation will make the presentation simpler.

\begin{notation}
	Let $f(\la h_1(j),\ldots,h_k(j)\ra_{j=1}^r)$ denote
	\[f(h_1(1),\ldots,h_k(1),h_1(2),\ldots,h_k(2),\ldots,h_1(r),\ldots,h_k(r))\]
\end{notation}

We will use this notation heavily in the following proof.

\begin{theorem}\label{thm:d-variate reduction}
	Let $n\ge 1$, $d\ge 1$ and $b\ge d-1$.  Let $\K$ be an extension of $\F$ such that $g\in\K$
	has order $\ge (n2^b)^{2^{d-1}}$. Let $T:\llb n\rrb^{2^d}\to\F$ be a tensor of rank $\le r$.  Let
	$\hat{f}_T(x_0,\ldots,x_{2^d-1})=\sum_{\ell=1}^r\prod_{i=0}^{2^d-1} p_{i,\ell}(x_i)$, where $\deg
	p_{i,\ell}<n$.

	Then $\hat{f}_T$ is non-zero (over $\F$) iff one of the univariate polynomials in the set
	\[\{\hat{f}_T(g^{L_{n,b}(0,i_1,\ldots,i_d)}x,g^{L_{n,b}(1,i_1,\ldots,i_d)}x,\ldots,g^{L_{n,b}(2^d-1,i_1,\ldots,i_d)}x)\}_{0\le
	i_1,\ldots,i_{d}<r}\] is non-zero (over $\K$).
\end{theorem}
\begin{proof}
	The proof will be by induction.  For simplicity we write $f$ for $\hat{f}_T$.

	\underline{$d=1$:}  Note that $L_{n,b}(0,i_1)=0$ and $L_{n,b}(1,i_1)=i_1$, so this case
	follows from Corollary~\ref{cor:bivariate poly evaluation}.

	\underline{$d>1$:}  We will first reshape $f$ into a bivariate polynomial, and appeal to the
	$d=1$ case.  We will then un-reshape this polynomial into a $2^{d-1}$-variate polynomial,
	and then appeal to induction.

	By induction on Lemma~\ref{lem:reshape} (and appealing to Lemma~\ref{lem:permute} to see
	that Lemma~\ref{lem:reshape} applies to any two variables, not just the first) we see that
	\begin{equation}\label{eqn:reshape}f(\la x_j\ra_{j=0}^{2^{d}-1})=0 \text{ iff } f(\la
	x_0^{(n2^b)^j},x_1^{(n2^b)^j}\ra_{j=0}^{2^{d-1}-1})=0\end{equation} (where so far we only
	need that $b\ge 1$).

	We split the rest of the proof into two claims.  The first claim shows how we can, using the
	bivariate case, test identity of the right-hand-side of Equation~\eqref{eqn:reshape} by
	testing identity of a set of $r$ polynomials, each of $2^{d-1}$ variables.  The second claim
	shows how testing identity of these new polynomials can be reduced to testing identity of
	univariate polynomials, where we use the induction hypothesis.

	\begin{claim}
		\[f(\la x_0^{(n2^b)^j},x_1^{(n2^b)^j}\ra_{j=0}^{2^{d-1}-1})=0\]
		iff
		\[\{f(\la x_j,g^{i_1(n2^b)^j}x_j\ra_{j=0}^{2^{d-1}-1})\}_{0\le i_1<r}=0\]
	\end{claim}
	\begin{proof}
		First observe that
		\begin{align*}
			f'(x_0,x_1)
				&\eqdef f(\la x_0^{(n2^b)^j},x_1^{(n2^b)^j}\ra_{j=0}^{2^{d-1}-1})\\
				&= f(x_0,x_1,x_0^{n2^b},x_1^{n2^b},x_0^{(n2^b)^2},x_1^{(n2^b)^2},\ldots,x_0^{(n2^b)^{2^{d-1}-1}},x_1^{(n2^b)^{2^{d-1}-1}})\\
				&=\sum_{\ell=1}^r \left(\prod_{j=0}^{2^{d-1}-1} p_{2j,\ell}(x_0^{(n2^b)^j})\right)\left(\prod_{j=0}^{2^{d-1}-1} p_{2j+1,\ell}(x_1^{(n2^b)^j})\right)
		\end{align*}
		so we can apply Corollary~\ref{cor:bivariate poly evaluation} to see that
		$f'(x_0,x_1)=0$ iff $\{f'(x_0,g^{i_1}x_0)\}_{0\le i_1<r}=0$, which, when expanded,
		is equivalent to \[\{f(\la
		x_0^{(n2^b)^j},g^{i_1(n2^b)^j}x_0^{(n2^b)^j}\ra_{j=0}^{2^{d-1}-1})\}_{0\le
		i_1<r}=0\] using that the order of $g$ is $\ge (n2^b)^{2^{d-1}}>
		\deg_{x_0}f',\deg_{x_1}f'$. Using that $2^b\ge 2^{d-1}\ge2$, we can undue the
		variable substitutions $x_j\mapsto x_0^{(n2^b)^j}$.  That is, applying
		Lemma~\ref{lem:reshape} in reverse, we see that the above set of polynomials is zero
		iff \[\{f(\la x_j,g^{i_1(n2^b)^j}x_j\ra_{j=0}^{2^{d-1}-1})\}_{0\le i_1<r}=0\] which
		is exactly the claim.
	\end{proof}

	\begin{claim}
		\[f(\la x_j,g^{i_1(n2^b)^j}x_j\ra_{j=0}^{2^{d-1}-1})=0\]
		iff
		\[\{f(\la g^{L_{n,b}(j,i_1,\ldots,i_d)}x\ra_{j=0}^{2^d-1})\}_{0\le i_2,\ldots,i_{d}<r}=0\]
	\end{claim}
	\begin{proof}
		First observe that
		\begin{align*}
			f'(x_0,x_1,\ldots,x_{2^{d-1}-1})
			&\eqdef f(\la x_j,g^{i_1(n2^b)^j}x_j\ra_{j=0}^{2^{d-1}-1})\\
			&=\sum_{\ell=1}^r\prod_{j=0}^{2^{d-1}-1}p_{2j,\ell}(x_{2j})\cdot p_{2j+1,\ell}(g^{i_1(n2^b)^j}x_{2j})
		\end{align*}
		so $f'$ is $2^{d-1}$-variate, having individual degrees $<2n$.  Thus, applying
		induction to the theorem for the $2^{d-1}$-variate case (and using $b-1$ instead of
		$b$, noticing that $b-1\ge (d-1)-1$ also holds), we get that
		$f'(x_0,x_1,\ldots,x_{2^{d-1}-1})=0$ iff \[\{f'(\la
		g^{L_{2n,b-1}(j,i_2,\ldots,i_{d})}x\ra_{j=0}^{2^{d-1}-1})\}_{0\le
		i_2,\ldots,i_{d}<r}=0\] or in terms of $f$, \[\{f(\la
		g^{L_{2n,b-1}(j,i_2,\ldots,i_{d})}x,g^{i_1(n2^b)^j+L_{2n,b-1}(j,i_2,\ldots,i_{d})}x\ra_{j=0}^{2^{d-1}-1})\}_{0\le
		i_2,\ldots,i_{d}<r}=0\] where we have used that the order of $g\ge
		(n2^b)^{2^{d-1}}\ge (2n\cdot 2^{b-1})^{2^{(d-1)-1}}$. Invoking
		Proposition~\ref{prop:lnbprop}.\eqref{prop:lnbprop:slosh} and
		Proposition~\ref{prop:lnbprop}.\eqref{prop:lnbprop:recur} we see that the above
		polynomials being zero is equivalent to \[\{f(\la
		g^{L_{n,b}(2j,i_1,\ldots,i_{d})}x,g^{L_{n,b}(2j+1,i_1,\ldots,i_{d})}x\ra_{j=0}^{2^{d-1}-1})\}_{0\le
		i_2,\ldots,i_{d}<r}=0\] and reindexing, this is equivalent to \[\{f(\la
		g^{L_{n,b}(j,i_1,\ldots,i_{d})}x\ra_{j=0}^{2^{d}-1})\}_{0\le i_2,\ldots,i_{d}<r}=0\]
		which is the claim.
	\end{proof}

	Chaining together Equation~\ref{eqn:reshape} and the above two claims, yields the theorem.
\end{proof}

\begin{remark}
	\label{remark:sqrt trick}
	Let $D=2^d$. In the above proof we use a recursion scheme that reduces to the problem when
	$D\to 2$ and $D\to D/2$.  This gives rise to the recursion $T(D)\le T(2)+T(D/2)$, where
	$T(D)$ is the minimum number such that a $D$-variate rank $\le r$ polynomial can be identity
	tested using $r^{T(D)}$ univariate polynomials. There is also the recursion $S(D)\le
	r(Dn)^{D/2}+S(D/2)$, where $S(D)$ is the maximum degree of $g$ seen in this reduction to the
	univariate case.

	One can do slightly better than this scheme by using the ``square root trick'', where we
	break up the $D$-variate case into two copies of the $\sqrt{D}$-variate case.  This yields
	the recursions $T(D)\le 2T(\sqrt{D})$ and $S(D)\le r(Dn)^{\sqrt{D}}\cdot
	S(\sqrt{D})+S(\sqrt{D})$.  This yields the same solution to $T$, but has now that
	$S(D)=\O(r(Dn)^{\O(\sqrt{D})})$ instead of $r(Dn)^{D/2}$.  While this is an improvement, it
	is somewhat mild.

	Similarly, one can give other recursion schemes that minimize $S$ (so it is $\poly(n,D,R)$),
	but at the cost of making $T(D)\approx D$.
\end{remark}

\subsection{The Hitting Set for Tensors}\label{sec:hittensor}

In this subsection we use the variable reduction theorem of the last subsection to construct hitting
sets for tensors.  First, recall our notion of a hitting set for tensors from
Section~\ref{sec:prelim}, as well as the definitions of the polynomial $f_T$ and $\hat{f}_T$
associated with $T$.  As $\coeff_{x_1^{i_1}\cdots x_d^{i_d}}(\hat{f}_T)=T(i_1,\ldots,i_d)$ we see
that $T=0$ iff $\hat{f}_T=0$.  Theorem~\ref{thm:d-variate reduction} shows that $\hat{f}_T=0$ iff a
set of univariate polynomials are all zero.  Thus, to test if $T$ is zero we can interpolate each of
these polynomials.  As these polynomials are defined via $\hat{f}_T$, these interpolations can be
realized as inner-products with $T$.  This will yield our hitting set, which we now make formal.

\begin{construction}\label{construction:construction of hitting set for tensors}
	Let $n,r\ge 1$ and $d\ge 2$. Let $\K$ be an extension of $\F$ such that $g\in\K$ is of order
	$\ge (2dn)^d$ and $\alpha_1,\ldots,\alpha_{dn}\in\K$ are distinct.  Let
	$B_{k,\ell_1,\ldots,\ell_{\lceil \lg d\rceil}}:\llb n\rrb^d\to\K$ to be the rank-1 tensor
	defined by \[B_{k,\ell_1,\ldots,\ell_{\lceil \lg d\rceil}}(i_1,\ldots,i_d)\eqdef
	\prod_{j=1}^d (g^{L_{n,\lceil\lg d\rceil}(j,\ell_1,\ldots,\ell_{\lceil \lg
	d\rceil})}\alpha_k)^{i_j}\] and let
	$\mathcal{B}_{d,n,r}\eqdef\{B_{k,\ell_1,\ldots,\ell_{\lceil \lg d\rceil}}\}_{0\le
	\ell_1,\ldots,\ell_{\lceil \lg d\rceil}<r,1\le k\le dn}$.
\end{construction}

We now give the analysis for this hitting set.

\begin{theorem}\label{thm:hitting set for tensors}
	Let $n,r\ge 1$ and $d\ge 2$.  Then $\mathcal{B}_{d,n,r}$, as defined in
	Construction~\ref{construction:construction of hitting set for tensors}, has the following
	properties:
	\begin{enumerate}
		\item $\mathcal{B}_{d,n,r}$ is a hitting set for $\llb n\rrb^d$ tensors of rank $\le
		r$ over $\F$.
		\item $|\mathcal{B}_{d,n,r}|=dnr^{\lceil \lg d\rceil}$
		\item $\mathcal{B}_{d,n,r}$ can be computed in $\poly((2dn)^d,r^{\lceil \lg
		d\rceil})$ operations, where operations (including a successor function in some
		enumeration of $\K$) over $\K$ are counted at unit cost.
	\end{enumerate}
\end{theorem}
\begin{proof}
	\underline{$|\mathcal{B}_{d,n,r}|=dnr^{\lceil \lg d\rceil}$:} This is by definition.

	\underline{$\mathcal{B}_{d,n,r}$ can be computed in $\poly((2dn)^d,r^{\lceil \lg d\rceil})$
	operations:}  We assume here an enumeration of elements in $\K$ such that the successor in
	this enumeration can be computed at unit cost.  We also will assume testing whether an
	element is zero, as well as the field elements, are done at unit cost.

	First observe that there are at most $(2dn)^d$ solutions to $x^{(2dn)^d}-1$ over $\K$, so if
	we enumerate $(2dn)^d+1$ elements of $\K$, they we can find a $g\in\K$ with order $\ge
	(2dn)^d$.  This is in $\poly((2dn)^d)$ operations. Similarly, the enumeration will give us
	$dn$ distinct elements which yield the desired $\alpha_k$.

	By Proposition~\ref{prop:lnbprop}, $L_{n,\lceil\lg d\rceil}(j,\ell_1,\ldots,\ell_{\lceil \lg
	d\rceil})$ can be computed in $\poly(d,n,r)$ steps, and this number is $\le (2dn)^d$, so
	computing $g^{L_{n,\lceil\lg d\rceil}(j,\ell_1,\ldots,\ell_{\lceil \lg d\rceil})}$ will take
	at most $\poly((2dn)^d,r)$ operations.  Computing the powers of $\alpha_k$ will take
	$\poly(d,r)$ time.  Thus, each $B_{k,\ell_1,\ldots,\ell_{\lceil \lg d\rceil}}$ can be done
	in $\poly((2dn)^d,r^{\lceil \lg d\rceil})$ steps.  As there are $\poly(dnr^{\lceil \lg
	d\rceil})$ of them, all of $\mathcal{B}_{r,n,m}$ can be computed in $\poly((2dn)^d,r^{\lceil
	\lg d\rceil})$ operations.

	\underline{$\mathcal{B}_{d,n,r}$ is a hitting set:} By construction $\mathcal{B}_{d,n,r}$ is
	a set of rank-1 tensors, so it remains to show that it hits each low-rank tensor.  Consider
	any $T:\llb n\rrb^d\to\F$ of rank $\le r$.	We now apply Theorem~\ref{thm:d-variate
	reduction} to $\hat{f}_T$, where we consider $\hat{f}_T$ as a $2^{\lceil \lg
	d\rceil}$-variate polynomial of rank $\le r$ (by padding $\hat{f}_T$ with dummy variables),
	individual degrees $<n$, and taking $b=\lceil \lg d\rceil$.  This shows that $\hat{f}_T=0$
	iff \[\{\hat{f}_T(g^{L_{n,\lceil \lg d\rceil}(0,\ell_1,\ldots,\ell_{\lceil \lg
	d\rceil})}x,g^{L_{n,{\lceil \lg d\rceil}}(1,\ell_1,\ldots,\ell_{\lceil \lg
	d\rceil})}x,\ldots,g^{L_{n,{\lceil \lg d\rceil}}(d-1,\ell_1,\ldots,\ell_{\lceil \lg
	d\rceil})}x)\}_{0\le \ell_1,\ldots,\ell_{\lceil \lg d\rceil}<r}=0\] (over $\K$). Each of the
	above univariate polynomials has degree $\le d(n-1)$, so interpolating them at $dn\ge
	d(n-1)+1$ points will completely determine them.  In particular, the above polynomials are
	zero iff all the evaluations at any $dn$ are zero.

	Now we observe, just as in the matrix case, that evaluating the $(\ell_1,\ldots,\ell_{\lceil
	\lg d\rceil})$-th polynomial in the above set at the point $\alpha_k$ is exactly the same as
	the inner product $\la T, B_{k,\ell_1,\ldots,\ell_{\lceil \lg d\rceil}}\ra$.  Thus, $T=0$ iff
	$\hat{f}_T=0$ iff all of these inner-products is zero.  This exactly means that
	$\mathcal{B}_{d,n,r}$ is a hitting set.
\end{proof}

We remark that this hitting set is of quasi-polynomial size as a rank $\le r$ tensor $T:\llb
n\rrb^d\to F$ can be represented using $dnr$ field elements.  However, its construction time is
exponential in $d$.  We leave it as an open question as to whether the construction time can be made
to match (up to polynomial factors) the size of the hitting set.

\subsection{Identity Testing for Tensors over Small Fields}\label{sec:pit for tensors over small fields}

Thus far we have assumed the existence of an element $g\in\K$ of large order.  In doing so, all of
our hitting sets are tensors over the field $\K$ instead of the base field $\F$.  While this is a
common assumption when the polynomials of interest are of high degree, the polynomials arising from
$\llb n\rrb^d$ tensors on $dn$ variables are of degree $\le d$, so hitting sets still exist for when
$\F$ is $\O(d)$ sized (as seen in Lemma~\ref{lem:szhit}). In this section, we explore this question
and  show how to transform hitting sets over $\K$ to hitting sets over $\F$, with some loss.
Combining this with the above results, we construct explicit hitting sets over any $\F$.

We first detail a field simulation result that produces improper hitting sets.

\begin{proposition}\label{prop:simfieldimproper}
	Let $\K$ be an extension of $\F$, with $k=\dim_\F\K$. For $\ell\in\llb k\rrb$, let
	$\varphi_\ell:\K\to\F^k$ denote the $k$ projection maps to the standard basis coordinates of
	$\K$.

	Let $\mathcal{H}\subseteq \K^{\llb n\rrb^d}$ be an improper hitting-set for $\llb n\rrb^d$
	tensors of rank $\le r$.  For $H\in\mathcal{H}$ define $\tilde{H}_\ell$ by
	\[(\tilde{H}_\ell)_{i_1,\ldots,i_d}=\varphi_\ell(H_{i_1,\ldots,i_d})\] and define
	\[\tilde{\mathcal{H}}=\{\tilde{H}_{\ell}\}_{H\in\mathcal{H},\ell\in\llb k\rrb}\]
	Then
	\begin{enumerate}
		\item If all tensors in $\mathcal{H}$ are $s$-sparse, then so are all tensors in
		$\tilde{\mathcal{H}}$.\label{prop:simfieldimproper:sparse}
		\item $|\tilde{\mathcal{H}}|=k\cdot |\mathcal{H}|$.\label{prop:simfieldimproper:size}
		\item $\tilde{\mathcal{H}}$ is an improper hitting set for $\llb n\rrb^d$ tensors of
		rank $\le r$.\label{prop:simfieldimproper:hit}
	\end{enumerate}
\end{proposition}
\begin{proof}
	\underline{\eqref{prop:simfieldimproper:sparse}:} If $H_{i_1,\ldots,i_d}=0$ then it follows
	that $(\tilde{H}_\ell)_{i_1,\ldots,i_d}=0$ for all $\ell$.

	\underline{\eqref{prop:simfieldimproper:size}:} This is by construction.

	\underline{\eqref{prop:simfieldimproper:hit}:}  Let $\alpha_0,\ldots,\alpha_{k-1}$ be the
	standard basis for $\K$ as a $\F$-vector-space.  Then it follows that $H=\sum_{\ell\in\llb
	k\rrb} H_\ell\ \alpha_\ell$.

	Consider some tensor $T:\llb n\rrb^d\to\F$ of rank $\le r$.  Then we know that there is some
	$H\in\mathcal{H}$ such that $\la T,H\ra\ne 0$. It follows that there must be some $\ell$
	with $\la T,H_\ell\ra\ne 0$.
\end{proof}

We now apply this to our hitting set results.

\begin{corollary}\label{cor:hitsmallmatriximproper}
	Let $m\ge n\ge r\ge 1$.  Over any field $\F$, there is an $\poly(m)$-explicit improper
	hitting set for $n\times m$ matrices of rank $\le r$, of size $\O(rm\lg m)$.  Further, each
	matrix in the hitting set is $\O(n)$-sparse.
\end{corollary}
\begin{proof}
	If $\F$ has an element of order $\ge m$, then Theorem~\ref{thm:diagonal set for matrices}
	suffices.

	If not, let $\K$ be an extension field of $\F$ such that $\dim_\F\K=\Theta(\lg m)$, and thus
	there is an element of order $\ge m$ in $\K$.  Such an extension can be explicitly described
	by an irreducible polynomial over $\F$ of degree $\Theta(\lg m)$, which can found in
	$\poly(m)$ time, in which time we can also find $g$.  Using Theorem~\ref{thm:diagonal set
	for matrices} to get an $n$-sparse (improper) hitting-set over $\K$ for these $\F$-matrices,
	and applying Proposition~\ref{prop:simfieldimproper} yields the result.
\end{proof}

\begin{corollary}\label{cor:hitsmalltensorimproper}
	Let $n,r\ge1$, $d\ge 2$.  Over any field $\F$, there is an $\poly((2nd)^d,r^{\O(\lg
	d)})$-explicit improper hitting set for $\llb n\rrb^d$-tensors of rank $\le r$, of size
	$\O(dnr^{\O(\lg d)}\cdot (d\lg2dn))$.
\end{corollary}
\begin{proof}
	If $\F$ has an element of order $\ge (2nd)^d$, then Theorem~\ref{thm:hitting set for
	tensors} suffices.

	If not, let $\K$ be an extension field of $\F$ such that $\dim_\F\K=\Theta(d\lg (2nd))$, and
	thus there is an element of order $\ge (2nd)^d$ in $\K$.  Such an extension can be
	explicitly described by an irreducible polynomial over $\F$ of degree $\Theta(d\lg 2nd)$,
	which can found in $\poly((2nd)^d)$ time, in which time we can also find $g$.  Using
	Theorem~\ref{thm:hitting set for tensors} to get a hitting-set over $\K$ for these
	$\F$-matrices, and applying Proposition~\ref{prop:simfieldimproper} yields the result.
\end{proof}

The above results only yield improper hitting sets.  We now show how to preserve the rank-1 property
of the original hitting set, and thus get proper hitting sets over small fields. To do,
we first recall a standard fact in algebra showing that $\K$ is isomorphic to a subring of
$\F$-matrices.

\begin{lemma}
	Let $\K$ be an extension of $\F$, and let $k=\dim_\mathbb{F}\mathbb{K}<\infty$ so that
	$\K=\F^k$ as vector spaces.  For any $\alpha\in\K$ define the linear map
	$\mu_\alpha:\F^k\to\F^k$ given by the multiplication map $x\mapsto \alpha x$.  Let
	$M_\alpha\in\F^{k\times k}$ be the associated matrix.  Then the map
	$M_{(\cdot)}:\K\to\F^{k\times k}$ is an isomorphism as $\F$-algebras.
	\label{lem:multmap}
\end{lemma}
\begin{proof}
	The map is clearly well-defined.  To see the additive homomorphism, note that as
	$(\alpha+\beta)\gamma=\alpha\gamma+\beta\gamma$ for any $\alpha,\beta,\gamma\in\K$, it
	follows that $M_{\alpha+\beta}\cdot \vec{\gamma}=M_{\alpha}\gamma+M_{\beta}\vec{\gamma}$ for
	any $\vec{\gamma}\in\F^k=\K$ (where we abuse notation by writing $\gamma$ to denote an
	element in $\K$ as well as its representation as a vector in $\F^k$).  Taking $\gamma$ for
	each vector in some basis shows that $M_{\alpha+\beta}=M_{\alpha}+M_{\beta}$.

	Similarly, to see the multiplicative homomorphism note that for any
	$\alpha,\beta,\gamma\in\K$ we have that $(\alpha\beta)\gamma=\alpha(\beta\gamma)$.  Thus it
	must be that $M_\alpha M_\beta \gamma=\alpha\beta\gamma=M_{\alpha\beta}\gamma$.  Again,
	taking $\gamma$ over each vector in a basis determines a linear operator.  Thus it must be
	that $M_\alpha M_\beta=M_{\alpha\beta}$.

	Noting that for $\alpha\in\F$ we have that $M_\alpha=\alpha I_k$ we then gain $\F$-linearity
	of the map.

	If $\alpha\ne 0$ then $M_\alpha\cdot M_{\alpha^{-1}}=M_1=I_k$, so $M_\alpha$ is invertible.
	Thus, if $M_\alpha=M_\beta$ then $M_{\alpha-\beta}=0_k$, which implies that $\alpha-\beta=0$
	(as else $M_{\alpha-\beta}$ would be invertible) and thus $\alpha=\beta$.  This implies the
	map is injective.

	As a map is surjective onto its image by definition, this establishes the $\F$-algebra
	homomorphism.
\end{proof}

We now show how to use this alternate representation of $\K$ as a way to simulate hitting sets
defined over $\K$ by hitting sets defined over $\F$.

\begin{proposition}\label{prop:simfield}
	Let $\K$ be an extension of $\F$, with $k=\dim_\F\K$. Let $\mathcal{H}\subseteq \K^{\llb
	n\rrb^d}$ be a hitting-set for $\llb n\rrb^d$ tensors of rank $\le r$. For
	$H=\otimes_{j=1}^d\vec{v}_j\in\K^{\llb n\rrb ^d}$ define
	$\tilde{\vec{v}}_{j,\ell_0,\ldots,\ell_d}\in\F^n$ by
	\[(\tilde{\vec{v}}_{j,\ell_0,\ldots,\ell_d})_i=(M_{(\vec{v}_j)_i})_{\ell_{j-1},\ell_j}\] where
	$M_{(\cdot)}:\K\to\F^{k\times k}$ is the isomorphism of Lemma~\ref{lem:multmap} and define
	\[\tilde{H}_{\ell_0,\ldots,\ell_d}=\bigotimes_{j=1}^d\tilde{\vec{v}}_{j,\ell_0,\ldots,\ell_d}\]
	and define
	\[\tilde{\mathcal{H}}=\{\tilde{H}_{\ell_0,\ldots,\ell_{d-1},0}\}_{H\in\mathcal{H},0\le
	\ell_0,\ldots,\ell_{d-1}<k}\] Then
	\begin{enumerate}
		\item $\tilde{\mathcal{H}}$ is a set of rank-1 $\F$-tensors of shape $\llb
		n\rrb^d$.\label{prop:simfield:rank1}
		\item $|\tilde{\mathcal{H}}|=k^d\cdot |\mathcal{H}|$.\label{prop:simfield:size}
		\item $\tilde{\mathcal{H}}$ is a hitting set for $\llb n\rrb^d$ tensors of rank $\le
		r$.\label{prop:simfield:hit}
	\end{enumerate}
\end{proposition}
\begin{proof}
	\underline{\eqref{prop:simfield:rank1}:} This is by construction.

	\underline{\eqref{prop:simfield:size}:} This is by construction.

	\underline{\eqref{prop:simfield:hit}:}  Consider some tensor $T:\llb n\rrb^d\to\F$ of rank
	$\le r$.  Then we know that there is some $H\in\mathcal{H}$ with $H=\otimes_{j=1}^d
	\vec{v}_j$, such that $\la T,H\ra\ne 0$. Then we see that (we now abuse notation, by writing
	$\mu$ now to denote the map $M_{(\cdot)}$)
	\begin{align*}
		\mu(\la T,H\ra)_{\ell_0,\ell_d}
			&=\mu\left(\sum_{i_1,\ldots,i_d\in\llb n\rrb} T(i_1,\ldots,i_d) \prod_{j=1}^d (\vec{v}_j)_{i_j}\right)_{\ell_0,\ell_d}\\
			&=\sum_{i_1,\ldots,i_d\in\llb n\rrb} T(i_1,\ldots,i_d) \left(\prod_{j=1}^d \mu\left((\vec{v}_j)_{i_j}\right)\right)_{\ell_0,\ell_d}\\
			\intertext{fully expanding the matrix multiplication of $d$ matrices, each $k\times k$,}
			&=\sum_{i_1,\ldots,i_d\in\llb n\rrb} T(i_1,\ldots,i_d) \sum_{\ell_1,\ell_1,\ldots,\ell_{d-1}\in\llb k\rrb}\prod_{j=1}^d\mu\left((\vec{v}_j)_{i_j}\right)_{\ell_{j-1},\ell_j}\\
			&=\sum_{\ell_1,\ell_1,\ldots,\ell_{d-1}\in\llb k\rrb}\sum_{i_1,\ldots,i_d\in\llb n\rrb} T(i_1,\ldots,i_d) \prod_{j=1}^d\mu\left((\vec{v}_j)_{i_j}\right)_{\ell_{j-1},\ell_j}\\
			&=\sum_{\ell_1,\ell_1,\ldots,\ell_{d-1}\in\llb k\rrb}\la T,\tilde{H}_{\ell_0,\ldots,\ell_d}\ra
	\end{align*}
	So it follows that if $\mu(\la T,H\ra)_{\ell_0,\ell_d}\ne 0$ then there is some
	$\ell_1,\ldots,\ell_{d-1}\in\llb k\rrb$ such that $\la
	T,\tilde{H}_{\ell_0,\ldots,\ell_d}\ra\ne 0$.

	Let $\gamma_0$ denote the element in $\K$ corresponding to $\vec{e}_0\in\F^k$ (the standard
	basis vector with a 1\ in the zero position).  Note that $\gamma_0\ne 0$.  Then it follows
	that for any $\alpha\in\K$ that $M_\alpha\vec{e}_0=M_\alpha\gamma_0=\alpha\gamma_0$ (where
	we abuse notation by writing $\alpha\gamma_0$ to denote an element in $\K$ as well as the
	vector representing $\alpha\gamma_0$ in $\F^k$).  Thus, $\alpha$ is fully recoverable from
	$M_\alpha\vec{e}_0$, and in particular, $\alpha= 0$ iff $M_\alpha\vec{e}_0=0$.

	Thus, to test if $\la T,H\ra=0$ (over $\K$) it is enough to test if $\mu(\la
	T,H\ra)_{\ell_0,0}=0$ (over $\F$) for all $\ell_0\in\llb k\rrb$.  Combining this with the
	above we see that $\la T,\mathcal{H}\ra=\vec{0}$ (over $\K$) iff $\la
	T,\tilde{\mathcal{H}}\ra=\vec{0}$.
\end{proof}

We now use the above result to get hitting sets for matrices and tensors over any field.

\begin{corollary}\label{cor:hitsmallmatrix}
	Let $m\ge n\ge r\ge 1$.  Over any field $\F$, there is an $\poly(m)$-explicit hitting set
	for $n\times m$ matrices of rank $\le r$, of size $\O(rm\lg^2m)$.
\end{corollary}
\begin{proof}
	If $\F$ has an element of order $\ge m$, then Theorem~\ref{thm:hitting set for matrices}
	suffices.

	If not, let $\K$ be an extension field of $\F$ such that $\dim_\F\K=\Theta(\lg m)$, and thus
	there is an element of order $\ge m$ in $\K$.  Such an extension can be explicitly described
	by an irreducible polynomial over $\F$ of degree $\Theta(\lg m)$, which can found in
	$\poly(m)$ time, in which time we can also find $g$.  Using Theorem~\ref{thm:hitting set for
	matrices} to get a hitting-set over $\K$ for these $\F$-matrices, and applying
	Proposition~\ref{prop:simfield} yields the result.
\end{proof}

\begin{corollary}\label{cor:hitsmalltensor}
	Let $n,r\ge1$, $d\ge 2$.  Over any field $\F$, there is an $\poly((2nd)^d,r^{\O(\lg
	d)})$-explicit hitting set for $\llb n\rrb^d$-tensors of rank $\le r$, of size
	$\O(dnr^{\O(\lg d)}(d\lg2dn)^d)$.
\end{corollary}
\begin{proof}
	If $\F$ has an element of order $\ge (2nd)^d$, then Theorem~\ref{thm:hitting set for
	tensors} suffices.

	If not, let $\K$ be an extension field of $\F$ such that $\dim_\F\K=\Theta(d\lg (2nd))$, and
	thus there is an element of order $\ge (2nd)^d$ in $\K$.  Such an extension can be
	explicitly described by an irreducible polynomial over $\F$ of degree $\Theta(d\lg 2nd)$,
	which can found in $\poly((2nd)^d)$ time, in which time we can also find $g$.  Using
	Theorem~\ref{thm:hitting set for tensors} to get a hitting-set over $\K$ for these
	$\F$-matrices, and applying Proposition~\ref{prop:simfield} yields the result.
\end{proof}

\section{Explicit Low Rank Recovery of Matrices}\label{sec:explicit LRR}

Thus far we have discussed identity testing for matrices (and tensors).  There the main concern is to
(deterministically) determine whether the matrix is identically zero.  However, we may also ask for
more, in that we may want to (deterministically) reconstruct the entire matrix.  Throughout this
section we will only discuss deterministic measurements which are linear (so are inner products with
the unknown matrix or vector), non-adaptive (so the measurements are independent of the unknown
matrix or vector) and noiseless. The focus on deterministic measurements differs from prior work,
which typically focuses on showing that certain distributions of measurements allow recovery with
high probability.  That the measurements are restricted to be linear is a common assumption in
compressed sensing.  Non-adaptiveness is also a common assumption, but it is important to note that
recent work~\cite{sradaptive} shows that adaptivity in (noisy) sparse-recovery can be more powerful
than non-adaptivity.  Finally, we assume our matrices are \textit{exactly} rank $\le r$, not just
close to some matrix that is rank $\le r$, and we assume that our measurements are noiseless.  This
is not quite practical for compressed sensing, but some previous work also makes this
assumption~\cite{GabidulinKorzhik72,Gabidulin85a,Gabidulin85b,delsarte,Roth91,Roth96,RechtFP10}.
Further, the noiseless case is more natural for our applications to rank-metric codes, and allows
the results to be field independent.

We begin by noting that low-rank recovery (recall Definition~\ref{defn:lrr}, which we consider in
this section only for matrices) generalizes the notion of sparse-recovery, which is the defined
formally as the following.

\begin{definition}
	A set of vectors $\mathcal{V}\subseteq\K^n$ is an \textbf{$s$-sparse-recovery set} if for
	every vector $\vec{x}\in\F^n$ with at most $s$ non-zero entries, $\vec{x}$ is uniquely
	determined by $\vec{y}$, where $\vec{y}\in\K^{\mathcal{V}}$ is defined by $y_{\vec{v}}\eqdef
	\la\vec{x},\vec{v}\ra$, for $\vec{v}\in\mathcal{V}$.

	An algorithm performs \textbf{recovery} from $\mathcal{R}$ if, for each such $\vec{x}$, it
	recovers that $\vec{x}$ given $\vec{y}$.
\end{definition}

That LRR generalizes the sparse-recovery is formalized in the following claim.

\begin{lemma}
	Given an $r$-low-rank recovery set $\mathcal{R}$ for $n\times n$ matrices, there is a set
	$\mathcal{V}\subseteq \F^n$, efficiently constructible from $\mathcal{R}$, with
	$|\mathcal{V}|=|\mathcal{R}|$, such that $\mathcal{V}$ is an $r$-sparse-recovery set.
\end{lemma}
\begin{proof}
	Given an $r$-sparse vector $\vec{x}\in\F^n$ construct the diagonal matrix
	$\Lambda\in\F^{n\times n}$ with $\vec{x}$ on its diagonal.  Thus, $\Lambda$ is rank $\le r$.
	Thus, if we can perform $r$-low-rank-recovery we can also do $r$-sparse recovery.  Each such
	measurement of $\Lambda$ can be seen to also be a linear measurement of $\vec{x}$, so this
	yields $\mathcal{V}$.
\end{proof}

The purpose of this section is to show that the two problems (when concerned with non-adaptive,
exact measurements) are essentially equivalent.  That is, one can (efficiently) perform
low-rank-recovery given \textit{any} construction of a sparse-recovery set.

To motivate the reduction from low-rank-recovery to sparse-recovery, we will show that our above
hitting set results already imply low-rank-recovery results, and that these hitting sets can be seen
as being constructed from a well-known sparse-recovery construction.  We begin by recalling
Lemma~\ref{lem:hittolrr} (standard) fact that \textit{any} hitting set family yields a
low-rank-recovery family, so in particular our results do so. Combining the above with our
constructions of hitting sets, we derive the following corollary.

\begin{corollary}
	The sets $\mathcal{B}_{2r,n,m}$, $\mathcal{D}_{2r,n,m}$, $\mathcal{D}_{2r,n,m}'$, and
	$\mathcal{B}_{2r,n,m}'$ (from Construction~\ref{construction:construction of hitting set for
	matrices}, Construction~\ref{constr:diag} and Construction~\ref{constr:better hitting for
	matrices}) are $r$-low-rank-recovery sets.
	\label{cor:detinfotheoryrecovery}
\end{corollary}

\sloppy However, the above results are non-constructive.  That is, they show that recovery is
information-theoretically possible from this set of matrices, but do not give any insight how to
perform this recovery efficiently.  The purpose of this section is to show that we can strengthen
Corollary~\ref{cor:detinfotheoryrecovery} such that the recovery can be efficiently performed.

To motivate our recovery algorithm, let us first discuss the $r$-low-rank-recovery set
$\mathcal{D}_{2r,n,m}$.  For an $n\times m$ matrix $M$, consider the constraints that the system
$\la M,\mathcal{D}_{2r,n,m}\ra=\vec{0}$ imposes on $M$.  By construction of $\mathcal{D}_{2r,n,m}$,
we see that each $k$-diagonal of $M$ has $2r$ constraints imposed on it.  If we write the
$k$-diagonal of $M$ as $\vec{x}$, we can express the constraints on $\vec{x}$ as $A\vec{x}=0$, where
$A$ is of size $2r\times|\vec{x}|$, where $|\vec{x}|$ denotes the size of the $k$-diagonal.
Further, $A$ has the format (when $2r\le k+1\le n$)
\begin{equation}\label{eqn:vmatrix}
\begin{pmatrix}
	1	&	1		&	1		&	\cdots	&	1\\
	1	&	g		&	g^2		&	\cdots	&	g^{|x|-1}\\
	1	&	g^2		&	g^4		&	\cdots	&	g^{2(|x|-1)}\\
	\vdots	&	\vdots		&	\vdots		&	\ddots	&	\vdots\\
	1	&	g^{2r-1}	&	g^{2(r-1)}	&	\cdots	&	g^{(2r-1)(|x|-1)}
\end{pmatrix}
\end{equation}
which is important because of the following claim.
\begin{lemma}\label{lem:vandermonde sparse recovery}
	Let $\vec{x}$ be an $r$-sparse $\F$-vector. Let $g$ be of order $\ge |\vec{x}|$ in some
	extension $\K$ of $\F$, and let $A$ be an $2r\times |\vec{x}|$ sized matrix of the form in
	Equation~\eqref{eqn:vmatrix}.  Then $\vec{x}$ is determined by $A\vec{x}$.
\end{lemma}
\begin{proof}
	Suppose $\vec{x}$ and $\vec{y}$ are two $r$-sparse vectors such that $A\vec{x}=A\vec{y}$.
	By linearity we then have that $A(\vec{x}-\vec{y})=0$, so that $A$ has a linear dependence
	on $\le 2r$ of the columns.

	However, as the order of $g$ is $\ge |\vec{x}|$, each $2r\times 2r$ minor of $A$ is a
	Vandermonde matrix on distinct entries, and so is full-rank.  In particular, any linear
	dependence on $\le 2r$ of the rows must be zero.  So $\vec{x}-\vec{y}=0$, so
	$\vec{x}=\vec{y}$.  Thus, $\vec{x}$ is determined by $A\vec{x}$.
\end{proof}

Note that the row-space of the above matrix is a Reed-Solomon code, and so the above lemma shows the
standard fact that the dual Reed-Solomon code has good distance.  In particular, we can do error
correction for up to $r$ errors.  This is exactly the question of $r$-sparse recovery (when we are
correcting errors from the $\vec{0}$ codeword).

This lemma shows that at each $k$-diagonal, $\mathcal{D}_{2r,n,m}$ embeds an $r$-sparse-recovery
set.  Thus, it seems plausible that a low-rank-recovery algorithm for $\mathcal{D}_{2r,n,m}$ might
only use this fact in its construction, and thus show low-rank-recovery can be done whenever each of
the $k$-diagonals are measured according to an $r$-sparse-recovery set.  Indeed, this is what is
shown by Theorem~\ref{thm:lrr to sparse}.

The reduction from low-rank-recovery to sparse-recovery is detailed in the following two
subsections.  The first subsection details a slightly stronger notion of sparse-recovery, which we
call \textit{advice-sparse-recovery}.  This notion requires sparse-recovery when supplied with some
advice on the support of the unknown vector.  This is the correct notion of sparse-recovery when
attempting to do low-rank-recovery, but the standard notion is sufficient with some loss in
parameters.  We describe a well-known algorithm, known as Prony's method, for efficiently performing
the recovery illustrated in Lemma~\ref{lem:vandermonde sparse recovery}, and show that this method
can be modified to also achieve advice-sparse-recovery.

The second subsection gives the reduction from low-rank-recovery to sparse-recovery.  Combining this
with our modifications to Prony's method, we conclude that the low-rank-recovery shown in
Corollary~\ref{cor:detinfotheoryrecovery} can also be performed efficiently.

\subsection{Prony's Method and Syndrome Decoding of Dual Reed-Solomon Codes}\label{sec:prony}

In this section we detail an algorithm for efficiently performing the sparse-recovery demonstrated
in Corollary~\ref{lem:vandermonde sparse recovery}.  While our discovery of the algorithm was
independent of prior work, it was original detailed by Prony~\cite{prony1795} in 1795 and is
well-known in the signal-processing community (see \cite{novelprony} and references there-in).  It
can also be seen as syndrome decoding of the dual to the Reed-Solomon code.  What we detail here is
not exactly the original method, as we seek an \textit{advice-sparse-recovery set}, which is a
slightly stronger condition which will be useful in our low-rank-recovery algorithm.  In coding
theory terminology, we are seeking to syndrome decode the dual Reed-Solomon code in the presence of
erasures.  We now define this stronger notion.

\begin{definition}
	A set of vectors $\mathcal{V}\subseteq\F^n$ is an \textbf{$s$-advice-sparse-recovery set} if
	for every $S\in\binom{\llb n\rrb}{\le 2s}$, and vector $\vec{x}\in\F^n$ with $\le s-|S|/2$
	non-zero entries outside of $S$, $\vec{x}$ is uniquely determined by $S$ and $\vec{y}$,
	where $\vec{y}\in\F^{\mathcal{V}}$ is defined by $y_{\vec{v}}\eqdef \la\vec{x},\vec{v}\ra$,
	for $\vec{v}\in\mathcal{V}$.

	An algorithm performs \textbf{recovery} from $\mathcal{V}$ if, for each such $\vec{x}$, it
	recovers that $\vec{x}$ given $S$ and $\vec{y}$.
\end{definition}

Note that the vector $\vec{y}$ can also be defined as $\vec{y}=V\vec{x}$, where $V\in\F^{\mathcal{V}\times
n}$ is the matrix whose rows are those vectors in $\mathcal{V}$.

The motivation for this new definition is to capture situations where $\vec{x}$ is known to have
sparse support overall, and further some of its support is already known and given by the set $S$.
The results below show that exploiting this knowledge allows $|\mathcal{V}|$ to be smaller.  To see
why this might be intuitively plausible, one can count degrees of freedom.  In an $s$-sparse vector
$\vec{x}$, there are intuitively $2s$ degrees of freedom: it takes $s$ degrees to determine
$\supp(\vec{x})$, and it takes $s$ degrees to determine $(x_i)_{i\in\supp(\vec{x})}$.

In the above definition of a $s$-advice-sparse-recovery set, the unknown vector $\vec{x}$ can have a
support of size $2s$ (when $|S|=2s$).  If one ignores the set $S$, there would be $4s$ degrees of
freedom, by the above argument, leading one to expect a lower bound of ``$|\mathcal{V}|\ge 4s$''.
However, if one exploits this knowledge, then there are only $s-|S|/2$ degrees of freedom to
determine $\supp(\vec{x})$, and $|S|+(s-|S|/2)$ degrees of freedom to determine
$(x_i)_{i\in\supp(\vec{x})}$, which gives a total of $2s$ degrees of freedom.

Thus we see that using the information given in $S$ can reduce the degrees of freedom in $\vec{x}$,
and below we match this intuition by recovering $\vec{x}$ from $2s$ measurements. This intuition is
the same intuition in coding theory that an erasure is a ``half error'', but specialized to syndrome
decoding.

In the next subsection, we will see that $r$-low-rank-recovery reduces to the problem of
$r$-advice-sparse-recovery.  When $S=\emptyset$ then $r$-advice-sparse-recovery is exactly the
notion of an $r$-sparse-recovery.  However, we will need $S$ to have size up to $2r$. Note that
regardless of the size of $S$, $\vec{x}$ will be $2r$-sparse.  Thus the following lemma is
immediate.

\begin{lemma}
	Let $\mathcal{V}$ be a $2s$-sparse-recovery set.  Then $\mathcal{V}$ is also a
	$s$-advice-sparse-recovery set.
\end{lemma}

To our knowledge, the existing work on Prony's method gives an algorithm for perform
sparse-recovery.  However, in our reduction advice-sparse-recovery is more natural.  The above lemma shows
that these notions are equivalent, up to a loss in parameters.  However, to get better constructions
we detail how to modify Prony's method to achieve advice-sparse-recovery without a loss in
parameters.

\begin{algorithm}\caption{Prony's method with an advice set}
\label{alg:prony}
\begin{algorithmic}[1]
	\Procedure{PronysMethod}{$n$,$s$,$S$,$y$,$\{g_0,\ldots,g_{n-1}\}$}
		\If{$|S|$ odd}
			\State Enlarge $S$ by 1 position
		\EndIf
		\State $t\eqdef |S|/2$
		\State Construct $A\in\F^{(s+t)\times(s+t+1)}$,
			$A_{i,j}\eqdef
				\begin{cases}
					g_{k_j}^{i}	&	\text{if } i<|S|\\
					y_{i+j-|S|}	&	\text{else}
				\end{cases}
			$
			\Comment{for $S=\{k_0,\ldots,k_{|S|-1}\}$}
		\State Convert $A$ to row-reduced echelon form
		\State Let $r\in \llb s+t\rrb$ be the largest number so the $r\times r$ leading
		principal minor of $A$ is full rank.\label{alg:prony:r}
		\State Let $\vec{c}\in \F^{r+1}$ be a non-zero vector in the nullspace of leading
		$r\times(r+1)$ minor of $A$.
		\State Define $p(x)\eqdef \sum_{i=0}^{r}c_{i}x^i$
		\State $T\eqdef \{k|p(g_k)=0\}$ \Comment{$T$ will be $\supp(\vec{x})$}
		\State $D\in\F^{2s\times T}$, $D_{i,k}\eqdef g_k^{i}$, for $k\in T$
		\State Solve $D\vec{z}=\vec{y}$ for $\vec{z}$ (using Gaussian Elimination)
		\State Define $\vec{x}\in\F^n$, as
		$x_k=
			\begin{cases}
				z_{k}	&	\text{if } k\in T\\
				0	&	\text{else}
			\end{cases}$
		\State \textbf{return} $\vec{x}$
	\EndProcedure
\end{algorithmic}
\end{algorithm}

\begin{theorem}\label{thm:prony}
	Let $\F$ be a field, and let $g_0,\ldots,g_{n-1}\in\F$ be distinct.  Let $\vec{v}_i\in \F^n$
	be the vector with entries $(\vec{v}_i)_j\eqdef g_j^{i}$. Then the set
	$\mathcal{V}=\{\vec{v}_i\}_{i=0}^{2s-1}$ is an $s$-advice-sparse-recovery set. Further,
	\textsc{PronysMethod}$(n,s,S,V\vec{x},\{g_0,\ldots,g_{n-1}\})$ (Algorithm~\ref{alg:prony})
	recovers $\vec{x}$ in $\O(s^3+sn)$ operations (where operations over $\F$ are counted at
	unit cost), where $V\in\F^{2s\times n}$ is the matrix with the vectors in $\mathcal{V}$ as
	its rows.

	In particular, if $g\in\F$ has order at least $n$, we can take $g_j=g^{j}$.
\end{theorem}
\begin{proof}
	As above, define $V\in\F^{2s\times n}$ to be the matrix whose rows are those vectors
	$\vec{v}_i$.  That is, $V_{i,j}=g_j^{i}$.  As the $g_j$ are distinct, it follows that every
	$2s\times 2s$ minor of $V$ is an invertible Vandermonde matrix.  It follows that each subset
	of $\le 2s$ columns of $V$ are linearly independent.

	Define $\vec{g}_j\in\F^{2s}$ by $(\vec{g}_j)_i\eqdef g_j^{i}$.  It follows that the
	$\vec{g}_j$ are the columns of $V$. For a vector $\vec{a}\in\F^m$, define
	$\vec{a}^{[\ell,k]}\in\F^{k-\ell+1}$ to be the vector with entries $a_\ell,\ldots,a_k$.

	\underline{$\mathcal{V}$ is a $s$-advice-sparse-recovery set:} Consider a set
	$S\in\binom{\llb n\rrb}{\le 2s}$ and vectors $\vec{x},\vec{w}\in\F^n$ where each have at
	most $s-|S|/2$ non-zero entries outside of $S$.  Suppose that $V\vec{x}=V\vec{y}$. By
	linearity, this yields the vector $\vec{x}-\vec{w}$ such that $V(\vec{x}-\vec{w})=0$ and
	$\vec{x}-\vec{w}$ has at most $2(s-|S|/2)$ non-zero entries outside of $S$.  In total,
	$\vec{x}-\vec{w}$ has at most $|S|+2(s-|S|/2)=2s$ non-zero entries. However, as mentioned
	above, each subset of $\le 2s$ columns of $V$ are linearly independent.  As
	$\vec{0}=V(\vec{x}-\vec{w})$ is a linear combination of $\le 2s$ columns of $V$, it follows
	that $\vec{x}-\vec{w}=0$. Thus, any such $\vec{x}$ is uniquely determined by $S$ and
	$V\vec{x}$.

	\underline{Algorithm~\ref{alg:prony} performs recovery:} Consider a set $S\in\binom{\llb
	n\rrb}{\le 2s}$, with $S=\{k_0,\ldots,k_{|S|-1}\}$.  For any vector $\vec{x}$ the condition
	that $|\supp(\vec{x})\setminus S|\le s-|S|/2$ implies that $|\supp(\vec{x})\setminus S|\le
	s-\lceil |S|/2\rceil$ by integrality.  It follows that we may assume the set $S$ has even
	size, as we can always enlarge it by one position without changing the above constraints on
	the support of $\vec{x}$. (If $S=\llb n\rrb$ prior to this enlargement, we simulate $n+1$
	long vectors).  Now define $t$ so $|S|=2t$.

	Consider vector $\vec{x}\in \F^n$ with at most $\nu\le s-|S|/2=s-t$ non-zero entries outside
	of $S$. By construction of $\vec{y}$ (recall $\vec{y}=V\vec{x}$),
	\begin{equation}
		\vec{y}=\sum_{k\in S} x_k\vec{g}_k+\sum_{k\in\supp(\vec{x})\setminus S}x_k\vec{g}_k\label{eqn:y}
	\end{equation}
	The aim of this analysis will be to show that we can determine $\supp(\vec{x})$ and then
	leverage this to solve the above equation for $\vec{x}$.

	We now establish some theory to analyze the algorithm.  The above equation can be refined to
	see that
	\begin{equation}
		\vec{y}^{[\ell,\ell']}
			=\sum_{k\in S} x_k\vec{g}_k^{[\ell,\ell']}+\sum_{k\in\supp(\vec{x})\setminus S}x_k\vec{g}_k^{[\ell,\ell']}
			=\sum_{k\in S} x_kg_k^{\ell}\vec{g}_k^{[0,\ell'-\ell]}+\sum_{k\in\supp(\vec{x})\setminus S}x_kg_k^{\ell}\vec{g}_k^{[0,\ell'-\ell]}
		\label{eqn:ys}
	\end{equation}
	We note here that the rows of $A$ involving $\vec{y}$ can be written as
	$\vec{y}^{[0,s+t]},\ldots,\vec{y}^{[s-t-1,2s-1]}$.  As $\vec{y}$ has $2s$ entries, each of
	these vectors is well-defined, and each entry in $\vec{y}$ is used in $A$.

	We now establish some claims about $A$ using that $\nu=|\supp(\vec{x})\setminus S|$.

	\begin{claim}
		The $(|S|+\nu+1)\times (|S|+\nu+1)$ leading principal minor of $A$ is singular.
	\end{claim}
	\begin{proof}
		Denote this leading minor by $M$.  The rows of $M$ are of the form
		$\vec{g}_{k_j}^{[0,|S|+\nu]}$ for $j<|S|$, and $\vec{y}^{[\ell,|S|+\nu+\ell]}$ for
		$0\le \ell< \nu$. Trivially, for each $j<|S|$, $\vec{g}_{k_j}^{[0,|S|+\nu]}\in
		\sspan\{\vec{g}_k^{[0,|S|+\nu]}\}_{k\in\supp(\vec{x})\cup S}$.  Further,
		Equation~\ref{eqn:ys} shows that
		$\vec{y}^{[\ell,|S|+\nu+\ell]}\in\sspan\{\vec{g}_k^{[0,|S|+\nu]}\}_{k\in\supp(\vec{x})\cup
		S}$.  Thus, the $|S|+\nu+1$ rows of $M$ each lie in a $\le(|S|+\nu)$-dimensional
		subspace, implying that $M$ is singular.
	\end{proof}

	\begin{claim}
		The $(|S|+\nu)\times(|S|+\nu)$ leading principal minor of $A$ is invertible.
	\end{claim}
	\begin{proof}
		Denote this leading minor by $M$.  We will show that $M=BC$, for
		$B,C\in\F^{(|S|+\nu)\times(|S|+\nu)}$ both invertible, which implies the claim.

		Let the rows of $C$ be the vectors $\vec{g}_k^{[0,|S|+\nu-1]}$, for each
		$k\in\supp(\vec{x})\cup S$.  We will index the rows by the $g_k$, and assume that
		the first $|S|$ such $g_k$ are those with $k\in S$. This is a Vandermonde matrix,
		and as such is invertible.

		Let $B$ be defined by \[B_{i,g_k}=\begin{cases}1&\text{if } i=k<|S|\\0&\text{if }
		i\ne k, i<|S|\\x_kg_k^{i-|S|}&\text{else}\end{cases}\] It follows from
		Equation~\ref{eqn:ys} that $M=BC$.  Note that $B$ has the form
		\[\begin{bmatrix}I_{|S|}&0\\EX_1&FX_2\end{bmatrix}\] where $X_1\in\F^{|S|\times|S|}$
		is the diagonal matrix with diagonal entries $x_k$, for $k\in S$ (ordered to match
		$C$), $X_2\in\F^{\nu\times\nu}$ is the diagonal matrix with diagonal entries $x_k$,
		for $k\in\supp(\vec{x})\setminus S$ (ordered to match $C$), $E\in\F^{\nu\times |S|}$
		is the Vandermonde matrix with entries $E_{i,g_k}\eqdef g_k^{i}$, for $k\in S$, and
		$F\in\F^{\nu\times \nu}$ is the (invertible) Vandermonde matrix with entries
		$F_{i,g_k}\eqdef g_k^{i}$ for $k\in\supp(\vec{x})\setminus S$.

		Note that $X_1$ might entirely be zero, but $X_2$ must be invertible by assumption
		that $\vec{x}$ has exactly $\nu$ non-zero entries outside of $S$.  As $F$ is
		invertible, it follows that $FX_2$ is invertible, and thus $B$ is also invertible.

		Thus, $M=BC$ with $B$ and $C$ both invertible matrices.  The claim follows.
	\end{proof}

	As the first $|S|$ rows of $A$ are rows of a Vandermonde matrix, it follows that the first
	$|S|$ leading principal minors are all invertible.  This, along with the above two claims,
	thus show that $|S|+\nu$ is the minimum $r$ such the $(r+1)\times (r+1)$ leading principal
	minor of $A$ is singular.  It follows that in Algorithm~\ref{alg:prony} the $r$ value chosen
	in Step~\ref{alg:prony:r} is in fact $|S|+\nu$.

	We now show that the $\vec{c}$ chosen by the algorithm also has significance.

	\begin{claim}
		Let $p(x)\eqdef\prod_{k\in\supp(\vec{x})\cup
		S}(x-g_k)=\sum_{i=0}^{|S|+\nu}c_{i}x^i$. Then the vector $\vec{c}\in\F^{|S|+\nu+1}$
		defined by those coefficients $c_i$ is in the nullspace of the
		$(|S|+\nu)\times(|S|+\nu+1)$ leading minor of $A$.
	\end{claim}
	\begin{proof}
		Denote this leading minor by $M$.

		Note that for any $g_k$ with $k\in \supp(\vec{x})\cup S$ has that $\la
		\vec{g}_k^{[0,|S|+\nu]},\vec{c}\ra=0$, as this simply says that $p(g_k)=0$.  Thus,
		we see that $\vec{c}$ is orthogonal to the first $|S|$ rows of $M$.

		Now observe that Equation~\ref{eqn:ys} shows that the last $\nu$ rows of $M$ are all
		in the span of the vectors $\vec{g}_k^{[0,|S|+\nu]}$ for $k\in \supp(\vec{x})\cup
		S$.  As $\vec{c}$ is orthogonal to each of these vectors by construction, we see
		that it must also be orthogonal to the last $\nu$ rows of $M$.

		Thus, $\vec{c}$ is orthogonal to each row of $M$, and thus is in its nullspace.
	\end{proof}

	The algorithm chooses \textit{some} $\vec{c}$ that is in the nullspace of the
	$(|S|+\nu)\times(|S|+\nu+1)$ leading minor of $A$.  However, as the
	$(|S|+\nu)\times(|S|+\nu)$ leading principal minor of $A$ is invertible, it follows that the
	$(|S|+\nu)\times(|S|+\nu+1)$ leading minor of $A$ has a nullspace of dimension 1.  Thus, the
	$\vec{c}$ chosen by the algorithm must be a (non-zero) multiple of the coefficient vector of
	$\prod_{k\in\supp(\vec{x})\cup S}(x-g_k)$.  It follows that the set $T$ is equal to
	$\supp(x)\cup S$.

	Thus, Equation~\ref{eqn:y} gives a linear system for $\vec{y}$ with $\le 2s$ variables, and
	$2s$ equations, where $\vec{x}$ (restricted to $\supp(x)\cup S$) is a solution.  The system
	is full-rank, so $\vec{x}$ is the only solution.  Further, $\vec{x}$ can be recovered via
	Gaussian Elimination, and this is exactly what Algorithm~\ref{alg:prony} does.  Thus,
	correctness is also established in this case.

	\underline{Algorithm~\ref{alg:prony} runs in $\O(s^3+sn)$ operations:} Constructing the
	matrix $A$ takes $\O(s^2)$ operations, as that is the size of the matrix and each entry can
	be computed in $\O(1)$ operations (the $g_k^{i}$ are computed with $i$ increasing).
	Converting $A$ to reduced-row echelon form takes $\O(s^3)$ operations.  Determining the
	number $r$ in Step~\ref{alg:prony:r} also takes $\O(s)$ operations, as $r=\max\{i|A_{i,i}\ne
	0\}$.  Determining the vector $\vec{c}$ takes $\O(s)$ because the $r\times (r+1)$ minor is
	row-reduced echelon form.  That is, for $1\le i\le r$, $c_i=-A_{i,r+1}$ and $c_{r+1}=1$.
	Constructing $p$ and $T$ takes $\O(sn)$ time, as we just test if $p(g_k)=0$ for each $k$,
	and $p$ is of degree $\O(s)$. $D$ is a Vandermonde matrix with at most $\O(s^2)$ entries,
	and so constructing $D$ takes $\O(s^2)$ steps.  Solving for $\vec{z}$ takes $\O(s^3)$ steps,
	and determining the final $\vec{x}$ takes $\O(n)$ steps.
\end{proof}

This theorem provides us with an $s$-advice-sparse-recovery set, using $2s$ measurements.  We will
now leverage this in the next subsection to get a full algorithm for low-rank-recovery.

\subsection{Low Rank Recovery}\label{sec:lrrtosparse}

In this subsection we describe how the problem of (exact, non-adaptive) $r$-low-rank-recovery
deterministically reduces to the problem of (exact, non-adaptive) $r$-advice-sparse-recovery.
We will first define a normal form for a matrix which we call \textit{$(<k)$-upper-echelon form},
which (recalling the notation of Section~\ref{sec:notation}) is roughly defined as saying that a
matrix $M$ has $M^{(< k)}$ in reduced row-echelon form.  We then show that for any matrix $M$ in
this form, the diagonal $M^{(k)}$ is sparse.  Thus, using sparse-recovery we can then recover this
diagonal. This process is then continued by using row-reduction to put $M$ in $(\le
k)$-upper-echelon form, and then recovering $M^{(k+1)}$ and so on.

The above process uses the sparse-recovery oracle in an adaptive way.  The algorithm we detail below
will actually use the sparse-recovery oracle non-adaptively.  The measurements made to the matrix
$M$ will be the sparse-recovery oracle applied to each $k$-diagonal.  While these diagonals are not
themselves sparse, we show that the row-reduction of $M$ (that makes $M$ into upper-echelon form)
acts such that we can simulate the adaptive measurements from the non-adaptive measurements by
computing the suitable corrections.

We now begin by describing some structural properties of matrices, which we will apply to understand
upper-echelon form.

\begin{definition}
	Let $M$ be an $n\times m$ matrix.  The entry $(i,j)$ is a \textbf{leading non-zero entry},
	if $M_{i,j}\ne0$ and $M_{i,j'}=0$ for $j'<j$.

	Denote $\lne(M)$ to be the set of all such leading non-zero entries.  If $S$ is a subset of
	entries in $M$, denote $\lne(S)\eqdef \lne(M)\cap S$.

	Denote $\lne_R(S)$ to be set containing the rows of the coordinates in $\lne(S)$, and denote
	$\lne_C(S)$ to be the multi-set containing the columns of the coordinates in $\lne(S)$.
\end{definition}

It is clear that each row can have at most one leading non-zero entry, and possibly none.  A column
could be associated with several leading non-zero entries.

\begin{definition}
	An $n\times m$ matrix $M$ is in \textbf{$(<k)$-upper-echelon form} if, for each
	$(i,j)\in\lne(M^{(<k)})$, $M_{i',j}=0$ for all $i<i'< k-j$.
\end{definition}

Note that a matrix is $(<k)$-upper-echelon if it is $(<k')$-upper-echelon and $k'\ge k$, and that every
matrix is vacuously in $(\le 0)$-upper-echelon form.

We now recall the following standard linear-algebraic fact about triangular systems, phrased in the
language of leading non-zero entries.

\begin{lemma}
	Let $M$ be an $n\times m$ matrix with all non-zero rows, such that $\lne_C(M)$ has no
	repetitions.  Then the rows of $M$ are linearly independent.
	\label{lem:tsi}
\end{lemma}
\begin{proof}
	Denote the column of the leading non-zero entry of row $i$ by $j_i$.  Each row must have
	such a value as each row is non-zero.  As linear independence is invariant under
	permutation, we assume without loss of generality that the rows are ordered such that the
	$j_i$ are strictly increasing with $i$.  This is possible as the $j_i$ are assumed to be
	distinct. Write these rows as vectors $\vec{v}^{(i)}$. Now consider any non-trivial linear
	combination $\sum_i c_i\vec{v}^{(i)}$.  Pick ${i_0}$ to be the least number such that
	$c_{i_0}\ne 0$.  As the $j_i$ are strictly increasing, it follows that the $j_{i_0}$-th
	entry of $\vec{v}^{(i)}$ is zero for $i>i_0$.  Thus, we now expand out the $i_0$-th index of
	the above summation \[(\sum_i c_i\vec{v}^{(i)})_{j_{i_0}}=\sum_{i<i_0} c_i\cdot
	\vec{v}^{(i)}_{j_{i_0}}+c_{i_0}\vec{v}^{(i_0)}_{j_{i_0}}+\sum_{i_0<i}c_i\cdot\vec{v}^{(i)}_{j_{i_0}}=\sum_{i<i_0}
	0\cdot
	\vec{v}^{(i)}_{j_{i_0}}+c_{i_0}\vec{v}^{(i_0)}_{j_{i_0}}+\sum_{i_0<i}c_i\cdot0=c_{i_0}\vec{v}^{(i_0)}_{j_{i_0}}\ne
	0\] Thus we see that this linear combination is non-zero, and as this was any non-trivial
	linear combination it follows these rows are linearly independent.
\end{proof}

We now show that matrices in upper-echelon form cannot have many leading non-zero entries.

\begin{lemma}
	Let $M$ be an $n\times m$ matrix of rank $\le r$.  If $M$ is $(<k)$-upper-echelon, then
	$|\lne(M^{(< k)})|\le r$.  Further, $\lne_C(M^{(<k)})$ has no repetitions.
	\label{lem:uer}
\end{lemma}
\begin{proof}
	Given $(i,j)\in\lne(M^{(<k)})$, $(<k)$-upper-echelon form implies that $M_{i',j}=0$ for any
	$i'$ with $i<i'<k-j$.  It follows that given two distinct entries $(i,j),(i',j)\in M^{(<
	k)}$ at most one can be a leading non-zero entry.  Thus we see that $\lne_C(M^{(< k)})$ has
	no repetitions.

	Lemma~\ref{lem:tsi} then implies that the rows in $\lne_R(M^{(< k)})$ are linearly
	independent.  Thus, $|\lne(M^{(< k)})|\le \rank(M)\le r$.
\end{proof}

The next lemma is the key insight of the algorithm.  It shows that, for any matrix in
$(<k)$-upper-echelon form, the $k$-diagonal must be sparse.  Further, the sparseness is bounded by
twice the rank of the matrix (the lemma presents a more refined statement).

\begin{lemma}
	Let $M$ be an $n\times m$ matrix with rank $\le r$, such that $M$ is in $(<k)$-upper-echelon
	form with $0\le k\le n+m-2$. Let $s\eqdef |\lne(M^{(<k)})|$, $I\eqdef \lne_R(M^{(<k)})$,
	$J\eqdef \lne_C(M^{(<k)})$.

	Then $M^{(k)}$ has $\le r-s$ non-zero entries with columns outside $S\eqdef (k-I)\cup J$,
	and thus $M^{(k)}$ is $(r+s)$-sparse.
	\label{lem:uesparse}
\end{lemma}
\begin{proof}
	Note that by Lemma~\ref{lem:uer} we have that $s\le r$, so that $r-s\ge 0$ and $r+s\le 2r$.

	Let $I'$ be the rows that contain non-zero entries in $M^{(k)}$, whose columns lie outside
	$S$.  We will show that the rows in $I\cup I'$ are linearly independent.  This will complete
	the claim as $|I'|\le \rank(M)-|I|\le r-s$, and observing that $|S|\le 2s$.

	Now consider the columns of the leading non-zero entries of the rows in $I'$.  Any row $i\in
	I$ intersects $M^{(k)}$ at column $k-i\in S$.  This means that row $i$ cannot contain a
	non-zero entry in $M^{(k)}$ with column outside of $S$, so $I$ and $I'$ are disjoint.

	Any row $i$ with a non-zero entry in $M^{(<k)}$ must have a leading non-zero entry in
	$M^{(<k)}$, and thus any such $i$ is contained in $I$.  Thus, as $I$ and $I'$ are disjoint,
	it follows that any row $i'\in I'$ only has zero entries within $M^{(<k)}$.  As such a row
	$i'$ has a non-zero entry on $M^{(k)}$, it follows that the leading non-zero entry of a row
	$i'\in I'$ is $(i',k-i')$.  This implies that the columns of the leading non-zero entries of
	the rows in $I'$ are distinct (and outside of $S$ by construction).

	The rows in $I$ have leading non-zero entries in $J\subseteq S$ and by Lemma~\ref{lem:uer},
	$J$ has no repetitions.  Thus, it follows that the rows $I\cup I'$ all have distinct columns
	for their leading non-zero entries, which, by Lemma~\ref{lem:tsi}, implies that these rows
	are linearly independent.  Invoking the rank bound, as mentioned above, completes the proof.
\end{proof}

This lemma motivates the following idea for low-rank reconstruction.  Iteratively, convert (using
row-reduction) the matrix into $(<k)$-upper-echelon form and then reconstruct, using any
sparse-recovery method, the $k$-th diagonal.  This is exactly the algorithm we will present.
However, to establish correctness, we need to first understand how to convert a matrix into
$(<k)$-upper-echelon form, even in situations when $M^{(\ge k)}$ is unknown.

To do this, we will use row-reduction, as implemented by left-multiplication by lower-triangular
matrices.  The following lemma shows that such multiplication can be computed on the partial
matrices $M^{(< k)}$.

\begin{lemma}
	Let $M$ be an $n\times m$ matrix, and $L$ be an $n\times n$ lower-triangular matrix.  Then
	$(LM)^{(< k)}$ is computable in $\O(\min(n,k)\min(m,k)k)$ arithmetic operations from $L$ and
	$M^{(< k)}$.
	\label{lem:kdiagcompute}
\end{lemma}
\begin{proof}
	An entry $(LM)_{i,j}$, for $i+j< k$, is equal to $\sum_{l=0}^n L_{i,l}M_{l,j}$, which equals
	$\sum_{l=1}^i L_{i,l}M_{l,j}$ as $L$ is lower-triangular.  Further, this sum is computable
	from $L$ and the $(< k)$-diagonals of $M$ as $l+j\le i+j< k$.  The time bound is the obvious
	bound on computing each of $\O(\min(n,k)\min(m,k))$ sums of $\le k$ terms.
\end{proof}

We now establish a useful property on composing left-multiplication of special types of
lower-triangular matrices.

\begin{lemma}
	Let $L,L'$ be $n\times n$ invertible, lower-triangular matrices, with all $1$'s along the
	main diagonal.  Then $LL'$ is an invertible, lower-triangular matrix, with all $1$'s along
	the main diagonal,

	Further, if both $L-I_n$ and $L'-I_n$ only have non-zero entries in a subset $J$ of the
	columns, then $LL'-I_n$ also has this property.
	\label{lem:lt}
\end{lemma}
\begin{proof}
	That facts that $LL'$ is an invertible, lower-triangular matrix and has all $1$'s along the
	main diagonal, are each straightforward.

	We now prove the desired property of $LL'-I_n$.  Consider some entry $(i,j)$ in $LL'$, with
	$j\notin J$ and $i>j$.  It is then that
	\begin{align*}
		(LL')_{i,j}
			&=\sum_{k\in\llb n\rrb} L_{i,k}L'_{k,j}
			=\sum_{i\ge k\ge j} L_{i,k}L'_{k,j}
			=L_{i,i}L'_{i,j}+\sum_{i> k>j} L_{i,k}L'_{k,j}+L_{i,j}L'_{j,j}\\
			&=1\cdot L'_{i,j}+\sum_{i> k>j} L_{i,k}L'_{k,j}+L_{i,j}\cdot 1
	\end{align*}
	Observe that as $i>j$ and $j\notin J$, $L'_{i,j}=L'_{k,j}=L_{i,j}=0$ (for any $k>j$).  Thus,
	the above sum is zero.  Hence, the desired entries $(i,j)$ with $i>j$ and $j\notin J$ are
	zero, proving the claim.
\end{proof}

We now use these lemmas to analyze Algorithm~\ref{alg:makeue}, which gives a way to transform a
matrix in $(<k)$-upper-echelon into one which is $(\le k)$-upper-echelon, and does so efficiently.

\begin{algorithm}\caption{Transform a $(<k)$-upper-echelon matrix into $(\le k)$-upper-echelon form}
\label{alg:makeue}
\begin{algorithmic}[1]
	\Procedure{MakeUpperEchelon}{$M$,$n$,$m$,$k$}
		\State $L\leftarrow I_n$
		\ForAll{$(i,j)\in\lne(M^{(<k)})$}
			\State $L\leftarrow (I_n-\frac{M_{k-j,j}}{M_{i,j}}E_{k-j,i})\cdot L$
			\Comment{$M_{i,j}\ne 0$ as $(i,j)$ is leading non-zero entry in row $i$}
			\label{line:updateL}
		\EndFor
		\State \textbf{return} $L$
	\EndProcedure
\end{algorithmic}
\end{algorithm}

\begin{claim}
	Let $M$ be an $n\times m$ matrix of rank $\le r$, such that $M$ is in $(<k)$-upper-echelon
	form, for $0\le k\le n+m-2$.  Then the procedure \textsc{MakeUpperEchelon}$(M,n,m,k)$
	(Algorithm~\ref{alg:makeue}) runs in $\O(rn)$ time and returns an invertible $n\times n$
	lower-triangular matrix $L$ computed only from $M^{(\le k)}$, such that $LM$ is $(\le
	k)$-upper-echelon and $(LM)^{(<k)}=M^{(<k)}$.

	Also, $L$ is the product of $\le r$ elementary matrices and each main diagonal entry is
	equal to 1.

	Further, $L-I_n$ only has non-zero entries with columns in $\lne_R(M^{(<k)})$.
	\label{clm:makeue}
\end{claim}
\begin{proof}
	\underline{$(LM)^{(<k)}=M^{(<k)}$:} We argue that the identity $(LM)^{(<k)}=M^{(<k)}$ is
	invariant.  As $L=I_n$ initially, the identity holds at the beginning of the algorithm.  We
	now proceed by induction.

	In each run of Line~\ref{line:updateL}, we add a multiple of row $i$ to row $k-j$ in $LM$,
	where $(i,j)\in\lne(M^{(<k)})$ and thus $i+j< k$.  Thus, row $i$ in $M$ has the first $j-1$
	entries being zero.  By induction on the identity $(LM)^{(<k)}=M^{(<k)}$, the first $j-1$
	entries in row $i$ of $LM$ are also zero when Line~\ref{line:updateL} is run.  It follows
	that the only action of this update to $(LM)^{(\le k)}$ is to set $(LM)_{k-j,j}=0$.  Thus,
	$(LM)^{(<k)}$ is unchanged, so $(LM)^{(<k)}=M^{(<k)}$ still holds.

	\underline{$LM$ has $(\le k)$-upper-echelon form:} As $(LM)^{(<k)}=M^{(<k)}$ throughout the
	algorithm, and $M$ is in $(<k)$-upper-echelon form, it follows that $LM$ is in
	$(<k)$-upper-echelon form at termination.  To show $LM$ is in $(\le k)$-upper-echelon form
	upon termination,  it suffices to show that $(LM)_{k-j,j}=0$ for all $j\in\lne_C(M^{(<k)})$.
	As running Line~\ref{line:updateL} has exactly this effect (and these updates are disjoint
	and idempotent, thus do not conflict), and this line is run for all
	$(i,j)\in\lne(M^{(<k)})$, it follows that $LM$ is in $(\le k)$-upper-echelon form on
	termination.

	\underline{$L$ computable from the $(\le k)$-diagonals of $M$:} This is straightforward, as
	each query to $M$ is within the $(\le k)$-diagonals.

	\underline{$L$ is the product of $\le r$ elementary matrices:} Each update to $L$ by
	Line~\ref{line:updateL} left-multiplies $L$ by an elementary matrix.  By
	Lemma~\ref{lem:uer}, $|\lne(M^{(<k)})|\le r$, so the loop of the algorithm is run at most
	$r$ times.

	\underline{Structure of $L$:} By construction,  $L$ is the product of matrices of the form
	$I_n+cE_{k-j,i}$, where $i+j<k$ and $(i,j)\in\lne(M^{(<k)})$.  Regardless of the value of
	$c$, such a matrix is invertible, lower-triangular, with main diagonal entries all 1, and
	all non-zero entries of $(I_n+cE_{k-j,i})-I_n$ have columns in $\lne_R(M^{(<k)})$.  By
	Lemma~\ref{lem:lt} it follows that $L$ also has these properties.

	\underline{Complexity:} Left-multiplication by an elementary matrix can be done in $\O(n)$
	steps, and by the above analysis, there are $\le r$ such multiplications.  Further, by
	storing the leading non-zero entries in each row, the pairs $(i,j)$ can be determined in
	$\O(n)$ time.  Thus the time is $\O(rn)$ overall.
\end{proof}

We now present the low-rank recovery algorithm, and its analysis.

\begin{algorithm}\caption{Reconstruct a matrix from inner-products $\{\la M,R\ra\}_{R\in\mathcal{R}_k,0\le k\le n+m-2}$}
\label{alg:lrr}
\begin{algorithmic}[1]
	\Procedure{LowRankRecovery}{$n$,$m$,$\{\la M,R\ra\}_{R\in\mathcal{R}_k,0\le k\le n+m-2}$}
		\State $L\leftarrow I_n$
		\State $N\leftarrow 0^{n\times m}$
		\State $P\leftarrow 0^{n\times m}$
		\For{$0\le k\le n+m-2$}\label{line:for}
			\State $A\leftarrow 0^{n\times m}$\label{line:forbegin}
			\State $A^{(k)}\leftarrow ((L-I_n)N)^{(k)}$\label{line:upa}
			\State $S\leftarrow \left(k-\lne_R((P^{(<k)}))\right)\cup \lne_C(P^{(<k)})$
				\label{line:s}
			\State $P^{(k)}\leftarrow \sr_k(\{\la M,R\ra+\la A,R\ra\}_{R\in\mathcal{R}_k},S)$
				\label{line:srk}
			\State $N^{(k)}\leftarrow P^{(k)}-A^{(k)}$ \label{line:n}
			\State $L_k\leftarrow \Call{MakeUpperEchelon}{P,n,m,k}$ \label{line:mue}
			\State $P^{(k)}\leftarrow (L_kP)^{(k)}$ \label{line:upp}\Comment{Update
			$\lne(P^{(\le k)})$}
			\State $L\leftarrow L_kL$\label{line:upl}
		\EndFor\label{line:endfor}
	\EndProcedure
\end{algorithmic}
\end{algorithm}

\begin{theorem}\label{thm:lrr to sparse}
	Let $m\ge n\ge r\ge 1$. For $0\le k\le n+m-2$, let $\mathcal{R}_k$ be sets of $n\times m$
	matrices such that
	\begin{enumerate}
		\item For $k'\ne k$, $R^{(k')}=0$ for $R\in\mathcal{R}_k$
		\item $\{R^{(k)}\}_{R\in\mathcal{R}_k}$ forms a
		$\min(r,k+1,(n+m)-(k+1))$-advice-sparse-recovery set.
	\end{enumerate}
	Then $\mathcal{R}=\bigcup_k \mathcal{R}_k$ is an $r$-low-rank-recovery set.

	If, for each $k$, the set $\{R^{(k)}\}_{R\in\mathcal{R}_k}$ has an
	$\min(r,k+1,(n+m)-(k+1))$-advice-sparse-recovery algorithm $\sr_k$ running in time $t_k$,
	then Algorithm~\ref{alg:lrr} performs $r$-low-rank-recovery for $\mathcal{R}$ in time
	$\O\left(rnm+\sum_{k=2}^{n+m} (t_k+n|\mathcal{R}_k|)\right)$.
\end{theorem}
\begin{proof}
	We will first show that $\mathcal{R}$ is an $r$-low-rank-recovery set by showing that
	Algorithm~\ref{alg:lrr} performs recovery, assuming oracle access to
	$r$-advice-sparse-recovery oracles $\sr_k$.  We will then analyze the run-time.

	\begin{claim}
		The following invariants hold at Line~\ref{line:endfor}, at the end of the loop.
		\begin{enumerate}
			\item $N^{(\le k)}=M^{(\le k)}$\label{inv:m}
			\item $P^{(\le k)}=(LM)^{(\le k)}$\label{inv:lm}
			\item $P$ is in $(\le k)$-upper-echelon form\label{inv:ue}
			\item $L$ is lower-triangular, invertible , main diagonal is all 1's, and
			$L-I_n$ only has non-zero entries with columns in
			$\lne_R(P^{(<k)})$\label{inv:l}
		\end{enumerate}
	\end{claim}
	\begin{proof}
		The proof will be by induction.

		\underline{$k=0$:} The loop begins with $L=I_n$, $N=0_n$, $P=0_n$.  It follows that
		$A=0_n$ in this run of the loop, and that $S=\emptyset$.  Thus, $P^{(0)}$ is set to
		$\sr_0(\{\la M,R\ra\}_{R\in\mathcal{R}_0},\emptyset)$.  As $r\ge 1$, we get that
		$\mathcal{R}_0^{(0)}$ is a $1$-advice-sparse-recovery set and as $M^{(0)}$ has at
		most $1$ element, it follows that $\sr_0$ recovers it correctly and thus
		$P^{(0)}=M^{(0)}$ after Line~\ref{line:srk}.  As $A=0_n$ it follows that $N^{(\le
		0)}=M^{(\le 0)}$ also, satisfying Invariant~\ref{inv:m}.

		Now observe that the procedure \textsc{MakeUpperEchelon}, when run on $k=0$, will
		always return $I_n$.  Thus, $L_k$, and $L$, are both $I_n$ at the end of the loop,
		satisfying Invariant~\ref{inv:l}.  Invariant~\ref{inv:ue} is vacuously true as any
		matrix is in $1$-upper-echelon form.  Finally, using that $L=L_k=I_n$, we see that
		$P$ is unchanged after Line~\ref{line:srk} and so $P^{(\le 0)}=(LM)^{(\le 0)}$,
		satisfying Invariant~\ref{inv:lm}.

		\underline{$k>0$:} Using that the invariants held at $k-1$, we now establish them at
		$k$.  As $P^{(<k)}=(LM)^{(<k)}$ and $P$ is in $(<k)$-upper-echelon form, it follows
		that $LM$ is in $(<k)$-upper-echelon form.  By Lemma~\ref{lem:uesparse}, it follows
		$(LM)^{(k)}$ has at most $r-s/2$ non-zero entries with columns outside of
		$S=(k-\lne_R((LM)^{(<k)}))\cup \lne_C((LM)^{(<k)})$, where $s=|\lne((LM)^{(<k)})|$
		and $|S|\le 2s$.  However, using again that $P^{(<k)}=(LM)^{(<k)}$ it follows that
		$(LM)^{(k)}$ has at most $r-|S|/2$ non-zero entries with columns outside of $S$,
		where $S$ is as constructed in Line~\ref{line:s}.  As $(LM)^{(k)}$ has
		$\min(k+1,(n+m)-(k+1),n)$ non-zero entries total, and $\mathcal{R}_k$ is an
		$\min(r,k+1,(n+m)-(k+1))$-advice-sparse-recovery set, it follows (as $r\le n$) that
		$\sr_k(\{\la LM,R\ra\}_{R\in\mathcal{R}_k},S)$ successfully recovers $(LM)^{(k)}$.
		That is, if $r\ne \min(r,k+1,(n+m)-(k+1))$ then we have enough measurements to fully
		recover $(LM)^{(k)}$ regardless of its sparsity and the value of $S$(and the oracle
		will perform this recovery), and if $r=\min(r,k+1,(n+m)-(k+1))$ then we use the
		advice-sparse-recovery oracle.

		We now use the following claim to show how the $\{\la LM,R\ra\}$ can be computed.

		\begin{claim}
			At the beginning of the loop in Line~\ref{line:for},
			$(LM)^{(k)}=M^{(k)}+((L-I_n)N)^{(k)}$
		\end{claim}
		\begin{proof}
			As $LM=M+(L-I_n)M$, it is enough to show that
			$((L-I_n)M)^{(k)}=((L-I_n)N)^{(k)}$.

			By induction on the above invariants, $L$ is lower-triangular, with all 1's
			along the main diagonal, and $N^{(<k)}=M^{(<k)}$.  Thus,
			$(L-I_n)_{i,\ell}=0$ for $i\le \ell$, and $M_{l,j}=N_{\ell,j}$ for
			$\ell<k-j$. For any $j\le k$,
			\begin{align*}
				((L-I_n)M)_{k-j,j}
					&=\sum_{\ell\in\llb n\rrb} (L-I_n)_{k-j,\ell}M_{\ell,j}
					=\sum_{\ell<k-j} (L-I_n)_{k-j,\ell}M_{\ell,j}
					=\sum_{\ell<k-j} (L-I_n)_{k-j,\ell}N_{\ell,j}\\
					&=\sum_{\ell\in\llb n\rrb} (L-I_n)_{k-j,\ell}N_{\ell,j}
					=((L-I_n)N)_{k-j,j}
			\end{align*}
			Thus $((L-I_n)M)^{(k)}=((L-I_n)N)^{(k)}$, giving the claim.
		\end{proof}

		The above claim shows that at Line~\ref{line:srk} we have that $\la LM,R\ra=\la
		M,R\ra+\la A,R\ra$, for all $R\in\mathcal{R}_k$, using that $R^{(k')}=0$ for $k'\ne
		k$. This shows that Line~\ref{line:srk} correctly implements advice-sparse-recovery
		of $(LM)^{(k)}$, and thus sets $P^{(k)}$ to this value.  It follows that at the end
		of this line that $P^{(\le k)}=(LM)^{(\le k)}$.

		\underline{Invariant~\ref{inv:m}:} Using the identity proved in the above claim, and
		the just proven fact that $P^{(\le k)}=(LM)^{(\le k)}$ at the end of
		Line~\ref{line:srk}, it follows that at the end of Line~\ref{line:n} that
		$N^{(k)}=M^{(k)}$, and thus $N^{(\le k)}=M^{(\le k)}$.  As $N$ is not changed
		further, this establishes Invariant~\ref{inv:m}.

		\underline{Invariant~\ref{inv:ue}:} We now examine
		Lines~\ref{line:mue}--\ref{line:upl}.  As $P$ has only changed in its $k$-diagonal,
		it is still in $(<k)$-upper-echelon form. Thus, Line~\ref{line:mue} returns $L_k$
		such that $L_kP$ is in $(\le k)$-upper-echelon form, by Claim~\ref{clm:makeue}.
		Further $(L_kP)^{(\le k)}$ only differs from $P^{(\le k)}$ along the $k$-diagonal,
		so it follows that after the update in Line~\ref{line:upp} that $P$ is in $(\le
		k)$-upper-echelon form.  As $P$ is not further modified, this establishes
		Invariant~\ref{inv:ue}.

		\underline{Invariant~\ref{inv:lm}:} Further, as we take $L\leftarrow L_kL$ in
		Line~\ref{line:upl} and previously had that $P^{(\le k)}=(LM)^{(\le k)}$, it follows
		that at the end of Line~\ref{line:upl} we have that $P^{(\le k)}=(LM)^{(\le k)}$
		still, as both $P$ and $LM$ have been multiplied by $L_k$.  This establishes
		Invariant~\ref{inv:lm}.

		\underline{Invariant~\ref{inv:l}:} In Line~\ref{line:mue}, Claim~\ref{clm:makeue}
		shows that $L_k$ is a lower-triangular and invertible matrix, with main diagonal
		entries all 1's, and $L_k-I_n$ only has non-zero entries in columns
		$\lne_R(P^{(<k)})$.  As $P^{(<k)}$ is not modified further, this remains true at the
		end of the loop at Line~\ref{line:endfor}. By induction, at Line~\ref{line:for} we
		have that $L$ is lower-triangular, invertible, with main diagonal entries all 1's,
		and $L-I_n$ only has non-zero entries in columns $\lne_R(P^{(<(k-1))})$.  As
		$P^{(<(k-1))}$ remains unchanged throughout this iteration of the loop, this is also
		true at the beginning of Line~\ref{line:upl}.  By Lemma~\ref{lem:lt}, it follows
		that after Line~\ref{line:upl} $L$ still has the properties of being
		lower-triangular, invertible, main diagonal entries being 1's, and $L-I_n$ only has
		non-zero entries in $\lne_R(P^{(<k)})$. This establishes Invariant~\ref{inv:l}

		Thus, each of the invariants are established for this value of $k$ given that they
		hold for $k-1$, so the invariants hold for all $k$ by induction.
	\end{proof}

	The above claim shows that at the end of the algorithm, $N^{(\le k)}=M^{(\le k)}$ for
	$k=n+m-2$.  But this implies $N=M$, and thus $M$ is reconstructed successfully.

	\underline{Run-time Analysis:} We now bound the run-time of Algorithm~\ref{alg:lrr}. The
	steps outside the for-loop take $\O(nm)$, so it suffices to bound each step of the loop. We
	will show that each step of the loop takes $\O(rn+t_k+n|\mathcal{R}_k|)$ steps.  As there
	are $n+m$ such iterations of the loop, the quoted bound follows.

	We begin by noting that the algorithm will not recompute $\lne(P^{(<k)})$ at each stage.
	Instead, this will be maintained throughout the algorithm.  As each row of $P$ can have at
	most one leading non-zero entry, this is easily stored and indexed.  Further, as
	$P^{(<k)}=(LM)^{(<k)}$ and the rank bound on $M$ shows, via Lemma~\ref{lem:uer}, that
	$|\lne((LM)^{(<k)})|\le r$, it follows that if the set $\lne(P^{(<k)})$ is maintained as a
	linked list, that traversing it entirely takes $\O(r)$ time.

	Note that we do not need to modify $\lne(P^{(<k)})$ when running \textsc{MakeUpperEchelon},
	and can defer modification to after Line~\ref{line:upp}.  At that point $P^{(\le k)}$ has
	been determined, and can be used to compute $\lne(P^{(< (k+1))})=\lne(P^{(\le k)})$ in
	$\O(n)$ time.  Thus, $\lne(P^{(<k)})$ can be maintained within the quoted time bounds, and
	accessed as a $\O(r)$ sized linked list.

	We now analyze the lines of the loop. As written, Line~\ref{line:forbegin} takes
	$\Theta(nm)$ time, which is above the quoted run-time bounds.  However, one can observe that
	$A$ is only ever accessed at the values $A^{(k)}$, when noting that $R\in\mathcal{R}_k$ is
	only non-zero on its $k$-diagonal.  Thus, Line~\ref{line:forbegin} is actually superfluous
	and can be omitted.

	Line~\ref{line:upa} takes $\O(rn)$ steps. For, the above invariants show that $L-I_n$ only
	has non-zero entries in the columns $\lne_R(P^{(<k)})$, and as discussed above this set has
	at most $r$ elements.  Thus, each of the $\le n)$ elements of $A^{(k)}$ is the sum of $\le
	r$ elements of $N$.  Thus $A^{(k)}$ can be computed in $\O(rn)$ steps.

	Line~\ref{line:s} takes $\O(r)$ steps, as $\lne_R(P^{(<k)})$ is pre-computed.

	Line~\ref{line:srk} takes $\O(t_k+n|\mathcal{R}_k|)$ steps.  For, each inner product $\la
	A,R\ra$ takes $\O( n)$ steps (as each matrix is only non-zero on the $k$-diagonal, which has
	at most $n$ entries), and there are $|\mathcal{R}_k|$ such inner-products.  Running $\sr_k$
	takes $t_k$ steps, by definition.

	Line~\ref{line:n} takes $\O(n)$ steps, as the $k$-diagonal has at most this many entries.

	Line~\ref{line:mue} takes $\O(rn)$ steps by Claim~\ref{clm:makeue}.

	Lines~\ref{line:upp} takes $\O(rn)$ steps, for as used above, $L_k-I_n$ has only non-zero
	entries with columns in $\lne_R(P^{(<k)})$, so each entry in $(L_kP)^{(k)}$ is the sum of at
	most $r+1$ products of entries in $L_k$ and $P$, and these products are determined by
	$\lne_R(P^{(<k)})$.  As there are at most $n$ such entries, the bound follows.

	Line~\ref{line:upl} takes $\O(rn)$ steps.  This is because $L_k$, by Claim~\ref{clm:makeue},
	is the product of $\le r$ elementary matrices, and left-multiplication by an elementary
	matrix takes $\O(n)$ steps.  As \textsc{MakeUpperEchelon} computes $L_k$ as a product of
	elementary matrices, the computation of $L_kL$ can also use this decomposition and thus is
	compute in $\O(rn)$ steps.

	Thus, the entire loop runs in $\O(rn+t_k+n|\mathcal{R}_k|)$ steps, and there are at most
	$n+m$ iterations of the loop, giving the bound.
\end{proof}

We now apply this reduction to our hitting set $\mathcal{D}_{2r,n,m}'$, which embeds the
sparse-recovery measurements corresponding to the dual Reed-Solomon code.

\begin{corollary}\label{cor:diaglrr}
	Let $1\le r\le n/2$, $m\ge n\ge 1$. Then $\mathcal{D}_{2r,n,m}'$ (from
	Construction~\ref{constr:diag}) has
	\begin{enumerate}
		\item $|\mathcal{D}_{2r,n,m}'|=2(n+m-2r)r$ \label{cor:diaglrr:size}
		\item Each matrix in $\mathcal{D}_{2r,n,m}'$ is
		$n$-sparse.\label{cor:diaglrr:sparse}
		\item $\mathcal{D}_{2r,n,m}'$ is a $r$-low-rank-recovery set\label{cor:diaglrr:lrr}
		\item Algorithm~\ref{alg:lrr}, combined with Algorithm~\ref{alg:prony}, performs
		low-rank-recovery for $\mathcal{D}_{2r,n,m}'$ in time
		$\O(rnm+(n+m)r^3)$\label{cor:diaglrr:time}
	\end{enumerate}
\end{corollary}
\begin{proof}
	\underline{\eqref{cor:diaglrr:size}:}  This is by construction.

	\underline{\eqref{cor:diaglrr:sparse}:} Each matrix in $\mathcal{D}_{2r,n,m}'$ has its
	support contained in some $k$-diagonal, and each $k$-diagonal has at most $n$ elements.

	\underline{\eqref{cor:diaglrr:lrr}:} We will first show that the measurements that
	$\mathcal{D}_{2r,n,m}'$ performs on each $k$-diagonal comprise a
	$\min(2r,k+1,(n+m)-(k+1))$-advice-sparse-recovery set.

	First consider the case when $k+1<2r\le n$.  Then $\min(2r,k+1,(n+m)-(k+1))=k+1$, and
	$\mathcal{D}_{2r,n,m}'$ places $k+1$ constraints on this $k$-diagonal $M^{(k)}$, which has
	$k+1$ entries. The constraint matrix $V$ is of size $(k+1)\times (k+1)$ with
	$V_{\ell,j}=g^{\ell j}$.  As $g$ has order $\ge n$, the elements $1,g,\ldots,g^k$ are
	distinct.  So these constraints form an invertible Vandermonde system and so $M^{(k)}$
	(regardless of the rank of $M$) can be completely recovered from these measurements.  In
	particular, $V$ forms a $(k+1)$-advice-sparse-recovery set.  As the Vandermonde system can
	be inverted in $\O(k^3)=\O(r^3)$ time, we see that $(k+1)$-advice-sparse-recovery can be
	performed in this time.

	Similarly, now consider the case when $(n+m)-(k+1)<2r\le n$ (so it follows that $m\le k$).
	Then $\min(2r,k+1,(n+m)-(k+1))=(n+m)-(k+1)$, and $\mathcal{D}_{2r,n,m}'$ places
	$(n+m)-(k+1)$ constraints on this $k$-diagonal $M^{(k)}$, which has $(n+m)-(k+1)$ entries.
	The constraint matrix $V$ is of size $((n+m)-(k+1))\times ((n+m)-(k+1))$ with
	$V_{\ell,j}=g^{\ell (k-(m-1)+j)}$.  As $g$ has order $\ge n$, the elements
	$g^{k-(m-1)},g^{k-(m-1)+1},\ldots,g^{n-1}$ are distinct.  So these constraints form an
	invertible Vandermonde system and so $M^{(k)}$ (regardless of the rank of $M$) can be
	completely recovered from these measurements.  In particular, $V$ forms a
	$((n+m)-(k+1))$-advice-sparse-recovery set.  As the Vandermonde system can be inverted in
	$\O(((n+m)-(k+1))^3)=\O(r^3)$ time, we see that $((n+m)-(k+1))$-advice-sparse-recovery can
	be performed in this time.

	Now consider the general case when $2r\le k+1,(n+m)-(k+1)$.  Then
	$\min(2r,k+1,(n+m)-(k+1))=2r$, and $\mathcal{D}_{2r,n,m}'$ places $2r$ constraints on this
	$k$-diagonal $M^{(k)}$, which has $\min(k+1,n,(n+m)-(k+1))$ entries. The constraint matrix
	$V$ is of size $2r\times \min(k+1,n,(n+m)-(k+1))$ with $V_{\ell,j}=g^{\ell
	(\max(0,k-(m-1))+j)}$.  As $g$ has order $\ge n$, the elements
	\[g^{\max(0,k-(m-1))},g^{\max(0,k-(m-1))+1},\ldots,
	g^{\max(0,k-(m-1))+\min(k+1,n,(n+m)-(k+1))-1}\] are distinct.  Thus, it follows from
	Theorem~\ref{thm:prony} that $V$ is a $r$-advice-sparse-recovery set, and that recovery can
	be done in $\O(r^3+n)$ steps.

	Thus, by Theorem~\ref{thm:lrr to sparse}, it follows that $\mathcal{D}_{2r,n,m}'$ is a
	$r$-low-rank-recovery set.

	\underline{\eqref{cor:diaglrr:time}:}  By the analysis done for \eqref{cor:diaglrr:lrr}, we
	see that Theorem~\ref{thm:lrr to sparse} shows that Algorithm~\ref{alg:lrr} (along with the
	$r$-advice-sparse-recovery performed by Algorithm~\ref{alg:prony}) yields a
	$\O(rnm+(n+m)r^3)$-time recovery algorithm for $\mathcal{D}_{2r,n,m}'$.
\end{proof}

\begin{remark}
	We briefly note that for $r>n/2$ we have that $|\mathcal{D}_{2r,n,m}'|\ge nm$ (one cannot
	use the formula ``$|\mathcal{D}_{2r,n,m}'|=2(n+m-2r)r$'' here, but the bound
	$|\mathcal{D}_{2r,n,m}'|\le|\mathcal{D}_{2r,n,m}|=2(n+m-1)r$ is still valid).  Thus, for
	$r>n/2$ there is no gain from using $\mathcal{D}_{2r,n,m}'$ over the obvious $nm$
	low-rank-recovery set that queries each entry in the matrix.
\end{remark}

\begin{remark}\label{rmk:reprovediag}
	One can also use Algorithm~\ref{alg:lrr} to reprove Theorem~\ref{thm:diagonal set for
	matrices}, that is, to reprove that $\mathcal{D}_{r,n,m}$ is a hitting set (note that we use
	$r$ and not $2r$ here).  To do so, note that Lemma~\ref{lem:uesparse} shows that for a rank
	$\le r$ matrix $M$, if $M^{(<k)}=0$ then $M^{(k)}$ is $r$-sparse.

	Thus, if $\la M,\mathcal{D}_{r,n,m}\ra=\vec{0}$ then this implies that for each $k$, $\la
	M^{(k)},\mathcal{R}_k\ra=\vec{0}$, where $\mathcal{R}_k$ is the $r$-sparse-recovery set
	formed from the dual Reed-Solomon code.  So if $M^{(k)}$ is $r$-sparse then by the
	properties of $\mathcal{R}_k$ it must be that $M^{(k)}=\vec{0}$.

	Combining the two observations above, we see that $M^{(<k)}=0\implies M^{(k)}=\vec{0}$, and
	thus $M^{(<k)}=0\implies M^{(\le k)}=\vec{0}$.  Inducting on $k$ shows that $M=0_{n\times
	m}$.  Thus, if $M\ne 0$ and $M$ is rank $\le r$ then $\la M,\mathcal{D}_{r,n,m}\ra\ne
	\vec{0}$, showing that $\mathcal{D}_{r,n,m}$ is a hitting set.
\end{remark}

Given that $\mathcal{D}_{2r,n,m}'$ admits efficient low-rank-recovery, we can recall the above
results that show that these measurements are equivalent to the $\mathcal{B}_{2r,n,m}'$
measurements.  Thus, we also get that this second set admits efficient low-rank-recovery.

\begin{corollary}\label{cor:rank1lrr}
	Let $1\le r\le n/2$, $m\ge n\ge 1$. Then $\mathcal{B}_{2r,n,m}'$
	(from Construction~\ref{constr:better hitting for matrices}) has
	\begin{enumerate}
		\item $|\mathcal{B}_{2r,n,m}'|=2(n+m-2r)r$ \label{cor:rank1lrr:size}
		\item Each matrix in $\mathcal{B}_{2r,n,m}'$ is rank 1.\label{cor:rank1lrr:rank1}
		\item $\mathcal{B}_{2r,n,m}'$ is a $r$-low-rank-recovery set\label{cor:rank1lrr:lrr}
		\item Algorithm~\ref{alg:lrr}, combined with Algorithm~\ref{alg:prony}, performs
		low-rank-recovery for $\mathcal{B}_{2r,n,m}'$ in time
		$\O(rm^2+mr^3)$\label{cor:rank1lrr:time}
	\end{enumerate}
\end{corollary}
\begin{proof}
	\underline{\eqref{cor:rank1lrr:size}:}  This is by construction.

	\underline{\eqref{cor:rank1lrr:rank1}:} This is also by construction.

	\underline{\eqref{cor:rank1lrr:lrr}:} By Theorem~\ref{thm:better hitting for matrices} and
	Theorem~\ref{thm:diagonal set for matrices} we see have that
	$\sspan\mathcal{D}_{r,n,m}'=\sspan\mathcal{B}_{r,n,m}'$.  In particular, the measurements
	$\la M,\mathcal{D}_{r,n,m}'\ra$ can be reconstructed from the measurements $\la
	M,\sspan\mathcal{B}_{r,n,m}'\ra$.  As the above corollary shows that $\mathcal{D}_{r,n,m}'$
	is $r$-low-rank-recovery set, it follows that $\mathcal{B}_{r,n,m}'$ is also.

	\underline{\eqref{cor:rank1lrr:time}:}  The analysis given in Theorem~\ref{thm:better
	hitting for matrices} gives an algorithm for reconstructing the measurements $\la
	M,\mathcal{D}_{r,n,m}'\ra$ from the measurements $\la M,\sspan\mathcal{B}_{r,n,m}'\ra$, and
	does so interpolating $r$ polynomials of degree $\le n+m$. As evaluations of these
	polynomials takes $\O(r)$ steps, and polynomial interpolation takes $\O(m^2)$ steps for
	polynomials of this degree, we see that we can complete this interpolation in
	$\O(rm^2+r^2m)=\O(rm^2)$ steps. Once the measurements $\la M,\mathcal{D}_{r,n,m}'\ra$ are
	computed, we can appeal to the above corollary.
\end{proof}

The above results only work over fields when we have an element $g$ of large order.  However, the
results of Subsection~\ref{sec:pit for tensors over small fields} show that we can simulate these
results over small fields.  Indeed, this is also the case here.

\begin{corollary}\label{cor:smalllrr}
	Let $m\ge n\ge r\ge 1$.  Over any field $\F$, there is an $\poly(m)$-explicit
	$r$-low-rank-recovery set for $n\times m$ matrices, which has size $\O(rm\lg m)$ and is such
	that each recovery matrix is $\O(n)$-sparse. There is also an $\poly(m)$-explicit
	$r$-low-rank-recovery set for $n\times m$ matrices, which has size $\O(rm\lg^2 m)$ and is
	such that each recovery matrix is rank 1.  Further, recovery from either of these
	low-rank-recovery sets can be performed in $\poly(m)$ time.
\end{corollary}
\begin{proof}
	We begin by noting that both Proposition~\ref{prop:simfieldimproper} and
	Proposition~\ref{prop:simfield} preserve the property of being a low-rank-recovery set, not
	just that of being a hitting set. That is, each of these propositions take a $\K$-matrix $H$
	in the original low-rank-recovery set and construct some family of $\F$-matrices
	$\{\tilde{H}_{\ell,\ell'}\}_{\ell,\ell'}$ such that for any matrix $M$, $\la M,H\ra$ can be
	efficiently recovered from the sums $\{\sum_\ell \alpha_\ell\la
	M,\tilde{H}_\ell\ra\}_{\ell'}$, for some coefficients $\alpha_{\ell}\in\K$.  Thus the
	measurements $\la M,\mathcal{H}\ra$ are efficiently recoverable from the measurements $\la
	M,\mathcal{H}\ra$.

	Finally, appealing to the constructions of low-rank-recovery sets as given in
	Corollary~\ref{cor:diaglrr} (to which Proposition~\ref{prop:simfieldimproper} is applied)
	and Corollary~\ref{cor:rank1lrr} (to which Proposition~\ref{prop:simfield} is applied)
	completes the claim.
\end{proof}

\section{Rank-Metric Tensor codes}\label{sec:tensor codes}

We now discuss low-rank-recovery of tensors, for any $d$, and apply our results to the construction
of rank-metric codes.  We begin with showing that the matrix low-rank-recovery algorithm can be
extended to the $d>2$ case.

\begin{theorem}\label{thm:tensorlrr}
	Let $n,r\ge 1$ and $d\ge 2$.  Then $\mathcal{B}_{d,n,2r}$, as defined in
	Construction~\ref{construction:construction of hitting set for tensors}, has
	\begin{enumerate}
		\item $|\mathcal{B}_{d,n,2r}|\le\O(dn(2r)^{\O(\lg d)})$\label{thm:tensorlrr:size}
		\item \sloppy $\mathcal{B}_{d,n,2r}$ is an $r$-low-rank-recovery set, and recovery
		can be performed in time $\poly((2dn)^d,(2r)^{\O(\lg d)})$\label{thm:tensorlrr:lrr}
	\end{enumerate}
\end{theorem}
\begin{proof}
	\underline{\eqref{thm:tensorlrr:size}:} This is by construction.

	\underline{\eqref{thm:tensorlrr:lrr}:} The hitting set allows us to interpolate the
	polynomials stated in the hypothesis of Theorem~\ref{thm:d-variate reduction}.  Once we have
	the coefficients of this polynomial, we can undo the reductions used in the proof of
	Theorem~\ref{thm:d-variate reduction}.  That is, that proof uses Lemmas~\ref{lem:permute}
	and \ref{lem:reshape} to reshape polynomials by merging their variables.  This is clearly
	efficiently reversible.  More crucially, the proof uses the bivariate variable reduction of
	Theorem~\ref{thm:bivariate poly evaluation} for rank $\le r$ matrices, but when we take $2r$
	distinct powers of $g$.  However, Corollary~\ref{cor:diaglrr} shows that one can recover
	$\hat{f}_M(x,y)$ from the polynomials $\{\hat{f}_M(x,g^ix)\}_{i\in\llb 2r\rrb}$ in
	$\poly(\deg_x(\hat{f}_M),\deg_y(\hat{f}_M),r)$ steps.  As the degrees involved in
	Theorem~\ref{thm:d-variate reduction} are only up to $(2dn)^d$, this is within the stated
	time bounds. Thus, we can also reverse the bivariate variable reduction steps used in
	Theorem~\ref{thm:d-variate reduction}.  Combining these steps shows that we can fully
	recover the entire polynomial $\hat{f}_T(x_1,\ldots,x_d)$, which gives the tensor $T$.
\end{proof}

We next observe that, just as with Corollary~\ref{cor:smalllrr}, we can perform this
low-rank-recovery over small fields, when incurring a loss.

\begin{corollary}\label{cor:smalllrrtensor}
	Let $n,r\ge 1$ and $d\ge 2$.  Over any field $\F$, there is an $\poly((2nd)^d,r^{\O(\lg
	d)})$-explicit $r$-low-rank-recovery set for $\llb n\rrb^d$ tensors, which has size
	$\O(dn(2r)^{\O(\lg d)}\cdot(d\lg2dn)^d)$ and is such that each recovery tensor is rank 1.
	Further, there is an $\poly((2nd)^d,r^{\O(\lg d)})$-explicit $r$-low-rank-recovery set for
	$\llb n\rrb^d$ tensors, which has size $\O(dn(2r)^{\O(\lg d)}\cdot d\lg2dn)$.  Further,
	recovery from either of these low-rank-recovery sets can be performed in
	$\poly((2nd)^d,r^{\O(\lg d)})$ time.
\end{corollary}
\begin{proof}
	Like Corollary~\ref{cor:smalllrr}, we apply Propositions~\ref{prop:simfield} and
	\ref{prop:simfieldimproper} to a low-rank-recovery set, where here we use the above set from
	Theorem~\ref{thm:tensorlrr}.  As Propositions~\ref{prop:simfield} and
	\ref{prop:simfieldimproper}, as well as Theorem~\ref{thm:tensorlrr}, are efficiently
	implementable, so are the resulting low-rank-recovery sets.
\end{proof}

We now apply these results to create error correcting codes over the rank-metric, which we now
define.  We will restrict our attention to linear codes in this work.

\begin{definition}
	A $[\llb n\rrb^d,k,r]_\F$ \textbf{rank-metric code} $\C$ is a $k$-dimensional subspace of
	$\F^{\llb n\rrb^d}$ (the space of $\llb n\rrb^d$ tensors) such that for all $T_1\ne
	T_2\in\mathcal{C}$, $\rank(T_1-T_2)\ge r$.  Denote $r$ as the \textbf{distance} of the code.

	An algorithm \dec\ \textbf{corrects $e$ errors} against $\C$ if for any $T\in\C$ and
	$E\in\F^{\llb n\rrb^d}$ with $\rank(E)\le e$ it is such that $\dec(T+E)=T$.
\end{definition}

Thus this is the natural definition for error-correcting codes when we use the rank-metric (notice
that rank-distance is in fact a metric) as the notion of distance.  As we are interested in linear
codes $T_1-T_2\in \C$ also, so an equivalent definition to the above would say that $r\le\rank(T)$
for all $0\ne T\in\C$.  Just as with the Hamming-metric, if we have a distance $2r+1$ code $\C$ then
it is information theoretically possible to decode up to $r$ errors.  The converse is shown below.

\begin{lemma}\label{lem:mindist}
	Let $\C$ be a $[\llb n\rrb^d,k,r']_\F$ rank-metric code that can correct up to $r$ errors.
	Then $r'\ge 2r+1$.
\end{lemma}
\begin{proof}
	Suppose not for contradiction.  Then there are two tensors $T_1\ne T_2\in \C$ such that
	$\rank(T_2-T_1)\le 2r$.  But then $T_2-T_1=S_1+\cdots+S_{2r}$, where these $S_i$ are all
	rank-1 (or rank-0) tensors.  Then it follows that $T_1+S_1+\cdots+S_{r}$ is $r$-close to
	both $T_1$ and $T_2$, which is impossible as the correctness of the decoding procedure
	indicates that there should be a unique tensor that $T_1+S_1+\cdots+S_{r}$ is $r$-close to.
\end{proof}

\begin{corollary}\label{cor:rankmetricmatrices}
	Let $\F$ be a field, $m\ge n\ge r\ge 1$.  Then there are $\poly(m)$-explicit rank-metric
	codes with $\poly(m)$-time decoding for up to $r$ errors, with parameters:
	\begin{enumerate}
		\item $[\llb n\rrb\times \llb m\rrb,nm-2(n+m-2r)r,2r+1]_\F$, if $|\F|>m$, and the
		parity checks on this code can be either all rank-1 matrices, or all $\O(n)$-sparse
		matrices.
		\item $[\llb n\rrb\times \llb m\rrb,nm-2(n+m-2r)r\cdot \O(\lg m),2r+1]_\F$, any
		$\F$, and the parity checks on this code are all $\O(n)$-sparse matrices.
		\item $[\llb n\rrb\times \llb m\rrb,nm-2(n+m-2r)r\cdot \O(\lg^2 m),2r+1]_\F$, any
		$\F$, and the parity checks on this code are all rank-1 matrices.
	\end{enumerate}
\end{corollary}
\begin{proof}
	We first generically show how to define an $[nm,nm-|\mathcal{H}|,2r+1]_\F$ rank-metric code
	$\C$ from an $r$-low-rank-recovery set $\mathcal{H}$ and how to use the low-rank-recovery
	algorithm for $\mathcal{H}$ to decode $\C$ up to $r$ errors.  The corollary is then
	immediate by using the results of Corollaries~\ref{cor:rank1lrr}, \ref{cor:diaglrr},
	\ref{cor:smalllrr}, and invoking the efficiency of their low-rank-recovery.

	Define $\C$ to be the matrices in the nullspace of $\mathcal{H}$.  That is, $\C=\{M:\la
	M,\mathcal{C}\ra =0\}$.  It is clear that $\C$ is a subspace (and assuming that the matrices
	in $\mathcal{H}$ are linearly independent, which is true for the low-rank-recovery sets
	$\mathcal{D}_{2r,n,m}'$ and $\mathcal{B}_{2r,n,m}'$) and has dimension $nm-|\mathcal{H}|$.

	Now consider some $T\in\C$ and matrix $E$ with $\rank(E)\le r$.  Abusing notation, consider
	$T$ and $E$ as $nm$-long vectors, and $\mathcal{H}$ as a $|\mathcal{H}|\times nm$ matrix.
	It follows that $\mathcal{H}(T+E)=\mathcal{H}E$ as $T\in\C$.  As $\mathcal{H}$ is an
	$r$-low-rank-recovery set, it follows that we can recover $E$ from $\mathcal{H}E$, and thus
	can recover $T$, performing successful decoding of up to $r$ errors.  By
	Lemma~\ref{lem:mindist} we see that the minimum distance of this code is $\ge 2r+1$.
\end{proof}

We now separately state the result for tensors, which is proved exactly as the above corollary, but
using the relevant low-rank-recovery results for tensors.

\begin{corollary}\label{cor:rankmetrictensors}
	Let $\F$ be a field, $n,r\ge 1$ and $d\ge 2$.  Then there are $\poly((2nd)^d,(2r)^{\O(\lg
	d)})$-explicit rank-metric codes with $\poly((2nd)^d,(2r)^{\O(\lg d)})$-time decoding for up
	to $r$ errors, with parameters:
	\begin{enumerate}
		\item $[\llb n\rrb^d,n^d-dn(2r)^{\lceil \lg d\rceil},2r+1]_\F$, if $|\F|>(2nd)^d$,
		and the parity checks on this code are all rank-1 tensors,
		\item $[\llb n\rrb^d,n^d-dnr^{\lceil \lg d\rceil}\cdot\O(d\lg(2dn)),2r+1]_\F$, any
		$\F$,
		\item $[\llb n\rrb^d,n^d-dnr^{\lceil \lg d\rceil}\cdot\O((d\lg(2dn))^d),2r+1]_\F$,
		any $\F$, and the parity checks on this code are all rank-1 tensors.
	\end{enumerate}
\end{corollary}

\section{Discussion}\label{sec:discuss}

We briefly discuss some directions for further research.

\paragraph{Reducing Noisy Low-Rank Recovery to Noisy Sparse Recovery} We showed in
Theorem~\ref{thm:lrr to sparse} that low-rank-recovery of matrices can be done using any
sparse-recovery oracle.  This reduction was for non-adaptive measurements, and was done in the
presence of no noise.  As much of the compressed sensing community is interested in the noisy case
(so $M$ is only close to rank $\le r$) the main open question of this work is whether the reduction
extends to the noisy case.

\paragraph{Smaller Hitting Sets} While the observations of Roth~\cite{Roth91} show that our hitting
set for matrices is optimal over algebraically closed fields, our results
(Corollary~\ref{cor:hitsmalltensor}) over tensors with $d>2$ are much larger than the existential
bounds of Lemma~\ref{lem:szhit}.  Can these hitting sets be improved to size $\O(\poly(d)nr^k)$ for
$k=\O(1)$?  As mentioned in the preliminaries (Lemma~\ref{lem:hittolbs}), any such hitting set with
$k<2$ would yield improved tensor-rank lower bounds (and thus circuit lower bounds) for odd $d$ such
as $d=3$.  However, as the best tensor-rank lower bounds for $d=3$ are $\Theta(n)$ and our hitting
set (over infinite fields) yields this bound (with a smaller constant), even improving our hitting
set for $d=3$ by constant factors could yield interesting new results.  Specifically, for $d=3$ can
one construct (say over infinite fields) a hitting set of size $\le nr^2/10$ for $\llb n\rrb^3$
tensors of rank $\le r$?

\paragraph{Better Variable Reduction} Theorem~\ref{thm:bivariate poly evaluation} shows that a
bivariate polynomial with bounded individual degrees can be identity tested by identity testing a
collection of univariate polynomials, where the size of this collection depends on the rank of
bivariate polynomial.  This naturally led to our hitting sets for matrices.  We generalized this to
$d$-variate polynomials in Theorem~\ref{thm:d-variate reduction}, but the collection of univariate
polynomials has a size with a much worse dependence on the tensor-rank of the $d$-variate polynomial
and is much less explicit.  Can the size of the collection be reduced, or can the explicitness of
this set be only polynomially larger than its size?  We note that according to
Lemma~\ref{lem:hittolbs} a more explicit hitting set will yield lower bounds on tensor rank, however
for tensors of high degrees such lower bounds are known \cite{NisanWigderson96}.

\paragraph{Large Field Simulation} The results of Section~\ref{sec:pit for tensors over small
fields} show that hitting sets (and LRR sets) that involve tensors over an extension
field imply hitting sets (and low-rank recovery sets) over the base field.  While
Proposition~\ref{prop:simfield} shows that we can preserve the rank-1 property of these tensors
while doing so, it introduces an $\exp(d)$ factor in the size of the hitting set.  Can this be
improved?

\section*{Acknowledgements}

We would like to thank Olgica Milenkovic for pointing us to the
low-rank recovery problem, and Madhu Sudan for some helpful
comments regarding decoding dual Reed-Solomon codes.

Part of this work was done while the first author was visiting Stanford University, as well as when
the second author was visiting the Bernoulli center at EPFL.

\newpage

\bibliographystyle{alpha}
\bibliography{bibliography}

\newpage

\appendix

\section{Cauchy-Binet Formula}\label{app: missing proofs}

For completeness we give the proof of the Cauchy-Binet formula here.

\begin{lemma}[Cauchy-Binet Formula]\label{lem: cauchy-binet}
	Let $m\ge n\ge 1$. Let $A\in\F^{n\times m}$, $B\in\F^{m\times n}$.  For $S\subseteq\llb
	m\rrb$, let $A_S$ be the $n\times |S|$ matrix formed from $A$ by taking the columns with
	indices in $S$. Let $B_S$ be defined analogously, but with rows.  Then
	\[\det(AB)=\sum_{S\in\binom{\llb m\rrb}{n}}\det(A_S)\det(B_S)\]
\end{lemma}
\begin{proof}
	Let $C$ be an $m\times m$ diagonal matrix with the variables $x_1,\ldots,x_m$ on the
	diagonal. Define the polynomial $f(x_1,\ldots,x_m)\eqdef\det(ACB)$, so that
	$f(1,\ldots,1)=\det(AB)$. Every entry of $ACB$ is a homogeneous linear function in
	$x_1,\ldots,x_m$, which implies (as the determinant is homogeneous of degree $n$) that $f$
	is homogeneous of degree $n$, or zero. Let $S\in\binom{\llb m\rrb}{n}$ and consider all
	monomials only containing variables in $\{x_i\mid i\in S\}$. Note that also consider
	monomials with individual degrees above 1.  Each monomial of degree $n$ (and thus each
	monomial with non-zero coefficient in $f$) must be associated with some such $S$.

	Define $\rho_S$ to be the vector of variables when the substitution $x_i\mapsto 0$ is
	performed for $i\notin S$.  It follows then that $f(\rho_S)=\det(A_S C_S B_S) =
	\det(A_S)\det(B_S)\cdot \prod_{i\in S}x_i$, where the last equality follows as $A_S,B_S$ and
	$C_S$ are all $n\times n$ matrices. By the above reasoning, this implies that the only
	monomials with non-zero coefficients in $f$ are monomials of the form $\prod_{i\in S}x_i$
	and such monomials have coefficient $\det(A_S)\det(B_S)$.  Thus $f= \sum_{S\in\binom{\llb
	m\rrb}{n}}\det(A_S)\det(B_S)\prod_{i\in S}x_i$, and so
	$\det(AB)=f(1,\ldots,1)=\sum_{S\in\binom{\llb m\rrb}{n}}\det(A_S)\det(B_S)$, yielding the
	claim.
\end{proof}

\end{document}